\batchmode
\makeatletter
\def\input@path{{D:/Documents/MATLAB/Beams/FBMR/M_S_Paper_LaTex_Template_2017/Arxiv/}}
\makeatother
\documentclass[10pt,oneside]{amsart}
\usepackage[T1]{fontenc}
\usepackage[latin9]{inputenc}
\usepackage{array}
\usepackage{verbatim}
\usepackage{float}
\usepackage{mathtools}
\usepackage{amsthm}
\usepackage{amssymb}
\usepackage{graphicx}

\makeatletter

\providecommand{\tabularnewline}{\\}
\floatstyle{ruled}
\newfloat{algorithm}{tbp}{loa}
\providecommand{\algorithmname}{Algorithm}
\floatname{algorithm}{\protect\algorithmname}

\numberwithin{equation}{section}
\numberwithin{figure}{section}
\theoremstyle{plain}
\newtheorem{thm}{\protect\theoremname}
  \theoremstyle{definition}
  \newtheorem{example}[thm]{\protect\examplename}
  \theoremstyle{plain}
  \newtheorem{lem}[thm]{\protect\lemmaname}
  \theoremstyle{plain}
  \newtheorem{cor}[thm]{\protect\corollaryname}
  \theoremstyle{plain}
  \newtheorem{prop}[thm]{\protect\propositionname}

\pdfoutput=1
\usepackage{multirow}
\usepackage{amsmath}
\usepackage{amsthm}
\usepackage{amssymb}
\usepackage{graphicx}
\usepackage{alltt}
\usepackage{color}
\definecolor{string}{rgb}{0.7,0.0,0.0}
 \definecolor{comment}{rgb}{0.13,0.54,0.13}
 \definecolor{keyword}{rgb}{0.0,0.0,1.0}

  \providecommand{\corollaryname}{Corollary}
  \providecommand{\examplename}{Example}
  \providecommand{\lemmaname}{Lemma}
  \providecommand{\propositionname}{Proposition}
\providecommand{\theoremname}{Theorem}

\makeatother

  \providecommand{\corollaryname}{Corollary}
  \providecommand{\examplename}{Example}
  \providecommand{\lemmaname}{Lemma}
  \providecommand{\propositionname}{Proposition}
\providecommand{\theoremname}{Theorem}

\begin{document}

\title{Sampling for approximating $R$-limited functions}

\author{\textbf{Can Evren Yarman}}

\address{Schlumberger, High Cross, Madingley Road, Cambridge,CB3 0EL, United
Kingdom (cyarman@slb.com)}
\begin{abstract}
$R$-limited functions are multivariate generalization of band-limited
functions whose Fourier transforms are supported within a compact
region $R\subset\mathbb{R}^{n}$. In this work, we generalize sampling
and interpolation theorems for band-limited functions to $R$-limited
functions. More precisely, we investigated the following question:
``For a function compactly supported within a region similar to $R$,
does there exist an $R$-limited function that agrees with the function
over its support for a desired accuracy?''. Starting with the Fourier
domain definition of an $R$-limited function, we write the equivalent
convolution and a discrete Fourier transform representations for $R$-limited
functions through approximation of the convolution kernel using a
discrete subset of Fourier basis. The accuracy of the approximation
of the convolution kernel determines the accuracy of the discrete
Fourier representation. Construction of the discretization can be
achieved using the tools from approximation theory as demonstrated
in the appendices. The main contribution of this work is proving the
equivalence between the discretization of the Fourier and convolution
representations of $R$-limited functions. Here discrete convolution
representation is restricted to shifts over a compactly supported
region similar to $R$. We show that discrete shifts for the convolution
representation are equivalent to the spectral parameters used in discretization
of the Fourier representation of the convolution kernel. This result
is a generalization of the cardinal theorem of interpolation of band-limited
functions. The error corresponding to discrete convolution representation
is also bounded by the approximation of the convolution kernel using
discretized Fourier basis. 
\end{abstract}

\maketitle

\section{Introduction}

$R$-limited functions are functions whose Fourier transforms are
supported within a region $R\subset\mathbb{R}^{n}$. They are multivariate
generalization of band-limited functions. The terminology was coined
by Slepian in \cite{slepian1964prolate}. In this work, we generalize
sampling and interpolation theorems for band-limited functions to
$R$-limited functions. Specifically, we explore answers to the following
questions: ``For a function compactly supported within a region similar
to $R$, does there exist an $R$-limited function that agrees with
the function over its support? If so, how shall we sample the function
to construct an $R$-limited function that approximates the original
function within a desired accuracy?''. The first question has been
answered in \cite{slepian1964prolate}. Combining these results with
methods from approximation theory, we answer the second question.
Answering how to sample also provides a guide for where to sample.

In our exposition we choose an approximation theory perspective which
provides an alternative insight to understanding of band-limited functions
through discretization of the sine cardinal function as well as a
deterministic framework for constructing sampling schemes for $R$-limited
functions. Starting with the Fourier domain definition of a $R$-limited
function, we write the equivalent convolution representation and write
a discrete Fourier transform representation for $R$-limited functions
through approximation of the convolution kernel using a discrete subset
of Fourier basis. The accuracy of the approximation of the kernel
determines the accuracy of the discrete Fourier representation. Construction
of the discretization can be achieved using the tools from approximation
theory such as generalization of Pad{é} approximation, which is
summarized in Appendix \ref{sec:Generalization-of-Pad=00003D0000E9}.

Our main results, Theorems \ref{thm:discrete-Fourier-approximation-R-limited}
and \ref{thm:sampling-theorem-R_A-limited-functions}, prove the equivalence
between the discretization of the Fourier and convolution representations
that approximate a compactly supported function within a region similar
to $R$. We show that discrete shifts for the convolution representation
are equivalent to the spectral parameters used in discretization of
the Fourier representation of the convolution kernel. Discretization
of the convolution representation is also referred to as sampling
and interpolation theorem.

Theorem \ref{thm:sampling-theorem-R_A-limited-functions} is a generalization
of the sampling and interpolation theorem for band-limited functions
summarized in Theorem \ref{thm:sampling-theorem-bandlimited-functions}.
It also provides a way to analyze and approximate the resulting error.
We show that the error corresponding to discrete convolution representation
is bounded by error obtained from discretization of the Fourier transform
of the convolution kernel. In single dimension, it provides a new
way to prove truncation of discrete Fourier series as well as sinc
interpolation formula. Furthermore, in single dimension, our result
indicates that for a support of interest, instead of uniform sampling,
improvement in discrete the Fourier representation of band-limited
functions can be obtained using Gauss-Legendre type quadratures (see
Appendix \ref{sec:Generalization-of-Pad=00003D0000E9}). This also
raises the questions on what is the most cost efficient way to implement
fast Fourier transforms using Gaussian quadratures which may be addressed
using the ideas from \cite{dutt1993fast,duijndam1999nonuniform,edelman1998future,greengard2004accelerating}
and left for a future discussion. Similar discretizations are obtained
for special cases of $R$-limited functions in Appendix \ref{sec:Kernel-for-T-bandlimited-functions}
and \ref{sec:Cone-limited-functions} where we make use of cascaded
quadratures that are equivalent to Gauss-Legendre or Clenshaw-Curtis
quadratures. While the body of the manuscript contains our main results,
the Appendices also provide as valuable information by providing a
constructive way for computing quadratures to discretize convolution
kernels which can be utilized in the sampling and interpolation theorems.

The outline of the manuscript is as follows. In Section \ref{sec:Conventions},
we present the conventions used in the rest of our discussion. To
motivate the multivariate case, in Section \ref{sec:Bandlimited-functions},
we study discretization of Fourier transform and sinc interpolation
formula for one-dimensional (univariate) band-limited functions, or
band-limited projection of compactly supported functions. Both discrete
Fourier transform and sinc interpolation formulas have been studied
in the literature with many books devoted to this topic. We refer
the reader to \cite{jerri1977shannon,unser2000sampling,stankovic2006some,jerri2013gibbs}
for a comprehensive list of references on these topics. We give an
alternative exposition, which leads to proof of the equivalence of
sampling in the domain of the function (Theorem \ref{thm:sampling-theorem-bandlimited-functions})
and its Fourier transform (Theorem \ref{thm:discrete_Fourier_approx_bandlimited_functions}).
The necessary background material for Section \ref{sec:Bandlimited-functions}
is provided in Appendices \ref{sec:Generalization-of-Pad=00003D0000E9}
and \ref{sec:On-approximations-of-sinc} which discuss Generalization
of Pade approximation and approximations to sine cardinal function.
Compared to band-limited function, sampling and representation of
multivariate functions whose Fourier transforms' support are not similar
to a hypercube is studied and understood less. In Section \ref{sec:R-limited-functions},
we extend our results for band-limited function (Theorems \ref{thm:discrete_Fourier_approx_bandlimited_functions}
and \ref{thm:sampling-theorem-bandlimited-functions}) to R-limited
functions (Theorems \ref{thm:discrete-Fourier-approximation-R-limited}
and \ref{thm:sampling-theorem-R_A-limited-functions}). Examples of
special cases of convolution kernels for $R$-limited functions are
presented in Appendices \ref{sec:Kernel-for-T-bandlimited-functions}
and \ref{sec:Cone-limited-functions}. In Appendix \ref{sec:Kernel-for-T-bandlimited-functions},
we provide a method to construct quadratures for isosceles triangle
and trirectangular tetrahedron which are used to construct quadratures
for equilateral triangle and regular tetrahedron. Similar method is
used in Appendix \ref{sec:Cone-limited-functions} to construct quadratures
for a finite cone and a ball in three dimensions which have practical
importance in multidimensional signal processing seismic data, image
processing and video processing.

\section{Conventions\label{sec:Conventions}}

We employ the following conventions of Fourier transform, inverse
Fourier transform and convolution.

The \emph{Fourier transform}\textbf{\emph{ }}$\mathcal{F}\left[f\right]\left(\mathbf{k}\right)$
of $f\left(\mathbf{x}\right)$, an absolutely integrable function
for $\mathbf{x},\,\mathbf{k}\in\mathbb{R}^{N}$, which we denote by
$\hat{f}\left(\mathbf{k}\right)$, is defined by 
\begin{align}
\mathcal{F}\left[f\right]\left(\mathbf{k}\right)=\hat{f}\left(\mathbf{k}\right) & =\int_{\mathbb{R}^{N}}f\left(\mathbf{x}\right)\mathrm{e}^{-\mathrm{i}2\pi\mathbf{k}\cdot\mathbf{x}}d\mathbf{x}\label{eq:Fourier_transform}
\end{align}
The \emph{inverse Fourier transform} is defined by 
\begin{align}
\mathcal{F}^{-1}\left[\hat{f}\right]\left(\mathbf{x}\right)=f\left(\mathbf{x}\right) & =\int_{\mathbb{R}^{N}}\hat{f}\left(\mathbf{k}\right)\mathrm{e}^{\mathrm{i}2\pi\mathbf{k}\cdot\mathbf{x}}d\mathbf{k}\label{eq:inverse_Fourier_transform}
\end{align}

Denoting the \emph{convolution} operator by $\ast$, convolution of
two functions is defined by 
\begin{align}
\left(f\ast g\right)\left(\mathbf{x}\right) & =\int_{\mathbb{R}^{N}}f\left(\mathbf{y}\right)g\left(\mathbf{x}-\mathbf{y}\right)d\mathbf{y}\label{eq:convolution}
\end{align}
The Fourier transform of the convolutions is the product of the Fourier
transforms:

\begin{align}
\mathcal{F}\left[f\ast g\right]\left(\mathbf{k}\right) & =\hat{f}\left(\mathbf{k}\right)\hat{g}\left(\mathbf{k}\right)\label{eq:Plencherel}
\end{align}
also referred to as convolution theorem.

\section{Band-limited functions \label{sec:Bandlimited-functions}}

We say that $f_{B}\left(t\right)$, for $t\in\mathbb{R}$, is a \emph{band-limited
}function with band-limit $B$ if there exists an $\hat{f_{B}}\left(\omega\right)$
such that 
\begin{align}
f_{B}\left(t\right) & =\int_{-B}^{B}\hat{f}_{B}\left(\omega\right)\mathrm{e}^{\mathrm{i}2\pi t\omega}d\omega\label{eq:band-limited_inverse_Fourier_transform}\\
 & =B\int_{-1}^{1}\hat{f}_{B}\left(B\omega\right)\mathrm{e}^{\mathrm{i}2\pi Bt\omega}d\omega\nonumber 
\end{align}

Given a function $f\left(t\right)$, its \emph{band-limited projection}
$P_{B}\left[f\right]\left(t\right)$, denoted by $f_{B}\left(t\right)$
for short, is defined by 
\begin{align}
P_{B}\left[f\right]\left(t\right)=f_{B}\left(t\right) & =\int_{-B}^{B}\hat{f}\left(\omega\right)\mathrm{e}^{\mathrm{i}2\pi t\omega}d\omega\nonumber \\
 & =B\int_{-1}^{1}\hat{f}\left(B\omega\right)\mathrm{e}^{\mathrm{i}2\pi Bt\omega}d\omega\label{eq:band-limited-projection-Fourier-transform}
\end{align}
or, equivalently, in the convolution representation using Parseval's
theorem 
\begin{align}
P_{B}\left[f\right]\left(t\right)=f_{B}\left(t\right) & =\int_{-\infty}^{\infty}f\left(\tau\right)\,2B\,\mathrm{sinc}\left(2\pi B\left[t-\tau\right]\right)d\tau\label{eq:band-limited-projection-sinc-interp}
\end{align}
where 
\begin{align}
\mathrm{sinc}\left(Bt\right) & =\frac{1}{2B}\int_{-B}^{B}\mathrm{e}^{\mathrm{i}\omega t}d\omega=\int_{0}^{1}\cos\left(B\omega t\right)d\omega=\frac{\mathrm{sin(Bt)}}{Bt}
\end{align}
is the sinc function normalized with band-limit $B$. Note that 
\begin{align}
\hat{f}_{B}\left(\omega\right) & =\hat{f}\left(\omega\right),\quad\omega\in\left[-B,B\right]
\end{align}

\subsection{Discrete Fourier representation of band-limited approximation of
compactly supported functions: }

In this section we derive discrete Fourier approximations of band-limited
projection of compactly supported functions starting from their convolution
representation.

Consider a discretization of the integral representation of sinc (see
Figures \ref{fig:Approximation-of-sinc_in_cosines} and \ref{fig:Approximation-of-sinc_in_cosines-uniform-sampling}
for two examples. Another example in Section 8 of \cite{beylkin2002generalized}.)

\begin{align}
2B\,\mathrm{sinc}\left(Bt\right) & =2\int_{0}^{B}\cos\left(\omega t\right)d\omega\nonumber \\
 & =\sum_{m=1}^{M}2\alpha_{m}\cos\left(B\omega_{m}t\right)+\alpha_{0}+\epsilon_{B}\left(t\right)
\end{align}
for a given $B\in\mathbb{R}^{+}$, with $\alpha_{m}\in\mathbb{R}^{+}$
and $\omega_{m}\in\left[0,1\right]$. Equivalently, using exponentials
instead of cosines, we write 
\begin{align}
2B\,\mathrm{sinc}\left(Bt\right)=\int_{-B}^{B}\mathrm{e}^{\mathrm{i}t\omega}d\omega & =\sum_{m=-M}^{M}\alpha_{m}\mathrm{e}^{\mathrm{i}B\omega_{m}t}+\epsilon_{B}\left(t\right)\label{eq:sinc-in-exp}
\end{align}
where $-\omega_{m}=\omega_{-m}$ and $\alpha_{-m}=\alpha_{m}$. Now
we can prove: 
\begin{thm}
\label{thm:discrete_Fourier_approx_bandlimited_functions}Given a
compactly supported function $f\left(t\right)$ over $\left[-T,T\right]$,
restriction of its band-limited projection, $f_{B}\left(t\right)$,
onto interval $\left[-T,T\right]$ can be approximated as a discrete
sum of Fourier basis by 
\begin{align}
f_{B}\left(t\right) & =\sum_{m=-M}^{M}\alpha_{m}\hat{f}\left(B\omega_{m}\right)\mathrm{e}^{\mathrm{i}2\pi B\omega_{m}t}+\epsilon_{f}\left(t\right)\label{eq:f_B_discrete_Fourier_approx}
\end{align}
using the approximation (\ref{eq:sinc-in-exp}) with the error bound
\begin{align}
\max_{t\in\left[-T,T\right]}\left|\epsilon_{f}\left(t\right)\right| & \le2T\max_{t\in\left[-T,T\right]}\left|f\left(t\right)\right|\max_{t\in\left[-2T,2T\right]}\left|\epsilon_{B}\left(2\pi t\right)\right|.\label{eq:max_epsilon_f}
\end{align}
\end{thm}
\begin{proof}
For a function $f\left(\tau\right)$ compactly supported on $\tau\in\left[-T,T\right]$,
substituting (\ref{eq:sinc-in-exp}) into (\ref{eq:band-limited-projection-sinc-interp}),
its band-limited projection can be approximated by (\ref{eq:f_B_discrete_Fourier_approx})
where 
\begin{align}
\epsilon_{f}\left(t\right) & =\int_{-T}^{T}f\left(\tau\right)\epsilon_{B}\left(2\pi\left[t-\tau\right]\right)d\tau\label{eq:e_f(t)}
\end{align}
(\ref{eq:f_B_discrete_Fourier_approx}) provides a discretization
of (\ref{eq:band-limited-projection-Fourier-transform}) through approximation
of the sinc function as a sum of cosines. 
\end{proof}
\begin{quote}
\end{quote}
\begin{example}
Choosing 
\begin{align}
\left(\alpha_{m},\omega_{m}\right)_{m=-M}^{M} & =\left(\frac{2B}{2M+1},\frac{2m}{2M+1}\right)_{m=-M}^{M}\label{eq:DFT-quadratures}
\end{align}
for some $M\ge0$, (\ref{eq:f_B_discrete_Fourier_approx}) becomes
the discrete Fourier transform representation of $f_{B}\left(t\right)$:
\begin{align}
f_{B}\left(t\right) & =\frac{2B}{2M+1}\sum_{m=-M}^{M}\hat{f}\left(B\frac{2m}{M+1}\right)\mathrm{e}^{\mathrm{i}2\pi B\frac{2m}{2M+1}t}+\epsilon_{f}\left(t\right)\label{eq:discrete-inverse-Fourier-transform}
\end{align}
The summation term is referred to as the\emph{ discrete inverse Fourier
transform} of $\hat{f}$. (\ref{eq:discrete-inverse-Fourier-transform})
is a Riemann sum approximation of the integral (\ref{eq:band-limited_inverse_Fourier_transform})
for uniform sampling of the interval $\left[-B,B\right]$. 
\end{example}
In practice measurements are performed over a finite duration. Thus
it is desirable to have the band-limited projection of a compactly
supported function approximately agree with the function at least
over $t\in\left[-T,T\right]$. Thus, by (\ref{eq:max_epsilon_f}),
for $f_{B}\left(t\right)$ to approximate accurately $f\left(t\right)$
over $\left[-T,T\right]$, one needs to build up an approximation
to $\mathrm{sinc}\left(Bt\right)$ that is accurate over the interval
$\left[-4\pi T,4\pi T\right]$. In Appendix \ref{sec:On-approximations-of-sinc},
we present two different approximations in the form of (\ref{eq:sinc-in-exp})
(see Figures \ref{fig:Approximation-of-sinc_in_cosines} and \ref{fig:Approximation-of-sinc_in_cosines-uniform-sampling}),
one using Gauss-Legendre quadratures (see Figure \ref{fig:alpha_m+omega_m-B0=00003D00003D20})
and the other using uniform sampling. We show that, for a desired
interval and bandwidth, a discrete representation of sinc that is
accurate upto machine precision can be achieved using Gauss-Legendre
quadratures without requiring as many uniform samples.
\begin{example}
From a finite duration measurement, only finite number of samples
are utilized for digital signal processing. This raises a natural
question: ``What should be the sampling rate for a band-limited measurement
such that band-limited projection of the the discrete measurement
agree with with the discrete measurement?''. In this regard, consider
the following model for a discrete measurement 
\begin{align}
f\left(t\right) & =\sum_{k=-K}^{K}f_{k}\delta\left(t-\frac{2k}{2K+1}T\right)\label{eq:f-sum-of-deltas}
\end{align}
where $2T/(2K+1)$ is the sampling period. Then 
\begin{align}
\hat{f}\left(\omega\right) & =\sum_{k=-K}^{K}f_{k}\mathrm{e}^{-\mathrm{i}2\pi\omega\frac{2k}{2K+1}T}\label{eq:discrete-Fourier-transform}
\end{align}
Assuming that the measurement has band-limit $B$, let $\omega_{m}=B\,2m\,(2M+1)^{-1}$.
By (\ref{eq:discrete-Fourier-transform}), we rewrite (\ref{eq:discrete-inverse-Fourier-transform})
in terms of $\hat{f}\left(\omega_{m}\right)$ and obtain 
\begin{multline}
f_{B}\left(t\right)=\frac{2B}{2M+1}\sum_{k=-K}^{K}f_{k}\left(\sum_{m=-M}^{M}\mathrm{e}^{\mathrm{i}2\pi B\frac{2m}{2M+1}\left(t-\frac{2k}{2K+1}T\right)}\right)\\
+\epsilon_{f}\left(t\right)
\end{multline}
which, for $t=2l/\left(2K+1\right)T$, $l=-K,\ldots,K$, becomes 
\begin{multline}
f_{B}\left(\frac{2l}{2K+1}T\right)=\frac{2B}{2M+1}\sum_{k=-K}^{K}f_{k}\left(\sum_{m=-M}^{M}\mathrm{e}^{\mathrm{i}2\pi B\frac{2m}{2M+1}\frac{2T}{2K+1}\left(l-k\right)}\right)\\
+\epsilon_{f}\left(\frac{2l}{2K+1}T\right)\label{eq:f_B-of-f-sum-of-deltas}
\end{multline}
(\ref{eq:discrete-Fourier-transform}) is known as the \emph{discrete
Fourier transform }of the vector $\left\{ f_{k}\right\} _{k=-K}^{K}$. 
\end{example}
For $T=\left(2K+1\right)/(4B)$, (\ref{eq:f-sum-of-deltas}), (\ref{eq:discrete-Fourier-transform})
for $\omega=B\,2m\,(2M+1)^{-1}$, and (\ref{eq:f_B-of-f-sum-of-deltas})
become 
\begin{align}
f\left(t\right) & ={\textstyle \sum_{k=-K}^{K}f_{k}\delta\left(t-\frac{k}{2B}\right)}\label{eq:f-sum-of-nyquist-deltas}
\end{align}
\begin{align}
\hat{f}\left(B\frac{2m}{2M+1}\right) & ={\textstyle \sum_{k=-K}^{K}f_{k}\mathrm{e}^{-\mathrm{i}2\pi\frac{m}{2M+1}k}}
\end{align}
\begin{align}
f_{B}\left(\frac{l}{2B}\right) & =\frac{2B}{2M+1}\sum_{k=-K}^{K}f_{k}\left(\sum_{m=-M}^{M}\mathrm{e}^{\mathrm{i}2\pi\frac{m}{2M+1}\left(l-k\right)}\right)+\epsilon_{f}\left(\frac{l}{2B}\right)
\end{align}
which, for $\left|l-k\right|\le2M$, is 
\begin{align}
f_{B}\left(\frac{l}{2B}\right) & =2B\sum_{k=-K}^{K}f_{k}\delta_{kl}+\epsilon_{f}\left(\frac{l}{2B}\right)=2Bf_{l}+\epsilon_{f}\left(\frac{l}{2B}\right),\label{eq:nyquist-sampling-quadratures}
\end{align}
where $\delta_{kl}$ is the Kronecker delta function. \footnote{Consider (\ref{eq:f-sum-of-deltas}) for $t=lT(K+1)^{-1}$: 
\[
{\textstyle f\left(l\frac{T}{K+1}\right)=\sum_{k=-K}^{K}f_{k}\frac{K+1}{T}\delta\left(l-k\right).}
\]
For $T=(K+1)/(2B)$, $f\left(\frac{l}{2B}\right)=2B\sum_{k=-K}^{K}f_{k}\delta\left(l-k\right)$.
Thus the factor $2B$ in front of the sum in (\ref{eq:nyquist-sampling-quadratures})
is due to difference between continious and discrete nature of Dirac
delta, $\delta\left([l-k](2B)^{-1}\right)=2B\,\delta\left(l-k\right)$,
and Kronecker delta, $\delta_{lk}$, functions. } The inequality $\left|l-k\right|\le2M$ imposes that $M\ge K$. For
$T=(2K+1)/(4B)$, by (\ref{eq:f-sum-of-deltas}) and (\ref{eq:sinc-in-exp}),
(\ref{eq:e_f(t)}) becomes 
\begin{align}
\epsilon_{f}\left(\frac{l}{2B}\right) & =\int_{-T}^{T}f\left(\tau\right)\epsilon_{B}\left(2\pi\left[t-\frac{l}{2B}\right]\right)d\tau\nonumber \\
 & =\sum_{k=-K}^{K}f_{k}\epsilon_{B}\left(2\pi\left[\frac{k}{2B}-\frac{l}{2B}\right]\right)\nonumber \\
 & =\sum_{k=-K}^{K}f_{k}\left[\begin{array}{l}
2B\mathrm{sinc}\left(\pi\left[k-l\right]\right)\\
-\left(\frac{2B}{2M+1}\sum_{m=-M}^{M}\mathrm{e}^{\mathrm{i}2\pi\frac{m}{2M+1}\left(l-k\right)}\right)
\end{array}\right]\nonumber \\
 & =\sum_{k=-K}^{K}f_{k}\left[2B\delta_{k,l}-\frac{2B}{2M+1}\left(2M+1\right)\delta_{k,l}\right]\nonumber \\
 & =0.
\end{align}
The special sampling rate $1/(2B)$ that gave rise to the band-limited
function $f_{B}\left(t\right)$ whose values are equal to the original
function at $t=l/(2B)$. This sampling rate is referred to as the
Nyquist rate. By sinc interpolation, also known as Whittaker-Shannon
interpolation formula \cite{stankovic2006some}, 
\begin{align}
f_{B}\left(t\right) & =\sum_{k=-\infty}^{\infty}f_{B}\left(\frac{k}{2B}\right)\,\mathrm{sinc}\left(2\pi B\left[t-\frac{k}{2B}\right]\right),\label{eq:sinc-interpolation-formula}
\end{align}
a band-limited function can be exactly determined from samples obtained
using Nyquist rate. In digital signal processing, because finitely
many samples are measured, which are modelled by (\ref{eq:f-sum-of-nyquist-deltas}),
(\ref{eq:sinc-interpolation-formula}) is approximated by 
\begin{align}
f_{B}\left(t\right) & \approx\sum_{k=-K}^{K}f\left(\frac{k}{2B}\right)\,\mathrm{sinc}\left(2\pi B\left[t-\frac{k}{2B}\right]\right),\label{eq:approx-sinc-interpolation-formula}
\end{align}
which is band-limited projection of (\ref{eq:f-sum-of-deltas}). This
approximation agrees with the discrete measurements. However, there
is an approximation error between the sample locations as a result
of implicit imposition $f\left(k(2B)^{-1}\right)$ equal to zero for
$\left|k\right|>K$. This imposition is eliminated when, instead of
discrete Fourier basis, prolate spheroidal wave functions (PSWFs)
are used as a basis to represent band-limited functions. Expressing
band-limited function in terms of PSWFs doesn't directly answer how
a band-limited function should be sampled but provides the necessary
foundation to answer ``How accurately can we approximate a band-limited
function from its samples given over a compact support?'', which
is addressed in Theorem \eqref{thm:sampling-theorem-bandlimited-functions}.

\subsection{Band-limited projections of compactly supported function and prolate
spheroidal wave functions}

In this section we present prolate spheroidal wave functions, their
properties and two methods on how we can numerically compute them.
For the rest of our discussion we will assume that $T=1$. This can
be compensated by choosing the band-limit to be $T$ times more.

\subsubsection{Prolate spheroidal wave functions (PSWF)}

\emph{Prolate spheroidal wave functions }(PSWF), $\varphi_{n}\left(t\right)$
can be defined as the eigenfunctions of the band-limited projection
operator restricted to a compact support, which, without loss of generality,
is given by \cite{slepian1961prolate,osipov2013prolate} 
\begin{align}
\int_{-1}^{1}2B\,\mathrm{sinc}\left(2\pi B\left(t-\tau\right)\right)\varphi_{n}\left(\tau\right)d\tau & =\mu_{n}\varphi_{n}\left(t\right),\quad t\in\mathbb{R}\label{eq:PSWF-definition}
\end{align}
PSWF form a basis for band-limited functions as well as $L_{2}\left(\left[-1,1\right]\right)$,
and satisfy the following properties \cite{slepian1961prolate,osipov2013prolate}: 
\begin{enumerate}
\item PSWF are real valued and corresponding eigenvalues $\mu_{n}$ are
positive: $\varphi_{n}\left(t\right)\in\mathbb{R}$, $\mu_{n}\in\mathbb{R}^{+}$. 
\item PSWF are orthogonal within the interval $t\in\left[-1,1\right]$ as
well as over the real line: 
\begin{align}
\mu_{n}\int_{-\infty}^{\infty}\varphi_{n}\left(t\right)\varphi_{m}\left(t\right)dt & =\int_{-1}^{1}\varphi_{n}\left(t\right)\varphi_{m}\left(t\right)dt=\delta_{m,n}\label{eq:PSWF-orthogonality}
\end{align}
where $\delta_{m,n}$ is the Kronocker delta function equal to $1$
for $m=n$ and zero otherwise. 
\item PSWF are eigenfunctions of Fourier operator restricted to the interval
$\left[-1,1\right]$: 
\begin{align}
\int_{-1}^{1}\varphi_{n}\left(t\right)\mathrm{e}^{\mathrm{i}2\pi B\omega t}dt & =\lambda_{n}\varphi_{n}\left(\omega\right),\quad\omega\in\left[-1,1\right],\lambda_{n}\in\mathbb{C}\label{eq:PSWF-eigenfunction-Fourier}
\end{align}
\end{enumerate}
The second property implies that if a band-limited function is known
within an interval then it is known

The eigenvalues satisfy the following properties: 
\begin{enumerate}
\item Multiplying both sides of (\ref{eq:PSWF-eigenfunction-Fourier}) with
$\mathrm{e}^{-\mathrm{i}2\pi B\omega t}$, integrating with respect
to $\omega$ and comparing the result with (\ref{eq:PSWF-definition})
one obtains $\mu_{n}=B\left|\lambda_{n}\right|^{2}$ (see 3.48 in
\cite{osipov2013prolate}) 
\item {[}\cite{landau1980eigenvalue},Theorem 3.14 in \cite{osipov2013prolate}{]}
Let $N\left(\alpha\right)$ denote the number of eigenvalues $\mu_{n}>\alpha$
for some $0<\alpha<1$. Then 
\begin{align}
N\left(\alpha\right) & =4B+\left(\frac{1}{\pi^{2}}\log\left(\frac{1-\alpha}{\alpha}\right)\right)\log\left(2\pi B\right)+O\left(\log\left(2\pi B\right)\right)\label{eq:N(alpha)}
\end{align}
Thus there are about $4B$ eigenvalues $\mu_{n}$ that are close to
one, on the order of $\log\left(2\pi B\right)$ eigenvalues that decay
rapidly, and the rest of them are very close to zero. 
\end{enumerate}
For a comprehensive review on PSWF, we refer the reader to \cite{osipov2013prolate}. 
\begin{lem}
\label{lem:f_B-f} Given a compactly supported function $f\left(t\right)$
on $t\in\left[-1,1\right]$, it can be expressed in terms of PSWFs
by 
\begin{align}
f\left(t\right) & =\sum_{n}f_{B,n}\varphi_{n}\left(t\right)
\end{align}
where 
\begin{align}
f_{B,n} & =\int_{-1}^{1}f\left(t\right)\varphi_{n}\left(t\right)dt\nonumber \\
 & =\frac{1}{\mu_{n}}\int_{-1}^{1}f_{B}\left(t\right)\varphi_{n}\left(t\right)dt\nonumber \\
 & =\int_{-\infty}^{\infty}f_{B}\left(t\right)\varphi_{n}\left(t\right)dt\label{eq:f_B,n}
\end{align}
\end{lem}
\begin{proof}
Because $2B\,\mathrm{sinc}\left(2\pi B\left(t-\tau\right)\right)$
is band-limited, it can be expanded as a sum of PSWF, 
\begin{align}
2B\,\mathrm{sinc}\left(2\pi B\left(t-\tau\right)\right) & =\sum_{n}a_{n}\left(\tau\right)\varphi_{n}\left(t\right),
\end{align}
where 
\begin{align}
a_{n}\left(\tau\right) & =\int_{-1}^{1}2B\,\mathrm{sinc}\left(2\pi B\left(t-\tau\right)\right)\varphi_{n}\left(t\right)dt=\mu_{n}\varphi_{n}\left(\tau\right)
\end{align}
leading to the decomposition of sinc in terms of PSWF: 
\begin{align}
2B\,\mathrm{sinc}\left(2\pi B\left(t-\tau\right)\right) & =\sum_{n}\mu_{n}\varphi_{n}\left(\tau\right)\varphi_{n}\left(t\right)\label{eq:sinc-in-PSWF}
\end{align}
Given a compactly supported function $f\left(t\right)$ on $t\in\left[-1,1\right]$,
it's band-limited projection $f_{B}\left(t\right)$ can be expanded
in term of PSWF by substituting (\ref{eq:sinc-in-PSWF}) in (\ref{eq:band-limited-projection-sinc-interp})
\begin{align}
f_{B}\left(t\right) & =\sum_{n}\mu_{n}f_{B,n}\varphi_{n}\left(t\right).\label{eq:f_B-in-PSWF}
\end{align}
Then, by (\ref{eq:PSWF-orthogonality}), the coefficients can be computed
by either of the three ways in (\ref{eq:f_B,n}). 
\end{proof}

\subsubsection{Approximating PSWF as eigenvectors of $\mathrm{e}^{-\mathrm{i}2\pi B\omega_{k}\omega_{m}}$\label{subsec:approx-PSWF-1}}

Substituting (\ref{eq:sinc-in-exp}) into (\ref{eq:PSWF-definition}),
and recalling $\mu_{n}=B\left|\lambda_{n}\right|^{2}$, by (\ref{eq:PSWF-eigenfunction-Fourier}),
we obtain 
\begin{align}
\varphi_{n}\left(t\right) & =\frac{1}{\mu_{n}}\int_{-1}^{1}2B\,\mathrm{sinc}\left(2\pi B\left(t-\tau\right)\right)\varphi_{n}\left(\tau\right)d\tau\nonumber \\
 & =\frac{1}{B\lambda_{n}}\sum_{m=-M}^{M}\alpha_{m}\mathrm{e}^{\mathrm{i}2\pi B\omega_{m}t}\varphi_{n}\left(\omega_{m}\right)+\epsilon_{\varphi_{n}}\left(t\right)\label{eq:PSWF_cont_2_discrete}
\end{align}
where 
\begin{align}
\epsilon_{\varphi_{n}}\left(t\right) & =\frac{1}{\mu_{n}}\int_{-1}^{1}\epsilon_{B}\left(2\pi\left[t-\tau\right]\right)\varphi_{n}\left(\tau\right)d\tau\label{eq:epsilon_var_phi_n(t)}
\end{align}
with \footnote{By Hölder's inequality, 
\begin{align*}
\left|\epsilon_{\varphi_{n}}\left(t\right)\right|^{2} & \le\frac{1}{\mu_{n}^{2}}\int_{-1}^{1}\left|\epsilon_{B}\left(2\pi\left[t-\tau\right]\right)\right|^{2}\left|\varphi_{n}\left(\tau\right)\right|^{2}d\tau\\
 & \le\frac{1}{\mu_{n}^{2}}\max_{\tau\in\left[t-1,t+1\right]}\left|\epsilon_{B}\left(2\pi\tau\right)\right|^{2}\underbrace{\int_{-1}^{1}\left|\varphi_{n}\left(\tau\right)\right|^{2}d\tau}_{=1}
\end{align*}
} 
\begin{align}
\max_{t\in\left[-1,1\right]}\left|\epsilon_{\varphi_{n}}\left(t\right)\right| & \le\frac{1}{\mu_{n}}\,\max_{t\in\left[-2,2\right]}\left|\epsilon_{B}\left(2\pi t\right)\right|\label{eq:max-bound-epsilon_var_phi_n(t)}
\end{align}

In \cite{beylkin2002generalized} (see equation (8.19) and (8.20)
in \cite{beylkin2002generalized}), (\ref{eq:PSWF_cont_2_discrete})
was used to build approximate PSWFs by first solving the eigensystem
\begin{align}
\varphi_{n}\left(\omega_{k}\right) & =\frac{1}{B\lambda_{n}}\sum_{m=-M}^{M}\alpha_{m}\mathrm{e}^{\mathrm{i}2\pi B\omega_{m}\omega_{k}}\varphi_{n}\left(\omega_{m}\right)\label{eq:approx_PSWF_1}
\end{align}
for the eigenvector $\varphi_{n}\left(\omega_{k}\right)$, followed
by substituting $\varphi_{n}\left(\omega_{k}\right)$ back in (\ref{eq:PSWF_cont_2_discrete}):
\begin{align}
\varphi_{n}\left(t\right) & =\frac{1}{B\lambda_{n}}\sum_{m=-M}^{M}\alpha_{m}\mathrm{e}^{\mathrm{i}2\pi B\omega_{m}t}\varphi_{n}\left(\omega_{m}\right)+\epsilon_{\varphi_{n}}\left(t\right)
\end{align}
where $\epsilon_{\varphi_{n}}\left(\omega_{k}\right)=0$, for $k=-M,\ldots,M$.

Thus, the eigensystem (\ref{eq:approx_PSWF_1}) provides an approximation
to PSWF over the interval $\left[-1,1\right]$ bounded by (\ref{eq:max-bound-epsilon_var_phi_n(t)}).
Because (\ref{eq:approx_PSWF_1}) is a system of $2M+1$ equations,
it has $2M+1$ eigenvalues, which we will denote by $\mu_{n=0,\ldots,2M}$.
By (\ref{eq:N(alpha)}), in order to capture the dominant eigenvalues,
i.e. eigenvalues around 1, one shall have $M\ge\left\lceil 2B-1/2\right\rceil $.

For sake of simplicity of the discussion, we will assume that $M\gg\left\lceil 2B-1/2\right\rceil $,
$\mu_{2M}\ll1$. Thus the corresponding $2M+1$ approximate PSWFs
provides a sufficiently accurately approximate band-limited functions
over the interval $\left[-1,1\right]$ and, by the same token, sinc
function over $\left[-2,2\right]$. Thus we treat, (\ref{eq:sinc-in-PSWF})
is equivalent to its truncated version: 
\begin{align}
2B\,\mathrm{sinc}\left(2\pi B\left(t-\tau\right)\right) & \approx\sum_{n=-M}^{M}\mu_{n}\varphi_{n}\left(\tau\right)\varphi_{n}\left(t\right),t,\tau\in\left[-1,1\right],
\end{align}
and similarly all the infinite sums over PSWFs as finite sums. 
\begin{example}
Consider (\ref{eq:PSWF-eigenfunction-Fourier}) and the quadratures
of discrete inverse Fourier transform (\ref{eq:DFT-quadratures})
for $B=\frac{2M+1}{4}$ 
\begin{align}
\left(\alpha_{m},\omega_{m}\right)_{m=-M}^{M} & =\left(\frac{1}{2},\frac{2m}{2M+1}\right)_{m=-M}^{M}
\end{align}
for some positive integer $M$. Then (\ref{eq:PSWF_cont_2_discrete})
becomes 
\begin{align}
\varphi_{n}\left(t\right) & =\frac{1}{\lambda_{n}}\frac{2}{2M+1}\sum_{m=-M}^{M}\mathrm{e}^{\mathrm{i}\pi mt}\varphi_{n}\left(\frac{2m}{2M+1}\right)+\epsilon_{\varphi_{n}}\left(t\right)
\end{align}
Consequently, the eigensystem for approximate PSWF is 
\begin{align}
\varphi_{n}\left(\frac{2k}{2M+1}\right) & =\frac{1}{\lambda_{n}}\frac{2}{2M+1}\sum_{m=-M}^{M}\mathrm{e}^{\mathrm{i}2\pi\frac{mk}{2M+1}}\varphi_{n}\left(\frac{2m}{2M+1}\right)
\end{align}
for $k=-M,\ldots,M$, which for $M\ge2$ has four distinct eigenvalues,
$\pm2\sqrt{2M+1}$ and $\pm\mathrm{i}2\sqrt{2M+1}$ with multiplicities
(see page 32 of \cite{briggs1995dft}). 
\end{example}

\subsubsection{Approximating PSWF as eigenvectors of $\mathrm{sinc}\left(2\pi B\left(\omega_{m}-\omega_{k}\right)\right)$\label{subsec:Approximating-PSWF-from-shifts-of-sinc}}

An alternative to the method presented in Section \ref{subsec:approx-PSWF-1},
PSWF can be approximated through discretization of (\ref{eq:PSWF-definition}).

Starting with (\ref{eq:PSWF_cont_2_discrete}) and using (\ref{eq:PSWF-eigenfunction-Fourier}),
we have 
\begin{align}
\varphi_{n}\left(\omega\right) & =\frac{1}{\overline{\lambda_{n}}}\int_{-1}^{1}\mathrm{e}^{-\mathrm{i}2\pi B\omega t}\varphi_{n}\left(t\right)dt\label{eq:varphi_n(t)-in-sinc}\\
 & =\frac{1}{B\mu_{n}}\sum_{m=-M}^{M}\hspace{-0.2cm}\alpha_{m}2B\,\mathrm{sinc}\left(2\pi B\left(\omega_{m}-\omega\right)\right)\varphi_{n}\left(\omega_{m}\right)+\varepsilon_{\varphi_{n}}\left(\omega\right)
\end{align}
where 
\begin{align}
\varepsilon_{\varphi_{n}}\left(\omega\right) & =\frac{1}{B\overline{\lambda_{n}}}\int_{-1}^{1}\mathrm{e}^{-\mathrm{i}2\pi B\omega t}\epsilon_{\varphi_{n}}\left(t\right)dt
\end{align}
Because $\mu_{n}=B\left|\lambda_{n}\right|^{2}$, by (\ref{eq:max-bound-epsilon_var_phi_n(t)}),
\begin{align}
\max_{t\in\left[-1,1\right]}\left|\varepsilon_{\varphi_{n}}\left(t\right)\right| & \le\frac{2}{\left|\lambda_{n}\right|}\max_{t\in\left[-1,1\right]}\left|\epsilon_{\varphi_{n}}\left(t\right)\right|\\
 & \le\frac{2\sqrt{B}}{\mu_{n}^{3/2}}\,\max_{t\in\left[-2,2\right]}\left|\epsilon_{B}\left(2\pi t\right)\right|\label{eq:varepsilon_varphi_n(t)}
\end{align}

Similar to the method of \cite{beylkin2002generalized}, (\ref{eq:varphi_n(t)-in-sinc})
can be used to build approximate PSWF by first solving the eigensystem
\begin{align}
\varphi_{n}\left(\omega_{m}\right) & =\frac{1}{B\mu_{n}}\sum_{k=-M}^{M}\alpha_{k}2B\,\mathrm{sinc}\left(2\pi B\left(\omega_{m}-\omega_{k}\right)\right)\varphi_{n}\left(\omega_{k}\right)\label{eq:varphi_n(omega_m)}
\end{align}
for the eigenvector $\varphi_{n}\left(\omega_{k}\right)$, followed
by substituting $\varphi_{n}\left(\omega_{k}\right)$ back in (\ref{eq:varphi_n(t)-in-sinc}):
\begin{align}
\varphi_{n}\left(t\right) & =\frac{1}{B\mu_{n}}\sum_{k=-M}^{M}\alpha_{k}2B\,\mathrm{sinc}\left(2\pi B\left(t-\omega_{k}\right)\right)\varphi_{n}\left(\omega_{k}\right)+\varepsilon_{\varphi_{n}}\left(t\right),
\end{align}
where $\varepsilon_{\varphi_{n}}\left(\omega_{m}\right)=0$, for $m=-M,\ldots,M$.
The eigenvectors $\phi_{n}(\omega_{k})$ are generalizations of discrete
prolate spheroidal sequences (DPSS) \cite{slepian1978prolate}. When
$\omega_{k}$ are uniformly sampled they are equivalent to DPSS, which
asymptotically approximate PSWF \cite{slepian1978prolate}. Similar
to the discussion in Section \ref{subsec:approx-PSWF-1}, we say the
eigensystem (\ref{eq:varphi_n(omega_m)}) provides an approximation
to PSWF. Because it can only capture $2M+1$ of the eigenvalues, which
we denote by $\mu_{n=0,\ldots,2M}$, by (\ref{eq:N(alpha)}), one
shall choose $M\ge\left\lceil 2B-1/2\right\rceil $ in order to capture
all the eigenvalues close to one and some of the eigenvalues in the
transition zone from one to zero, depending on the desired accuracy
of the approximation. 
\begin{example}
Consider the quadratures of discrete inverse Fourier transform for
$B=\frac{2M+1}{4}$ 
\begin{align}
\left(\alpha_{m},\omega_{m}\right)_{m=-M}^{M} & =\left(\frac{1}{2},\frac{2m}{2M+1}\right)_{m=-M}^{M}
\end{align}
for some positive integer $M$. Then (\ref{eq:varphi_n(t)-in-sinc})
becomes 
\begin{align}
\varphi_{n}\left(t\right) & =\frac{1}{\mu_{n}}\sum_{k=-M}^{M}\mathrm{sinc}\left(\pi\left[\frac{2M+1}{2}t-k\right]\right)\varphi_{n}\left(\frac{2k}{2M+1}\right)+\varepsilon_{\varphi_{n}}\left(t\right)
\end{align}
Consequently, the eigensystem for approximate PSWF is 
\begin{align}
\varphi_{n}\left(\frac{2m}{2M+1}\right) & =\frac{1}{\mu_{n}}\sum_{k=-M}^{M}\mathrm{sinc}\left(\pi\left[m-k\right]\right)\varphi_{n}\left(\frac{2k}{2M+1}\right)\nonumber \\
 & =\frac{1}{\mu_{n}}\sum_{k=-M}^{M}\delta_{m,k}\varphi_{n}\left(\frac{2k}{2M+1}\right)\nonumber \\
 & =\frac{1}{\mu_{n}}\varphi_{n}\left(\frac{2m}{2M+1}\right)
\end{align}
for $m=-M,\ldots,M$ which implies that $\mu_{n}=1$ for $n=0,\ldots,2M$.
As mentioned above, $\varphi_{n}\left(\frac{2m}{M+1}\right)$ are
referred to as discrete prolate spheroidal sequences and were studied
in \cite{slepian1978prolate} along with their relationship to periodic
discrete prolate spheroidal sequences (P-DPSS). This example shows
that, similar to P-DPSS \cite{xu1984periodic}, eigenvalues of DPSS
are not necessarily simple and therefore definition of DPSS can be
non unique.
\end{example}

\subsection{Discrete convolution representation of band-limited approximation
of compactly supported functions}
\begin{thm}
\label{thm:sampling-theorem-bandlimited-functions} Consider a function
$f\left(t\right)$ compactly supported on $t\in\left[-1,1\right]$.
Its band-limited projection $f_{B}\left(t\right)$ can be computed
by 
\begin{align}
f_{B}\left(t\right) & =\sum_{k=-M}^{M}2B\,\mathrm{sinc}\left(2\pi B\left(t-\omega_{k}\right)\right)f_{k}+\epsilon_{f_{B}}\left(t\right)\label{eq:f_B_approx_interp}
\end{align}
where 
\begin{align}
f_{k} & ={\displaystyle \sum_{m=-M}^{M}f\left(\omega_{m}\right)\alpha_{m}R_{m}\left(\omega_{k}\right)}
\end{align}
and $\left(\alpha_{m},\omega_{m}\right)$ satisfy (\ref{eq:sinc-in-exp})
and 
\begin{align}
R_{m}\left(t\right) & =\sum_{n=0}^{2M}\mu_{n}^{-1}\varphi_{n}\left(\omega_{m}\right)\varphi_{n}\left(t\right)
\end{align}
\end{thm}
\begin{proof}
Substituting (\ref{eq:varphi_n(t)-in-sinc}) in (\ref{eq:f_B-in-PSWF}),
we obtain 
\begin{align}
f_{B}\left(t\right) & =\sum_{k=-M}^{M}\alpha_{k}2B\,\mathrm{sinc}\left(2\pi B\left(t-\omega_{k}\right)\right)\tilde{f}_{B}\left(\omega_{k}\right)+\epsilon_{f_{B}}\left(t\right)\label{eq:f_B(t)-in-sinc}
\end{align}
where 
\begin{align}
\tilde{f}_{B}\left(t\right) & =\frac{1}{B}\sum_{n=0}^{2M}f_{B,n}\varphi_{n}\left(t\right)\label{eq:hat_f_Bk}\\
\epsilon_{f_{B}}\left(t\right) & =\sum_{n=0}^{2M}\mu_{n}f_{B,n}\varepsilon_{\varphi_{n}}\left(t\right)\label{eq:epsilon_f_B}
\end{align}
with $f_{B,n}$ defined in (\ref{eq:f_B,n}). Similarly, substituting
(\ref{eq:varphi_n(t)-in-sinc}) in (\ref{eq:f_B,n}), we obtain 
\begin{align}
f_{B,n} & =\frac{1}{B\mu_{n}}\sum_{k=-M}^{M}\alpha_{k}\varphi_{n}\left(\omega_{k}\right)f_{B}\left(\omega_{k}\right)+\frac{1}{\mu_{n}}\int_{-1}^{1}f_{B}\left(t\right)\varepsilon_{\varphi_{n}}\left(t\right)dt,\label{eq:f_Bn-in-f_B(omega_k)}
\end{align}
and substituting (\ref{eq:f_Bn-in-f_B(omega_k)}) in (\ref{eq:hat_f_Bk}),
we obtain 
\begin{align}
\tilde{f}_{B}\left(t\right) & =\frac{1}{B^{2}}\sum_{m=-M}^{M}\alpha_{m}f_{B}\left(\omega_{m}\right)R_{m}\left(t\right)+\sum_{n=0}^{2M}\varphi_{n}\left(t\right)\left(\int_{-\infty}^{\infty}f_{B}\left(\tau\right)\varepsilon_{\varphi_{n}}\left(\tau\right)d\tau\right).\label{eq:hat_f_Bk-in-f_B(omega_k)}
\end{align}
\end{proof}
\begin{cor}
\label{cor:epsilon_f_B} The error $\epsilon_{f_{B}}\left(t\right)$,
for $t\in\left[-1,1\right]$ is bounded by 
\begin{align}
\max_{t\in\left[-1,1\right]}\left|\epsilon_{f_{B}}\left(t\right)\right| & \le C\max_{t\in\left[-1,1\right]}\left|f_{B}\left(t\right)\right|\max_{t\in\left[-2,2\right]}\left|\epsilon_{B}\left(2\pi t\right)\right|
\end{align}
for some constant $C$. 
\end{cor}
\begin{proof}
By (\ref{eq:epsilon_f_B}), for $t\in\left[-1,1\right]$ 
\begin{align}
\max_{t\in\left[-1,1\right]}\left|\epsilon_{f_{B}}\left(t\right)\right| & \le\sum_{n=0}^{2M}\mu_{n}\left|f_{B,n}\right|\left|\varepsilon_{\varphi_{n}}\left(t\right)\right|\nonumber \\
 & \le\sum_{n=0}^{2M}\mu_{n}\left|f_{B,n}\right|2\sqrt{B}\mu_{n}^{-3/2}\max_{t\in\left[-2,2\right]}\left|\epsilon_{B}\left(2\pi t\right)\right|
\end{align}
where we used (\ref{eq:varepsilon_varphi_n(t)}) to write the second
inequality. By (\ref{eq:f_Bn-in-f_B(omega_k)}) and (\ref{eq:varepsilon_varphi_n(t)}),
because $\alpha_{k}\in\mathbb{R}^{+}$, we have 
\begin{align}
\left|f_{B,n}\right| & \le\mu_{n}^{-1}\hspace{-0.2cm}\max_{t\in\left[-1,1\right]}\left|f_{B}\left(t\right)\right|\left[\begin{array}{l}
B^{-1}\sum_{k=-M}^{M}\alpha_{k}\\
+4\sqrt{B}\mu_{n}^{-1/2}\max_{t\in\left[-2,2\right]}\left|\epsilon_{B}\left(2\pi t\right)\right|
\end{array}\right]\nonumber \\
 & =\mu_{n}^{-1}\hspace{-0.2cm}\max_{t\in\left[-1,1\right]}\left|f_{B}\left(t\right)\right|\left[\begin{array}{l}
B^{-1}\left(2B-\epsilon_{B}\left(0\right)\right)\\
+4\sqrt{B}\mu_{n}^{-1/2}\max_{t\in\left[-2,2\right]}\left|\epsilon_{B}\left(2\pi t\right)\right|
\end{array}\right]
\end{align}
where we used the fact that $\sum_{k=-M}^{M}\alpha_{k}=2B-\epsilon_{B}(0)$
by (\ref{eq:sinc-in-exp}), $\varphi_{n}\left(\omega_{m}\right)$
are eigenvector obtained by solving (\ref{eq:varphi_n(omega_m)}),
hence have unit norm, i.e. $\sum_{m}\left|\varphi_{n}\left(\omega_{m}\right)\right|^{2}=1$,
and $\max_{m}\left|\varphi_{n}\left(\omega_{m}\right)\right|\le1$.
Then 
\begin{align}
\max_{t\in\left[-1,1\right]}\left|\epsilon_{f_{B}}\left(t\right)\right| & \le C\max_{t\in\left[-2,2\right]}\left|\epsilon_{B}\left(2\pi t\right)\right|\max_{t\in\left[-1,1\right]}\left|f_{B}\left(t\right)\right|
\end{align}
where the constant $C$ is given by 
\begin{align}
C & =\sum_{n=0}^{2M}\left[\begin{array}{l}
2B^{-1/2}\mu_{n}^{-3/2}\left(2B-\epsilon_{B}\left(0\right)\right)\\
+8B\mu_{n}^{-2}\max_{t\in\left[-2,2\right]}\left|\epsilon_{B}\left(2\pi t\right)\right|
\end{array}\right].
\end{align}
\end{proof}
\begin{example}
Considering the quadratures $\left(\alpha_{m},\omega_{m}\right)=\left(\frac{2B}{2M+1},\frac{2m}{2M+1}\right)_{m=-M}^{M}$
of the discrete inverse Fourier transform for fixed $B$, $\mu_{n}=1$
and, by (\ref{eq:sin(Bx)/sin(x)}) \footnote{$\lim_{x\rightarrow0}\mathrm{sinc}\left(Bx\right)=\lim_{x\rightarrow0}\frac{\sin\left(Bx\right)\left(2M+1\right)^{-1}}{\sin\left(Bx\left(2M+1\right)^{-1}\right)}=1$},
$\epsilon_{B}\left(0\right)=0$. Then, using (\ref{eq:f_B(t)-in-sinc})
and Corollary \ref{cor:epsilon_f_B}, the error between the nodes
is bounded by 
\begin{align}
\left|\epsilon_{f_{B}}\left(t\right)\right| & \le\left[\begin{array}{l}
\max_{t\in\left[-1,1\right]}\left|f_{B}\left(t\right)\right|\,\max_{t\in\left[-2,2\right]}\left|\epsilon_{B}\left(2\pi t\right)\right|\\
\times\sum_{n=0}^{2M}\left[4B^{1/2}+8B\max_{t\in\left[-2,2\right]}\left|\epsilon_{B}\left(2\pi t\right)\right|\right]
\end{array}\right]\nonumber \\
 & =\left[\begin{array}{l}
\max_{t\in\left[-1,1\right]}\left|f_{B}\left(t\right)\right|\,\max_{t\in\left[-2,2\right]}\left|\epsilon_{B}\left(2\pi t\right)\right|\\
\times\left(2M+1\right)\left[4B^{1/2}+8B\max_{t\in\left[-2,2\right]}\left|\epsilon_{B}\left(2\pi t\right)\right|\right]
\end{array}\right]
\end{align}
for $t\in\left[-1,1\right]$. By (\ref{eq:eps_B(x)_O(2N+1)^(-2)}),
one can achieve $\max_{t\in\left[-2,2\right]}\left|\epsilon_{B}\left(2\pi t\right)\right|=\mathcal{O}\left(\left(2M+1\right)^{-2}\right)$
consequently $\left|\epsilon_{f_{B}}\left(t\right)\right|\le\max_{t\in\left[-1,1\right]}\left|f_{B}\left(t\right)\right|\mathcal{O}\left(\left(2M+1\right)^{-1}\right)$
which is in the order of the truncation errors presented in Section
VI of \cite{jerri1977shannon}. 
\end{example}

\section{$R$-limited functions\label{sec:R-limited-functions}}

In this section we introduce an equivalent of $R$-limited functions
with respect to a linear transformation and $R$-Slepian functions
which are multivariate generalization of band-limited functions and
prolate spheroidal wave functions, respectively. Then we prove the
generalizations of Theorems \ref{thm:discrete_Fourier_approx_bandlimited_functions}
and \ref{thm:sampling-theorem-bandlimited-functions} to $R$-limited
functions.

Let $GL\left(\mathbb{R},N\right)$ denote the general linear group,
the set of invertible matrices in $\mathbb{R}^{N\times N}$, and 
\begin{align*}
R_{A} & =\left\{ \mathbf{k}=A\mathbf{x}\,|\,\mathbf{x}\in R\subset\mathbb{R}^{N},\,A\in GL\left(\mathbb{R},N\right)\right\} 
\end{align*}
for some compact $R\subset\mathbb{R}^{N}$. Employing the terminology
introduced in \cite{slepian1964prolate}, we define $R_{B}$-limited
functions by 
\begin{align}
f_{B}\left(\mathbf{x}\right) & =\int_{R_{B}}\hat{f}_{B}\left(\mathbf{k}\right)\mathrm{e}^{\mathrm{i}2\pi\mathbf{k}\cdot\mathbf{x}}d\mathbf{k}\label{eq:R-limited-definition}
\end{align}
where $B\in GL\left(\mathbb{R},N\right)$ is a real symmetric matrix
and 
\begin{align}
\hat{f}_{B}\left(\mathbf{k}\right) & =\int_{\mathbb{R}^{N}}f_{B}\left(\mathbf{x}\right)\mathrm{e}^{-\mathrm{i}2\pi\mathbf{k}\cdot\mathbf{x}}d\mathbf{x}.
\end{align}
Here $B$ is a multidimensional analogue of band-limit. When $B$
is the identity matrix, $R_{I}=R$, one obtains definition of $R$-limited
functions of \cite{slepian1964prolate}.

Alternatively, we can write 
\begin{align}
f_{B}\left(\mathbf{x}\right) & =\int_{R_{B}}\int_{\mathbb{R}^{N}}f_{B}\left(\mathbf{y}\right)\mathrm{e}^{\mathrm{i}2\pi\mathbf{k}\cdot\left(\mathbf{x}-\mathbf{y}\right)}d\mathbf{y}d\mathbf{k}\nonumber \\
 & =\int_{\mathbb{R}^{N}}f_{B}\left(\mathbf{y}\right)\left[\left|\det\left(B\right)\right|\,\int_{R}\mathrm{e}^{\mathrm{i}2\pi B\mathbf{k}\cdot\left(\mathbf{x}-\mathbf{y}\right)}d\mathbf{k}\right]d\mathbf{y}\nonumber \\
 & =\int_{\mathbb{R}^{N}}f_{B}\left(\mathbf{y}\right)\left|\det\left(B\right)\right|K\left(2\pi B\left(\mathbf{x}-\mathbf{y}\right)\right)d\mathbf{y}\label{eq:R-limited-convolution-representation}
\end{align}
where $\det\left(B\right)$ denotes the determinant of $B$ and 
\begin{align}
K\left(\mathbf{x}\right) & =\int_{R}\mathrm{e}^{\mathrm{i}\mathbf{k}\cdot\mathbf{x}}d\mathbf{k}.\label{eq:K-in-exp}
\end{align}
Given a function $f\left(\mathbf{x}\right)$, its \emph{$R_{B}$-limited
projection} $P_{B}\left[f\right]\left(\mathbf{x}\right)$, denoted
by $f_{B}\left(\mathbf{x}\right)$ for short, is defined by 
\begin{align}
P_{B}\left[f\right]\left(\mathbf{x}\right)=f_{B}\left(\mathbf{x}\right) & =\int_{R_{B}}\hat{f}\left(\mathbf{k}\right)\mathrm{e}^{\mathrm{i}2\pi\mathbf{k}\cdot\mathbf{x}}d\mathbf{k}
\end{align}
or, equivalently, in the convolution representation using convolution
theorem 
\begin{align}
P_{B}\left[f\right]\left(\mathbf{x}\right) & =\int_{\mathbb{R}^{N}}f\left(\mathbf{y}\right)\left|\det\left(B\right)\right|K\left(2\pi B\left(\mathbf{x}-\mathbf{y}\right)\right)d\mathbf{y}.\label{eq:R-limited-projection_of_f}
\end{align}
Note that $\hat{f}_{B}\left(\mathbf{k}\right)=\hat{f}\left(\mathbf{k}\right)$,
for $\mathbf{k}\in R_{B}$.

\subsection{Discrete Fourier representation of $R$-limited approximation of
compactly supported functions}
\begin{thm}
\label{thm:discrete-Fourier-approximation-R-limited}Consider a discretization
of the integral representation (\ref{eq:K-in-exp}) of $K\left(\mathbf{x}\right)$
\begin{align}
\left|\det\left(B\right)\right|K\left(B\mathbf{x}\right) & =\sum_{m}\alpha_{m}\mathrm{e}^{\mathrm{i}\mathbf{k}_{m}\cdot B\mathbf{x}}+\epsilon_{K}\left(\mathbf{x}\right)\label{eq:approx_K}
\end{align}
for $\mathbf{k}_{m}\in R$. For a function $f\left(\mathbf{x}\right)$
compactly supported within a region $X\subset\mathbb{R}^{n}$, its
R-limited projection can be approximated by 
\begin{align}
f_{B}\left(\mathbf{x}\right) & =\sum_{m}\alpha_{m}\mathrm{e}^{\mathrm{i}2\pi B\mathbf{k}_{m}\cdot\mathbf{x}}\hat{f}\left(B\mathbf{k}_{m}\right)+\epsilon_{f}\left(\mathbf{x}\right)\label{eq:Discrete-Fourier-approx-R-limited-f}
\end{align}
and (\ref{eq:Discrete-Fourier-approx-R-limited-f}) provides a discretization
of (\ref{eq:R-limited-definition}) with 
\begin{align}
\max_{\mathbf{x}\in X}\left|\epsilon_{f}\left(\mathbf{x}\right)\right| & \le\left|X\right|\max_{\mathbf{x}\in X}\left|f\left(\mathbf{x}\right)\right|\max_{\mathbf{x}\in X+X}\left|\epsilon_{K}\left(2\pi\mathbf{x}\right)\right|,\label{eq:error-discrete-fourier-approx-R-limited}
\end{align}
where $X+X=\left\{ \mathbf{x}|\mathbf{x}=\mathbf{x}_{1}+\mathbf{x}_{2},\,\mathbf{x}_{1},\mathbf{x}_{2}\in X\right\} $
and $\left|X\right|=\int_{X}d\mathbf{x}$. 
\end{thm}
\begin{proof}
Substituting (\ref{eq:approx_K}) into (\ref{eq:R-limited-projection_of_f}),
we obtain (\ref{eq:Discrete-Fourier-approx-R-limited-f}) where 
\begin{align}
\epsilon_{f}\left(\mathbf{x}\right) & =\int_{X}f\left(\mathbf{y}\right)\epsilon_{K}\left(2\pi\left[\mathbf{x}-\mathbf{y}\right]\right)d\mathbf{y}.
\end{align}
\end{proof}
For $f_{B}\left(\mathbf{x}\right)$ to accurately approximate the
compactly supported function $f\left(\mathbf{x}\right)$ over $\mathbf{x}\in X$,
by (\ref{eq:error-discrete-fourier-approx-R-limited}), one needs
to build up an approximation of $K\left(B\mathbf{x}\right)$ that
is accurate over the set $2\pi\left(X+X\right)=\left\{ \mathbf{y}|\mathbf{y}=2\pi\mathbf{x},\,\mathbf{x}\in X+X\right\} $.

\subsection{$R$-Slepian functions: A multivariate generalization of PSWF}

We can generalize the PSWF and their approximations presented in Sections
\ref{subsec:approx-PSWF-1} and \ref{subsec:Approximating-PSWF-from-shifts-of-sinc}
to multiple-variables to define $R$-Slepian functions and construct
their approximations.

\subsubsection{$R$-Slepian functions}

Consider the $R_{B}$-limited projection of a compactly supported
function with support $S\in\mathbb{R}^{N}$: 
\begin{align}
P_{B}\left[f\right]\left(\mathbf{x}\right)=f_{B}\left(\mathbf{x}\right) & =\int_{S}f\left(\mathbf{y}\right)\left|\det\left(B\right)\right|K\left(2\pi B\left(\mathbf{x}-\mathbf{y}\right)\right)d\mathbf{y}\label{eq:P_B=00003D00005Bf=00003D00005D-R-limited}
\end{align}
restricted to $\mathbf{x}\in S$. $P_{B}\left[f\right]\left(\mathbf{x}\right)$
is a positive definite operator. Furthermore, if $R$ is symmetric,
i.e. $R=-R=\left\{ \mathbf{-x}\,|\,\mathbf{x}\in R\right\} $, then
$K\left(\mathbf{x}\right)=K\left(-\mathbf{x}\right)$ is real, 
\begin{align}
K\left(\mathbf{x}\right) & =\int_{R}\cos\left(\mathbf{k}\cdot\mathbf{x}\right)d\mathbf{k},\label{eq:K-in-cos}
\end{align}
$P_{B}$ is a positive definite real symmetric operator, 
\begin{align}
 & \int_{S}\int_{S}f\left(\mathbf{y}\right)\left|\det\left(B\right)\right|K\left(2\pi B\left(\mathbf{x}-\mathbf{y}\right)\right)d\mathbf{y}\,f\left(\mathbf{x}\right)d\mathbf{x}\nonumber \\
 & =\int_{S}\int_{S}f\left(\mathbf{y}\right)\left|\det\left(B\right)\right|K\left(2\pi B\left(\mathbf{y}-\mathbf{x}\right)\right)d\mathbf{y}\,f\left(\mathbf{x}\right)d\mathbf{x}\nonumber \\
 & =\int_{R_{B}}\left|\hat{f}\left(\mathbf{k}\right)\right|^{2}d\mathbf{k}\ge0,
\end{align}
and, consequently, by Mercer's theorem (see page 245 \cite{riesz1990functional}),
accepts a discrete eigendecomposition 
\begin{align}
\mu_{n}\varphi_{n}\left(\mathbf{x}\right) & =\int_{S}\varphi_{n}\left(\mathbf{y}\right)\left|\det\left(B\right)\right|K\left(2\pi B\left(\mathbf{x}-\mathbf{y}\right)\right)d\mathbf{y}\label{eq:K-prolates}
\end{align}
with positive eigenvalues $\mu_{n}$ and real eigenfunctions $\varphi_{n}\left(\mathbf{x}\right)\in\mathbb{R}$
for $\mathbf{x}\in S$. We refer to eigenfunctions $\varphi_{n}\left(\mathbf{x}\right)$
as the $R$-Slepian functions. For the sake of simplicity of the discussion
we will consider symmetric $R$. The case of non-symmetric $R$ can
be reduced to the symmetric case by translation of $R$ away from
the origin to exclude origin and consider $R\cup-R$.

Consider, $S=R$ and solutions $\psi\left(\mathbf{x}\right)$ of the
equation 
\begin{align}
\lambda\psi\left(\mathbf{x}\right) & =\int_{R}\psi\left(\mathbf{k}\right)\mathrm{e}^{\mathrm{i}2\pi B\mathbf{k}\cdot\mathbf{x}}d\mathbf{k},\quad\mathbf{x}\in R\label{eq:R-prolates}
\end{align}
Define $\psi_{e}\left(\mathbf{x}\right)=\left[\psi\left(\mathbf{x}\right)+\psi\left(-\mathbf{x}\right)\right]/2$
and $\psi_{o}\left(\mathbf{x}\right)=\left[\psi\left(\mathbf{x}\right)-\psi\left(-\mathbf{x}\right)\right]/2$
as the even and odd parts of $\psi\left(\mathbf{x}\right)$. For symmetric
$R$, i.e. $R=-R$, we have 
\begin{align}
\lambda\left[\psi_{e}\left(\mathbf{x}\right)+\psi_{o}\left(\mathbf{x}\right)\right] & =\int_{R}\psi\left(\mathbf{k}\right)\mathrm{e}^{\mathrm{i}2\pi B\mathbf{k}\cdot\mathbf{x}}d\mathbf{k}\nonumber \\
 & =\left\{ \begin{array}{l}
\int_{R}\psi_{e}\left(\mathbf{k}\right)\cos\left(2\pi B\mathbf{k}\cdot\mathbf{x}\right)d\mathbf{k}\\
+\mathrm{i}\int_{R}\psi_{o}\left(\mathbf{k}\right)\sin\left(2\pi B\mathbf{k}\cdot\mathbf{x}\right)d\mathbf{k}
\end{array}\right\} 
\end{align}
Considering the equations 
\begin{align}
\beta_{e}\psi_{e}\left(\mathbf{x}\right) & =\int_{R}\psi_{e}\left(\mathbf{k}\right)\cos\left(2\pi B\mathbf{k}\cdot\mathbf{x}\right)d\mathbf{k}\label{eq:psi_e}\\
\beta_{o}\psi_{o}\left(\mathbf{x}\right) & =\int_{R}\psi_{o}\left(\mathbf{k}\right)\sin\left(2\pi B\mathbf{k}\cdot\mathbf{x}\right)d\mathbf{k},\label{eq:psi_o}
\end{align}
which have real symmetric kernel with real eigenvalues and eigenfunctions,
the real and imaginary eigenvalues of (\ref{eq:R-prolates}) are associated
with eigenfunctions of (\ref{eq:psi_e}) and (\ref{eq:psi_o}), respectively.
Completeness follow from Fourier theory using the same arguments as
in \cite{slepian1964prolate,riesz1990functional}.

Eigenfunctions $\varphi_{n}\left(\mathbf{x}\right)$ of (\ref{eq:R-prolates}),
\begin{align}
\lambda_{n}\varphi_{n}\left(\mathbf{x}\right) & =\int_{R}\varphi_{n}\left(\mathbf{k}\right)\mathrm{e}^{\mathrm{i}2\pi B\mathbf{k}\cdot\mathbf{x}}d\mathbf{k},\quad\mathbf{x}\in R,\label{eq:R-prolates-1}
\end{align}
are also eigenfunctions of (\ref{eq:K-prolates}): 
\begin{align}
 & \int_{R}\varphi_{n}\left(\mathbf{y}\right)\left|\det\left(B\right)\right|K\left(2\pi B\left(\mathbf{x}-\mathbf{y}\right)\right)d\mathbf{y}\nonumber \\
 & =\left|\det\left(B\right)\right|\int_{R}\mathrm{e}^{\mathrm{i}2\pi B\mathbf{x}\cdot\mathbf{k}}\left[\int_{R}\varphi_{n}\left(\mathbf{y}\right)\mathrm{e}^{-\mathrm{i}2\pi B\mathbf{y}\cdot\mathbf{k}}d\mathbf{y}\right]d\mathbf{k}\nonumber \\
 & =\overline{\lambda}_{n}\left|\det\left(B\right)\right|\int_{R}\mathrm{e}^{\mathrm{i}2\pi B\mathbf{x}\cdot\mathbf{k}}\varphi_{n}\left(\mathbf{k}\right)d\mathbf{k}\nonumber \\
 & =\left|\lambda_{n}\right|^{2}\left|\det\left(B\right)\right|\varphi_{n}\left(\mathbf{x}\right)
\end{align}
with 
\begin{align}
\mu_{n} & =\left|\det\left(B\right)\right|\left|\lambda_{n}\right|^{2}.\label{eq:mu_n-lambda_n-4-R-Slepian}
\end{align}
Similar to the case of PSWF, $R$-Slepian functions satisfy double
orthogonality relation \cite{slepian1964prolate} 
\begin{align}
\mu_{n}\int_{\mathbb{R}^{N}}\varphi_{n}\left(\mathbf{x}\right)\varphi_{m}\left(\mathbf{x}\right)d\mathbf{x} & =\int_{R}\varphi_{n}\left(\mathbf{x}\right)\varphi_{m}\left(\mathbf{x}\right)d\mathbf{x}=\delta_{m,n}\label{eq:R-Slepian-orthogonality}
\end{align}

\subsubsection{Approximating $R$-Slepian functions as eigenfunction of $\mathrm{e}^{\mathrm{i}2\pi B\mathbf{k}_{m}\cdot\mathbf{x}_{l}}$}

Similar to the single dimensional case, if one can build up an approximation
to (\ref{eq:K-in-cos}) or, equivalently, (\ref{eq:K-in-exp}), 
\begin{align}
K\left(\mathbf{x}\right) & =\sum_{m=1}^{M}\alpha_{m}\exp\left(\mathrm{i}\mathbf{k}_{m}\cdot\mathbf{x}\right)+\epsilon_{K}\left(\mathbf{x}\right)\label{eq:K-in-cos-1}
\end{align}
for some $\mathbf{k}_{m}\in R$, then, substituting (\ref{eq:K-in-cos-1})
in (\ref{eq:K-prolates}), we obtain 
\begin{align}
\lambda_{n}\varphi_{n}\left(\mathbf{x}\right) & =\sum_{m=1}^{M}\alpha_{m}\mathrm{e}^{\mathrm{i}2\pi B\mathbf{k}_{m}\cdot\mathbf{x}}\varphi_{n}\left(\mathbf{k}_{m}\right)+\epsilon_{\varphi_{n}}\left(\mathbf{x}\right)\label{eq:approx-R-prolates}
\end{align}
where 
\begin{align}
\epsilon_{\varphi_{n}}\left(\mathbf{x}\right) & =\int_{R}\varphi_{n}\left(\mathbf{y}\right)\epsilon_{K}\left(2\pi B\left[\mathbf{x}-\mathbf{y}\right]\right).
\end{align}
We can approximate $R$-Slepian functions by substituting the eigenvectors
$\varphi_{n}\left(\mathbf{x}_{l}\right)$ of the equation 
\begin{align}
\lambda_{n}\,\varphi_{n}\left(\mathbf{x}_{l}\right) & =\sum_{m=1}^{M}\alpha_{m}\mathrm{e}^{\mathrm{i}2\pi B\mathbf{k}_{m}\cdot\mathbf{x}_{l}}\varphi_{n}\left(\mathbf{k}_{m}\right),\quad\mathbf{x}_{l}\in\bigcup_{m}\left\{ \mathbf{k}_{m}\right\} 
\end{align}
into (\ref{eq:approx-R-prolates}) for $n,l=1,\ldots,M$.

Because there are $\left|\det\left(B\right)\right|\left|S\right|\left|R\right|$
number of eigenvalues $\mu_{n}$ close to one (see Theorem 3 in \cite{landau1975szego}\footnote{Recently this theorem is rediscovered in \cite{franceschetti2015landau}.}),
in order to capture all the eigenvalues close to one, one shall have
$M\ge\left\lceil \left|\det\left(B\right)\right|\left|S\right|\left|R\right|\right\rceil $.
A detailed analysis of the characterization of the eigenvalues of
the projection operator defined in (\ref{eq:P_B=00003D00005Bf=00003D00005D-R-limited})
around one, zero and the transition zone can be found in \cite{sobolev2013pseudo}.

\subsubsection{Approximating $R$-Slepian functions as eigenfunction of $K\left(2\pi B\left[\mathbf{k}_{m}-\mathbf{k}_{l}\right]\right)$
\label{subsec:Approximating--Slepian-functions-from-shifts-of-K}}

Multiplying both sides of (\ref{eq:approx-R-prolates}) by $\left|\det\left(B\right)\right|\mathrm{e}^{-\mathrm{i}2\pi B\mathbf{k}\cdot\mathbf{x}}$,
integrating over $\mathbf{x}$ and using (\ref{eq:mu_n-lambda_n-4-R-Slepian}),
we obtain 
\begin{align}
\left|\lambda_{n}\right|^{2}\left|\det\left(B\right)\right|\varphi_{n}\left(\mathbf{k}\right) & =\mu_{n}\varphi_{n}\left(\mathbf{k}\right)\nonumber \\
 & =\sum_{m=1}^{M}\alpha_{m}\left|\det\left(B\right)\right|K\left(2\pi B\left[\mathbf{k}-\mathbf{k}_{m}\right]\right)\varphi_{n}\left(\mathbf{k}_{m}\right)+\varepsilon_{\varphi_{n}}\left(\mathbf{k}\right)\label{eq:sinc_approx_Slepian_func}
\end{align}
where 
\begin{align}
\varepsilon_{\varphi_{n}}\left(\mathbf{k}\right) & =\int_{R}\epsilon_{\varphi_{n}}\left(\mathbf{x}\right)\mathrm{e}^{-\mathrm{i}2\pi B\mathbf{k}\cdot\mathbf{x}}\,d\mathbf{x}
\end{align}
Consequently, by (\ref{eq:sinc_approx_Slepian_func}), we can approximate
$R$-Slepian functions by substituting the eigenvectors $\varphi_{n}\left(\mathbf{k}_{m}\right)$
of the equation 
\begin{align}
\mu_{n}\varphi_{n}\left(\mathbf{k}_{l}\right) & =\sum_{m=1}^{M}\alpha_{m}\left|\det\left(B\right)\right|K\left(2\pi B\left[\mathbf{k}_{l}-\mathbf{k}_{m}\right]\right)\varphi_{n}\left(\mathbf{k}_{m}\right)
\end{align}
for $n=1,\ldots M$.

\subsection{Discrete convolution representation of $R$-limited approximation
of compactly supported functions}

Using $R$-Slepian functions, we can prove a generalization of the
sampling and interpolation theorem, Theorem \ref{thm:sampling-theorem-bandlimited-functions},
for $R_{A}$-limited functions. To do this, we first show the sampling
theorem of $R_{B}$-limited functions for symmetric $B$ and then
generalize it to $R_{A}$-functions for an arbitrary $A\in GL\left(N,\mathbb{R}\right)$. 
\begin{lem}
Given a symmetric $B\in GL\left(N,\mathbb{R}\right)$, i.e. $B=B^{T}$
and a function $f\left(\mathbf{x}\right)$, $R_{B}$-limited projection
$f_{B}\left(\mathbf{x}\right)$ of $f\left(\mathbf{x}\right)$ can
be approximated by 
\begin{align}
f_{B}\left(\mathbf{x}\right) & \approx\sum_{m=1}^{M}f_{m}\left|\det\left(B\right)\right|K\left(2\pi B\left(\mathbf{x}-\mathbf{k}_{m}\right)\right)
\end{align}
where 
\begin{align}
f_{m} & =\sum_{n=1}^{M}f_{B}\left(\mathbf{k}_{n}\right)\alpha_{n}\left|\det\left(B\right)\right|K\left(2\pi B\left(\mathbf{k}_{m}-\mathbf{k}_{n}\right)\right)\alpha_{m}.
\end{align}
\end{lem}
\begin{proof}
Consider a symmetric $B\in GL\left(N,\mathbb{R}\right)$. Because
$\varphi_{n}\left(\mathbf{x}\right)$ are complete for compactly supported
functions as well as $R_{B}$-limited functions, we can expand any
$R_{B}$-limited function $f_{B}$ using (\ref{eq:sinc_approx_Slepian_func})
as follows: 
\begin{align}
f_{B}\left(\mathbf{x}\right) & =\int_{R}f\left(\mathbf{x}\right)\left|\det\left(B\right)\right|K\left(2\pi B\left(\mathbf{x}-\mathbf{y}\right)\right)\nonumber \\
 & =\sum_{n=1}^{M}\mu_{n}f_{B,n}\varphi_{n}\left(\mathbf{x}\right)\nonumber \\
 & =\left\{ \hspace{-0.5cc}\begin{array}{l}
\sum_{m=1}^{M}\alpha_{m}\left|\det\left(B\right)\right|K\left(2\pi B\left(\mathbf{x}-\mathbf{k}_{m}\right)\right)\tilde{f}_{B}\left(\mathbf{k}_{m}\right)\\
+\sum_{n=1}^{M}\mu_{n}f_{B,n}\varepsilon_{\varphi_{n}}\left(\mathbf{x}\right)
\end{array}\hspace{-0.5cc}\right\} \label{eq:f_B(t)-in-K}
\end{align}
where 
\begin{align}
\tilde{f}_{B}\left(\mathbf{k}\right) & =\sum_{n=1}^{M}f_{B,n}\varphi_{n}\left(\mathbf{k}\right)\label{eq:hat_f_Bk-1}
\end{align}
(\ref{eq:hat_f_Bk-1}) is a multivariate generalization of (\ref{eq:f_B-in-PSWF})

We rewrite $f_{B,n}$ as 
\begin{align}
f_{B,n} & =\frac{1}{\mu_{n}}\sum_{m=1}^{M}\alpha_{m}\varphi_{n}\left(\mathbf{k}_{m}\right)f_{B}\left(\mathbf{k}_{m}\right)+\int_{\mathbb{R}^{N}}f_{B}\left(\mathbf{x}\right)\varepsilon_{\varphi_{n}}\left(\mathbf{x}\right)d\mathbf{x}.\label{eq:f_B,n_in-K}
\end{align}
Substituting (\ref{eq:f_B,n_in-K}) in (\ref{eq:hat_f_Bk-1}), we
have 
\begin{align}
\tilde{f}_{B}\left(\mathbf{k}\right) & =\sum_{m=1}^{M}\alpha_{m}f_{B}\left(\mathbf{k}_{m}\right)R_{m}\left(\mathbf{k}\right)+\left(\sum_{n=1}^{M}\varphi_{n}\left(\mathbf{k}\right)\int_{\mathbb{R}^{N}}f_{B}\left(\mathbf{x}\right)\varepsilon_{\varphi_{n}}\left(\mathbf{x}\right)d\mathbf{x}\right)
\end{align}
where 
\begin{align}
R_{m}\left(\mathbf{k}\right) & =\sum_{n=1}^{M}\frac{1}{\mu_{n}}\varphi_{n}\left(\mathbf{k}_{m}\right)\varphi_{n}\left(\mathbf{k}\right).
\end{align}
Similar to the single dimensional case, a regularized approximation
to $R_{m}\left(\mathbf{k}\right)$ is given by $\left|\det\left(B\right)\right|K\left(2\pi B\left(\mathbf{k}-\mathbf{k}_{m}\right)\right)$,
leading to the approximate interpolation formula 
\begin{align}
f_{B}\left(\mathbf{x}\right) & \approx\sum_{m=1}^{M}\left\{ \begin{array}{l}
\alpha_{m}\left|\det\left(B\right)\right|K\left(2\pi B\left(\mathbf{x}-\mathbf{k}_{m}\right)\right)\\
\times\left[\sum_{n=1}^{M}\alpha_{n}f_{B}\left(\mathbf{k}_{n}\right)\left|\det\left(B\right)\right|\,K\left(2\pi B\left(\mathbf{k}_{m}-\mathbf{k}_{n}\right)\right)\right]
\end{array}\right\} 
\end{align}
\end{proof}
\begin{thm}
\label{thm:sampling-theorem-R_A-limited-functions} Consider the $R_{A}$-limited
function \textup{ 
\begin{align}
f_{A}\left(\mathbf{x}\right) & =\int_{AR}f\left(\mathbf{y}\right)\left|\det\left(A\right)\right|K\left(2\pi A^{T}\left(\mathbf{x}-\mathbf{y}\right)\right)d\mathbf{y}
\end{align}
for some, not necessarily symmetric, $A\in GL\left(N,\mathbb{R}\right)$.
Then 
\begin{multline}
f_{A}\left(\mathbf{x}\right)=\sum_{m}\alpha_{m}\left|\det\left(B\right)\right|K\left(2\pi A^{T}\left(\mathbf{x}-A\mathbf{k}_{m}\right)\right)\tilde{f}_{A}\left(A\mathbf{k}_{m}\right)\\
+\sum_{n}\mu_{n}g_{B,n}\varepsilon_{\varphi_{n}}\left(A^{-1}\mathbf{x}\right),\quad\mathbf{x}\in\mathbb{R}^{N},\,A\mathbf{k}_{m}\in R_{A}
\end{multline}
where }$\tilde{f}_{A}\left(A\mathbf{k}_{m}\right)$ and $g_{B,n}$
are defined by (\ref{eq:tilde_f_A}) and (\ref{eq:g_B,n}), respectively. 
\end{thm}
\begin{proof}
Let $B=A^{T}A$. Then $g_{B}\left(\mathbf{x}\right)=f_{A}\left(A\mathbf{x}\right)$
is an $R_{B}$-limited projection of $g\left(\mathbf{x}\right)=f\left(A\mathbf{x}\right)$:
\begin{align}
f_{A}\left(A\mathbf{x}\right) & =\int_{R}f\left(A\mathbf{y}\right)\left|\det\left(A\right)\right|^{2}K\left(2\pi A^{T}\left(A\mathbf{x}-A\mathbf{y}\right)\right)d\mathbf{y}\nonumber \\
 & =\int_{R}f\left(A\mathbf{y}\right)\left|\det\left(B\right)\right|K\left(2\pi B\left(\mathbf{x}-\mathbf{y}\right)\right)d\mathbf{y}\nonumber \\
g_{B}\left(\mathbf{x}\right) & =\int_{R}g\left(\mathbf{y}\right)\left|\det\left(B\right)\right|K\left(2\pi B\left(\mathbf{x}-\mathbf{y}\right)\right)d\mathbf{y}
\end{align}
which, by (\ref{eq:f_B(t)-in-K}), can be approximated by 
\begin{multline}
g_{B}\left(\mathbf{x}\right)=\sum_{m}\alpha_{m}\left|\det\left(B\right)\right|K\left(2\pi B\left(\mathbf{x}-\mathbf{k}_{m}\right)\right)\tilde{g}_{B}\left(\mathbf{k}_{m}\right)\\
+\sum_{n}\mu_{n}g_{B,n}\varepsilon_{\varphi_{n}}\left(\mathbf{x}\right),\quad\mathbf{x}\in\mathbb{R}^{N},\,\mathbf{k}_{m}\in R.
\end{multline}
or equivalently 
\begin{multline}
f_{A}\left(\mathbf{x}\right)=\sum_{m}\alpha_{m}\left|\det\left(B\right)\right|K\left(2\pi A^{T}\left(\mathbf{x}-A\mathbf{k}_{m}\right)\right)\tilde{f}_{A}\left(A\mathbf{k}_{m}\right)\\
+\sum_{n}\mu_{n}g_{B,n}\varepsilon_{\varphi_{n}}\left(A^{-1}\mathbf{x}\right),\quad\mathbf{x}\in\mathbb{R}^{N},\,A\mathbf{k}_{m}\in R_{A}.
\end{multline}
Here 
\begin{align}
g_{B,n} & =\int_{\mathbb{R}^{N}}f_{A}\left(A\mathbf{x}\right)\varphi_{n}\left(\mathbf{x}\right)d\mathbf{x}\nonumber \\
 & =\frac{1}{\det\left(A\right)}\int_{\mathbb{R}^{N}}f_{A}\left(\mathbf{x}\right)\varphi_{n}\left(A^{-1}\mathbf{x}\right)d\mathbf{x}\nonumber \\
 & =\int_{R}f_{A}\left(A\mathbf{x}\right)\varphi_{n}\left(\mathbf{x}\right)d\mathbf{x}\nonumber \\
 & =\frac{1}{\mu_{n}}\left\{ \begin{array}{l}
\sum_{m=1}^{M}\alpha_{m}f_{A}\left(A\mathbf{k}_{m}\right)\varphi_{n}\left(\mathbf{k}_{m}\right)\\
+\int_{R}f_{A}\left(A\mathbf{x}\right)\varepsilon_{\varphi_{n}}\left(\mathbf{x}\right)d\mathbf{x}
\end{array}\right\} ,\label{eq:g_B,n}
\end{align}
and 
\begin{align}
\tilde{f}_{A}\left(A\mathbf{k}\right)=\tilde{g}_{B}\left(\mathbf{k}\right) & =\sum_{n}g_{B,n}\varphi_{n}\left(\mathbf{k}\right)\label{eq:tilde_f_A}
\end{align}
with $\varphi_{n}\left(\mathbf{x}\right)$ being the eigenvector of
the projection operator $P_{B}$ with respect to the symmetric matrix
$B$. 
\end{proof}

\subsection{Construction of discrete Fourier approximation of the kernel \eqref{eq:approx_K}}

It is important to note that both Theorems \ref{thm:discrete-Fourier-approximation-R-limited}
and \ref{thm:sampling-theorem-R_A-limited-functions} rely on finding
an approximation of the convolution kernel in the form of (\ref{eq:K-in-cos-1}).
Although there is no unique way of finding an approximation in the
form of (\ref{eq:K-in-cos-1}), it can be constructed using tools
from approximation theory. Considering that two and three dimensional
domains can be approximated using tetrahedral and triangular meshes
along with their multidimensional extensions \cite{cheng2012delaunay,lo2014finite},
it is necessary to build quadrature $\left(\alpha_{m},\mathbf{k}_{m}\right)$
to approximate triangle-limited and tetrahedral-limited (shortly $T$-limited)
functions. Because our results do not require $R_{A}$ to be connected,
and they can be generalized to $R_{A}=\cup_{l=1}^{L}A_{l}R_{l}$ where
$A_{l}\in GL\left(\mathbb{R}^{N}\right)$ and $R_{l}\subset\mathbb{R}^{N}$
such that intersection of $\left\{ A_{l}R_{l}\right\} _{l=1}^{L}$
has measure zero, i.e. $\left|\cap_{l=1}^{L}A_{l}R_{l}\right|=0$,
quadratures obtained for $T$-limited functions can be patched together
to construct an approximation of the form (\ref{eq:K-in-cos-1}).
In this case, the convolution kernel becomes 
\begin{align}
K_{\Sigma}\left(\mathbf{x}\right) & =\sum_{m=1}^{M}\left|\det\left(A_{m}\right)\right|K_{m}\left(2\pi A_{m}^{T}\mathbf{x}\right)
\end{align}
where $K_{m}\left(\mathbf{x}\right)=\int_{R_{m}}\mathrm{e}^{\mathrm{i}\mathbf{k}\cdot\mathbf{x}}d\mathbf{k}$
and, consequently, the $R_{A}$-limited projection of a function is
given by 
\begin{align}
f_{A}\left(\mathbf{x}\right) & =\int_{\mathbb{R}^{N}}f\left(\mathbf{y}\right)K_{\Sigma}\left(\mathbf{x}-\mathbf{y}\right)d\mathbf{y}=\sum_{l=1}^{L}f_{A_{l}}\left(\mathbf{x}\right)
\end{align}
where 
\begin{align}
f_{A_{l}}\left(\mathbf{x}\right) & =\int_{\mathbb{R}^{N}}f\left(\mathbf{y}\right)\left|\det\left(A_{l}\right)\right|K\left(2\pi A_{l}^{T}\left(\mathbf{x}-\mathbf{y}\right)\right)d\mathbf{y}.
\end{align}
By Corollary \ref{thm:sampling-theorem-R_A-limited-functions}, because
each $f_{A_{l}}\left(\mathbf{x}\right)$ can be approximated by 
\begin{multline}
f_{A_{l}}\left(\mathbf{x}\right)\approx\sum_{m_{l}}\alpha_{m_{l}}\left|\det\left(B_{l}\right)\right|K\left(2\pi A_{l}^{T}\left(\mathbf{x}-A\mathbf{k}_{m_{l},l}\right)\right)\tilde{g}_{B_{l}}\left(\mathbf{k}_{m_{l},l}\right)\\
+\sum_{m_{l}}\mu_{m_{l}}g_{B,m_{l}}\varepsilon_{\varphi_{m_{l}}}\left(A_{l}^{-1}\mathbf{x}\right),\quad\mathbf{x}\in\mathbb{R}^{N},\,A_{l}\mathbf{k}_{m_{l},l}\in A_{l}R_{l}
\end{multline}
then $f_{A}\left(\mathbf{x}\right)$ can be approximated using samples
of $f\left(\mathbf{x}\right)$ for $\mathbf{x}\in\left\{ A_{l}\mathbf{k}_{m_{l},l}\right\} _{\begin{subarray}{c}
l=1,\ldots,L\\
m_{l}=1,\ldots,M_{l}
\end{subarray}}$.

In Appendix \ref{sec:Kernel-for-T-bandlimited-functions}, we provide
a method to construct quadratures for $T$-limited functions, specifically
for isosceles triangle and trirectangular tetrahedron which are used
to construct quadratures for equilateral triangle and regular tetrahedron,
respectively, that satisfy the corresponding symmetry properties.
We present two ways to construct the quadrature for equilateral triangle,
one capturing the symmetries of the triangle and the other doesn't.
Although approximations are constructed to capture the behavior of
the kernel and its derivatives at zero, the quadrature that satisfy
the symmetry properties provide a more accurate approximation within
a larger vicinity of zero with fewer number of nodes.

Another special case of $R$-limited functions that have practical
importance in multidimensional signal processing seismic data is considered
in Appendix \ref{sec:Cone-limited-functions} which can also be extended
to image processing in 2D and video processing 3D. We present the
convolution kernels and construction of corresponding quadratures
that can be used in sampling and interpolation Theorems \ref{thm:discrete-Fourier-approximation-R-limited}
and \ref{thm:sampling-theorem-R_A-limited-functions} in Appendix
\ref{sec:Cone-limited-functions}.

\section{Conclusion \label{sec:Conclusion-and-discussion}}

In this manuscript, we proved duality between the discretization of
Fourier and convolution representations of $R$-limited functions
which lead to the sampling and interpolation theorem, Theorem \ref{thm:sampling-theorem-R_A-limited-functions},
where the interpolation is to be understood as an approximation within
a desired accuracy over a compact region. Because discretization of
the Fourier representation is over a compact support, so is the discretization
of the convolution representation. Thus, an $R$-limited function
can be approximated from samples over a compact support that is similar
to $R$. We provided examples of convolution kernels for some special
cases of $R$-limited functions, namely $T$-limited and $C$-limited
functions whose Fourier transforms are supported in a triangle, or
tetrahedron, and cone, respectively. We constructed discretization
of the Fourier representation of these kernels which can be used along
with sampling and inteprolation theorems.

\section*{Acknowledgments}

I would like to thank Lucas Monzon who introduced moment problems
and quadrature methods to me. He has been a mentor, a colleague and
most importantly a dear friend. Initial sketches of results in Sections
\ref{subsec:Approximating-PSWF-from-shifts-of-sinc} and \ref{subsec:Approximating--Slepian-functions-from-shifts-of-K}
were obtained with him in 2012. He also reviewed the manuscript in
detail which made the content more accurate and clearer. Our collaboration
wouldn't have been possible without support and trust of Konstantin
Osypov. I would like to thank Kemal {Ö}zdemir for extensive discussions
during the final preparation of this manuscript and his invaluable
signal processing perspective. I would like to thank Ozan {Ö}ktem
and Daan Huybrechs for giving me the opportunity to present parts
of this work at KTH Royal Institute of Technology and University of
Leuven where follow up questions, comments and discussions have improved
the flow and content of the manuscript. Finally I thank Garret Flagg
and Vladimir Druskin for reading the initial draft of the manuscript
and providing constructive feedback.

 \bibliographystyle{plain}
\bibliography{0D__Documents_MATLAB_Beams__1+2_D-kernels_Refer___nction_references_on_sine_integral_function,1D__Schlumberger_LaTeX_Notes_Bandwidth_Extention_slb-bib,2D__Documents_MATLAB_Beams_FBMR_references_analytic_rep}

\appendix

\section{Generalization of Pad{é} approximation\label{sec:Generalization-of-Pad=00003D0000E9}}

Let $f\left(x\right)$ and $g\left(x\right)$ be two analytic functions
related to each other by the Cauchy integral 
\begin{align}
f\left(x\right) & =\int_{\Gamma}\rho\left(z\right)g\left(zx\right)dz\label{eq:int_rep_of_f_in_g}
\end{align}
for some closed contour $\Gamma\in\mathbb{C}$ and a weighting function
$\rho\left(z\right)$. A generalization of Pad{é} approximation
is achieved by finding a rational approximation to the weighting function
\begin{align}
\rho\left(z\right) & =\frac{1}{2\pi\mathrm{i}}\sum_{m}\frac{\alpha_{m}}{z-\gamma_{m}}+\epsilon_{\rho}\left(z\right),\quad z\in\mathbb{C}\label{eq:phi(z)}
\end{align}
for some distinct $\gamma_{m}\in\mathbb{C}$ and error $\epsilon_{\rho}\left(z\right)$.
Then we refer to 
\begin{align}
f\left(x\right) & =\sum_{m}\alpha_{m}g\left(\gamma_{m}x\right)+\epsilon\left(x\right)\label{eq:generalized_pade_approx}
\end{align}
as the generalization of Pad{é} approximation from rational function
to analytic functions, for some error function $\epsilon\left(x\right)$.
Substituting the power series expansion of $f$ and $g$ at zero into
(\ref{eq:int_rep_of_f_in_g}) 
\begin{align}
f\left(x\right) & =\sum_{m=0}^{\infty}f_{n}x^{n}=\int\rho\left(z\right)\left[\sum_{m=0}^{\infty}g_{n}\left(zx\right)^{n}\right]dz
\end{align}
and equating the terms of the series, one obtains that the moments
of $\rho\left(z\right)$ are given by the ratio of the power series
coefficients, which we denote by $h_{n}$ 
\begin{align}
\int\rho\left(z\right)z^{n}dz=h_{n}=\frac{f_{n}}{g_{n}} & =\sum_{m}\alpha_{m}\gamma_{m}^{n}+\epsilon_{n}\label{eq:h_n_sum}
\end{align}
for some error $\epsilon_{n}$. Because (\ref{eq:generalized_pade_approx})
is a discrete approximation to the integral (\ref{eq:int_rep_of_f_in_g})
, $\left(\alpha_{m},\gamma_{m}\right)$ are referred to as the quadratures.
Individually, we refer to $\alpha_{m}$ and $\gamma_{m}$ as weights
and nodes, respectively. In \cite{YF2014}, we presented the detailed
theory of this generalization of Pad{é} approximation and a method
to compute the quadratures $\left(\alpha_{m},\gamma_{m}\right)$ which
is based on \cite{kung1978new}. Some examples of moment problems
are given in Table \ref{tab:Example-of-moment-problems}.

\begin{table*}
\center

\begin{tabular}{|c|c|c|c|c|c|}
\hline 
\multirow{2}{*}{{\scriptsize{}{}$f\left(Bx\right)$}}  &
\multirow{2}{*}{{\scriptsize{}{}$g\left(x\right)$}}  &
\multirow{2}{*}{{\scriptsize{}{}$f_{2n}$}}  &
\multirow{2}{*}{{\scriptsize{}{}$g_{2n}$}}  &
{\scriptsize{}{}Moment problem }  &
{\scriptsize{}{}Quadrature}\tabularnewline
 &
 &
 &
 &
{\scriptsize{}{}$h_{n}=\sum_{m}\alpha_{m}\gamma_{m}^{2n}$}  &
{\scriptsize{}{}names}\tabularnewline
\hline 
\hline 
{\scriptsize{}{}$\mbox{sinc}\left(Bx\right)$}  &
{\scriptsize{}{}$\cos\left(x\right)$}  &
{\scriptsize{}{}$\frac{\left(-1\right)^{n}B^{2n}}{\left(2n+1\right)!}$}  &
{\scriptsize{}{}$\frac{\left(-1\right)^{n}}{\left(2n\right)!}$}  &
{\scriptsize{}{}$\frac{B^{2n}}{2n+1}$}  &
{\scriptsize{}{}Gauss-Legendre}\tabularnewline
\hline 
{\scriptsize{}{}$J_{0}\left(Bx\right)$}  &
{\scriptsize{}{}$\cos\left(x\right)$}  &
{\scriptsize{}{}$\frac{\left(-1\right)^{n}B^{2n}}{(2^{n}n!)^{2}}$}  &
{\scriptsize{}{}$\frac{\left(-1\right)^{n}}{\left(2n\right)!}$}  &
{\scriptsize{}{}$\frac{\left(2n\right)!B^{2n}}{(2^{n}n!)^{2}}$}  &
{\scriptsize{}{}Clenshaw-Curtis}\tabularnewline
\hline 
{\scriptsize{}{}$\mathrm{e}^{-\left(Bx\right)^{2}}$}  &
{\scriptsize{}{}$\cos\left(x\right)$}  &
{\scriptsize{}{}$\frac{\left(-1\right)^{n}B^{2n}}{n!}$}  &
{\scriptsize{}{}$\frac{\left(-1\right)^{n}}{\left(2n\right)!}$}  &
{\scriptsize{}{}$\frac{\left(2n\right)!B^{2n}}{n!}$}  &
{\scriptsize{}{}Gauss-Hermite}\tabularnewline
\hline 
{\scriptsize{}{}$\mbox{sinc}\left(Bx\right)$}  &
{\scriptsize{}{}$\exp\left(-x^{2}\right)$}  &
{\scriptsize{}{}$\frac{\left(-1\right)^{n}B^{2n}}{\left(2n+1\right)!}$}  &
{\scriptsize{}{}$\frac{\left(-1\right)^{n}}{n!}$}  &
{\scriptsize{}{}$\frac{\left(-1\right)^{n}B^{2n}}{\left(2n+1\right)!}$}  &
\multicolumn{1}{c}{}\tabularnewline
\cline{1-5} 
{\scriptsize{}{}$J_{0}\left(Bx\right)$}  &
{\scriptsize{}{}$\mbox{sinc}\left(x\right)$}  &
{\scriptsize{}{}$\frac{\left(-1\right)^{n}B^{2n}}{(2^{n}n!)^{2}}$}  &
{\scriptsize{}{}$\frac{\left(-1\right)^{n}}{\left(2n+1\right)!}$}  &
{\scriptsize{}{}$\frac{\left(2n+1\right)!B^{2n}}{(2^{n}n!)^{2}}$}  &
\multicolumn{1}{c}{}\tabularnewline
\cline{1-5} 
{\scriptsize{}{}$\left(Bx\right)^{-1}J_{1}\left(Bx\right)$}  &
{\scriptsize{}{}$x^{-1}\mbox{cosinc}\left(x\right)$}  &
{\scriptsize{}{}$\frac{\left(-1\right)^{n}}{2^{2n+1}n!(n+1)!}$}  &
{\scriptsize{}{}$\frac{\left(-1\right)^{n}}{\left(2n+2\right)!}$}  &
{\scriptsize{}{}$\frac{(2n+2)!}{2^{2n+1}n!\,(n+1)!}$}  &
\multicolumn{1}{c}{}\tabularnewline
\cline{1-5} 
\end{tabular}\\
 \vspace{0.2cm}

\caption{Example of moment problems corresponding to approximation of some
even functions $f\left(Bx\right)\approx\sum_{m}\alpha_{m}g\left(\gamma_{m}x\right)$
in terms of other even functions $g\left(x\right)$, along with the
known quadrature names (see \cite{gautschi1997moments}). \label{tab:Example-of-moment-problems}}
\end{table*}

\section{On approximations of sinc \label{sec:On-approximations-of-sinc}}

A good approximation to $\mbox{sinc}\left(x\right)$ within the vicinity
of zero can be achieved by building up quadratures for the integral
representation of $\mbox{sinc}\left(Bx\right)$,

\begin{align}
\mbox{sinc}\left(Bx\right) & =\frac{1}{B}\int_{0}^{B}\cos\left(\omega x\right)d\omega,
\end{align}
and then rescaling the approximation by $1/B$. One way to do this
is using the method presented in Appendix \ref{sec:Generalization-of-Pad=00003D0000E9}
to obtain 
\begin{align}
\mbox{sinc}\left(Bx\right) & =\sum_{m}\alpha_{m}\cos\left(\omega_{m}x\right)+\epsilon_{B}\left(x\right),\label{eq:sinc_in_cos_approx}
\end{align}
where $\left(\alpha_{m},\omega_{m}\right)$ satisfies the moment problem
\begin{align}
h_{n}=\frac{f_{n}}{g_{n}} & =\sum_{m}\alpha_{m}\omega_{m}^{2n}+\epsilon_{n}\label{eq:moment-problem}
\end{align}
for some small $\left|\epsilon_{n}\right|$ (see Figure \ref{fig:alpha_m+omega_m-B0=00003D00003D20}).
Here 
\begin{align}
f_{n} & =\frac{\left(-1\right)^{n}B^{2n}}{\left(2n+1\right)!}
\end{align}
and 
\begin{align}
g_{n} & =\frac{\left(-1\right)^{n}}{\left(2n\right)!}
\end{align}
are the Taylor series coefficients of $\mbox{sinc}\left(x\right)$
and $\cos\left(x\right)$ at zero, respectively, and 
\begin{align}
\epsilon_{B}\left(x\right) & =\sum_{n=0}^{\infty}\epsilon_{n}x^{2n}.
\end{align}
Solution to the moment problem is equivalent to computing Gauss-Legendre
quadratures.

The approximation given in equation (\ref{eq:sinc_in_cos_approx})
yields a highly accurate approximation to the $\mbox{sinc}\left(Bx\right)$
in a neighborhood of zero (see Figure \ref{fig:Approximation-of-sinc_in_cosines}).
However, due to the rapid increase in the values of the moments $h_{n}$
for large band-limit $B$, the construction of this approximation
suffers from numerical instabilities, and therefore requires $B$
to be in the range \textbf{$0<B\le2$}. To overcome these challenges,
approximation (\ref{eq:sinc_in_cos_approx}) can be coupled with a
scaling property of the sinc, for example

\begin{align}
\mbox{sinc}\left(3^{n}Bx\right) & =\frac{1}{3}\left[2\cos\left(2\,3^{n-1}Bx\right)+1\right]\mbox{sinc}\left(3^{n-1}Bx\right),\label{eq:sinc_scaling_property}
\end{align}
to derive an error bound on the approximation of $\mathrm{sinc}\left(3^{n}Bx\right)$
in terms of the error in the approximation to the lower bandwidth
$\mathrm{sinc}\left(Bx\right)$ as a sum of scaled cosines: 
\begin{lem}
\label{lem:sinc_in_cos} Let 
\begin{align}
\epsilon_{B}\left(x\right) & =\mathrm{sinc}\left(Bx\right)-\sum_{m}\alpha_{m}\cos\left(B\theta_{m}x\right).\label{eq:epsilon_B}
\end{align}
Then 
\begin{align}
\left|\mathrm{sinc}\left(3^{n}Bx\right)-\frac{1}{3^{n}}\sum_{m}\alpha_{m}\hspace{-0.5cm}\sum_{k=-\left(3^{n}-1\right)/2}^{\left(3^{n}-1\right)/2}\hspace{-0.5cm}\cos\left(\left(\theta_{m}+2k\right)Bx\right)\right| & \le\left|\epsilon_{B}\left(x\right)\right|,\quad\mbox{for \ensuremath{n\ge0}.}
\end{align}
\end{lem}
\begin{proof}
We will prove by induction. For $n=0$, this is trivial by assumption
(\ref{eq:epsilon_B}).

Let us define 
\begin{multline}
\epsilon_{3^{n}B}\left(x\right)=\mbox{sinc}\left(3^{n}Bx\right)-\frac{1}{3^{n}}\sum_{m}\alpha_{m}\hspace{-0.5cm}\sum_{k=-\left(3^{n}-1\right)/2}^{\left(3^{n}-1\right)/2}\hspace{-0.5cm}\cos\left(\left(\theta_{m}+2k\right)Bx\right).
\end{multline}
For $n=1$, substituting (\ref{eq:sinc_in_cos_approx}) into the scaling
property (\ref{eq:sinc_scaling_property}), we obtain 
\begin{align}
\mbox{sinc}\left(3Bx\right) & =\frac{1}{3}\sum_{m}\alpha_{m}\sum_{k=-1}^{1}\cos\left(B\left(\theta_{m}-2k\right)x\right)+\frac{1}{3}\left[2\cos\left(2Bx\right)+1\right]\epsilon_{B}\left(x\right),
\end{align}
implying 
\begin{align}
\left|\epsilon_{3B}\left(x\right)\right| & =\left|\frac{1}{3}\left[2\cos\left(2Bx\right)+1\right]\epsilon_{B}\left(x\right)\right|\le\left|\epsilon_{B}\left(x\right)\right|.
\end{align}
Multiplying $\epsilon_{3^{n}B}\left(x\right)$ by $\frac{1}{3}\left[2\cos\left(2\,3^{n}Bx\right)+1\right]$,
we have 
\begin{align}
\frac{1}{3}\left[2\cos\left(2\,3^{n}Bx\right)+1\right]\epsilon_{3^{n}B}\left(x\right) & =\epsilon_{3^{n+1}B}\left(x\right)\label{eq:scaling_property_error_sinc}
\end{align}
which implies 
\begin{align}
 & \left|\epsilon_{3^{n+1}B}\left(x\right)\right|\le\left|\epsilon_{3^{n}B}\left(x\right)\right|\le\ldots\le\left|\epsilon_{3B}\left(x\right)\right|\le\left|\epsilon_{B}\left(x\right)\right|.
\end{align}
\end{proof}
\begin{cor}
\label{cor:error-sinc_approx} The error is given by 
\begin{align}
\epsilon_{3^{n+1}B}\left(x\right) & =\sum_{l=-\infty}^{\infty}\mathrm{sinc}\left(3^{n+1}B\left[x-\frac{\pi l}{B}\right]\right)\epsilon_{B}\left(x\right)\label{eq:e3n+1Bx_defn1}
\end{align}

which for $x=m\pi/B$ becomes 
\begin{align}
\epsilon_{3^{n+1}B}\left(\frac{m\pi}{B}\right) & =\epsilon_{B}\left(\frac{m\pi}{B}\right).
\end{align}
Furthermore, 
\begin{align}
\lim_{n\rightarrow\infty}3^{n+1}2B\,\epsilon_{3^{n+1}B}\left(x\right) & =\sum_{l=-\infty}^{\infty}\mathrm{\delta}\left(x-\frac{\pi l}{B}\right)\epsilon_{B}\left(x\right).
\end{align}
\end{cor}
\begin{proof}
At the end of the proof of Lemma \ref{lem:sinc_in_cos} we showed
that the error satisfies the scaling property (\ref{eq:scaling_property_error_sinc}).
Using this, we can write 
\begin{align}
\epsilon_{3^{n+1}B}\left(x\right) & =\frac{1}{3}\left[2\cos\left(2\,3^{n}Bx\right)+1\right]\epsilon_{3^{n}B}\left(x\right)\nonumber \\
 & =\frac{1}{3^{2}}\left[2\cos\left(2\,3^{n-1}Bx\right)+1\right]\left[2\cos\left(2\,3^{n}Bx\right)+1\right]\epsilon_{3^{n-1}B}\left(x\right)\nonumber \\
 & =\frac{1}{3^{n+1}}\prod_{l=1}^{n}\left[2\cos\left(2\,3^{l}Bx\right)+1\right]\epsilon_{B}\left(x\right)
\end{align}
which in the Fourier domain can be written as 
\begin{align}
\hat{\epsilon}_{3^{n+1}B}\left(k\right) & =\frac{1}{3^{n+1}}\left(\ast_{l=0}^{n}\left[\begin{array}{l}
\delta\left(k-2\,3^{l}B\right)\\
+\delta\left(k+2\,3^{l}B\right)+\delta\left(k\right)
\end{array}\right]\right)\ast\hat{\epsilon}_{B}\left(k\right)\nonumber \\
 & =\left(\frac{\chi_{[-3^{n+1}B,3^{n+1}B]}\left(k\right)}{3^{n+1}2B}2B\hspace{-0.1cm}\sum_{l=-\infty}^{\infty}\hspace{-0.1cm}\delta\left(k-2Bl\right)\right)\ast\hat{\epsilon}_{B}\left(k\right)
\end{align}
where 
\begin{align}
\hat{\epsilon}_{B}\left(k\right) & =\left(2\pi\right)^{-1}\int\epsilon_{B}\left(x\right)\mathrm{e}^{\mathrm{i}kx}dx
\end{align}
is the inverse Fourier transform of $\epsilon_{B}\left(x\right)$
and $\ast_{l=0}^{n}f_{l}\left(k\right)=\left(f_{0}\ast f_{1}\ast\cdots\ast f_{n}\right)\left(k\right)$
denotes a cascaded convolution operator. Taking the inverse Fourier
transform, we obtain

\begin{align}
\epsilon_{3^{n+1}B}\left(x\right) & =\left(\mathrm{sinc}\left(3^{n+1}Bx\right)\ast\sum_{l=-\infty}^{\infty}\delta\left(x-\frac{2\pi l}{2B}\right)\right)\epsilon_{B}\left(x\right)\nonumber \\
 & =\sum_{l=-\infty}^{\infty}\mathrm{sinc}\left(3^{n+1}B\left[x-\frac{\pi l}{B}\right]\right)\epsilon_{B}\left(x\right)
\end{align}
which for $x=m\pi/B$ is 
\begin{align}
\epsilon_{3^{n+1}B}\left(\frac{m\pi}{B}\right) & =\sum_{l=-\infty}^{\infty}\mathrm{sinc}\left(3^{n+1}\pi\left[m-l\right]\right)\epsilon_{B}\left(\frac{m\pi}{B}\right)\nonumber \\
 & =\sum_{l=-\infty}^{\infty}\delta_{ml}\epsilon_{B}\left(\frac{m\pi}{B}\right).
\end{align}
By using the identity 
\begin{align}
\lim_{a\rightarrow\infty}2a\,\mathrm{sinc}\left(ax\right) & =\delta\left(x\right)
\end{align}
and the dominated convergence theorem (see page 14 of \cite{friedlander_joshi_1998}),
we have 
\begin{align}
\lim_{n\rightarrow\infty}3^{n+1}2B\,\epsilon_{3^{n+1}B}\left(x\right) & =\lim_{n\rightarrow\infty}\sum_{l=-\infty}^{\infty}3^{n+1}2B\,\mathrm{sinc}\left(3^{n+1}B\left[x-\frac{\pi l}{B}\right]\right)\epsilon_{B}\left(x\right)\nonumber \\
 & =\sum_{l=-\infty}^{\infty}\delta\left(x-\frac{\pi l}{B}\right)\epsilon_{B}\left(x\right).
\end{align}
\end{proof}
The practical implications of Lemma \ref{lem:sinc_in_cos} are as
long as one has a good approximation $\mbox{sinc}\left(x\right)\approx\sum_{m}\alpha_{m}\cos\left(\theta_{m}x\right)$
over an interval around zero, the $\left(\alpha_{m},\theta_{m}\right)$
can be used to build up an approximation of (i) $\mbox{sinc}\left(Bx\right)$,
for any $B\in\mathbb{R}$, on the same interval, (ii) $\mbox{sinc}\left(x\right)$
on any interval around zero or, equivalently, (iii) $\mbox{sinc}\left(Bx\right)$,
for any $B\in\mathbb{R}$, on any interval around zero, as accurate
as the initial approximation to $\mbox{sinc}\left(x\right)$. (see
Figures \ref{fig:Approximation-of-sinc_in_cosines} and \ref{fig:Approximation-of-sinc_in_cosines-uniform-sampling})
Algorithm \ref{alg:sinc-in-cos} outlines our approach to approximating
a sinc of arbitrary bandwidth as a sum of scaled cosines.

\begin{algorithm*}
Given $0\le B_{0}\in\mathbb{R}$ 
\begin{enumerate}
\item Compute $n=\log_{3}\left\lfloor B_{0}\right\rfloor +1$ 
\item Set $B=B_{0}3^{-n}$. 
\item Solve the moment problem 
\begin{align*}
h_{n} & =B^{2n}\left(2n+1\right)^{-1}=\sum_{m}\alpha_{m}\left(\omega_{m}^{2}\right)^{n}
\end{align*}
for $\left(\alpha_{m},\omega_{m}\right)$ using the method of \cite{YF2014}\textbf{
}(see Appendix \ref{sec:Generalization-of-Pad=00003D0000E9} and Figure
\ref{fig:alpha_m+omega_m-B0=00003D00003D20}) 
\item Set $\theta_{m}=\omega_{m}B^{-1}$ 
\item Form the approximation (see Figure \ref{fig:Approximation-of-sinc_in_cosines})
\begin{align*}
\mbox{sinc}\left(B_{0}x\right) & \approx\frac{1}{3^{n}}\sum_{m}\alpha_{m}\sum_{k=-\left(3^{n}-1\right)/2}^{\left(3^{n}-1\right)/2}\cos\left(\left(\theta_{m}+2k\right)Bx\right)
\end{align*}
\end{enumerate}
\caption{\label{alg:sinc-in-cos}Representation of $\mbox{sinc}\left(B_{0}x\right)$
as a sum of scaled cosines}
\end{algorithm*}

\begin{figure}
\includegraphics[scale=0.6]{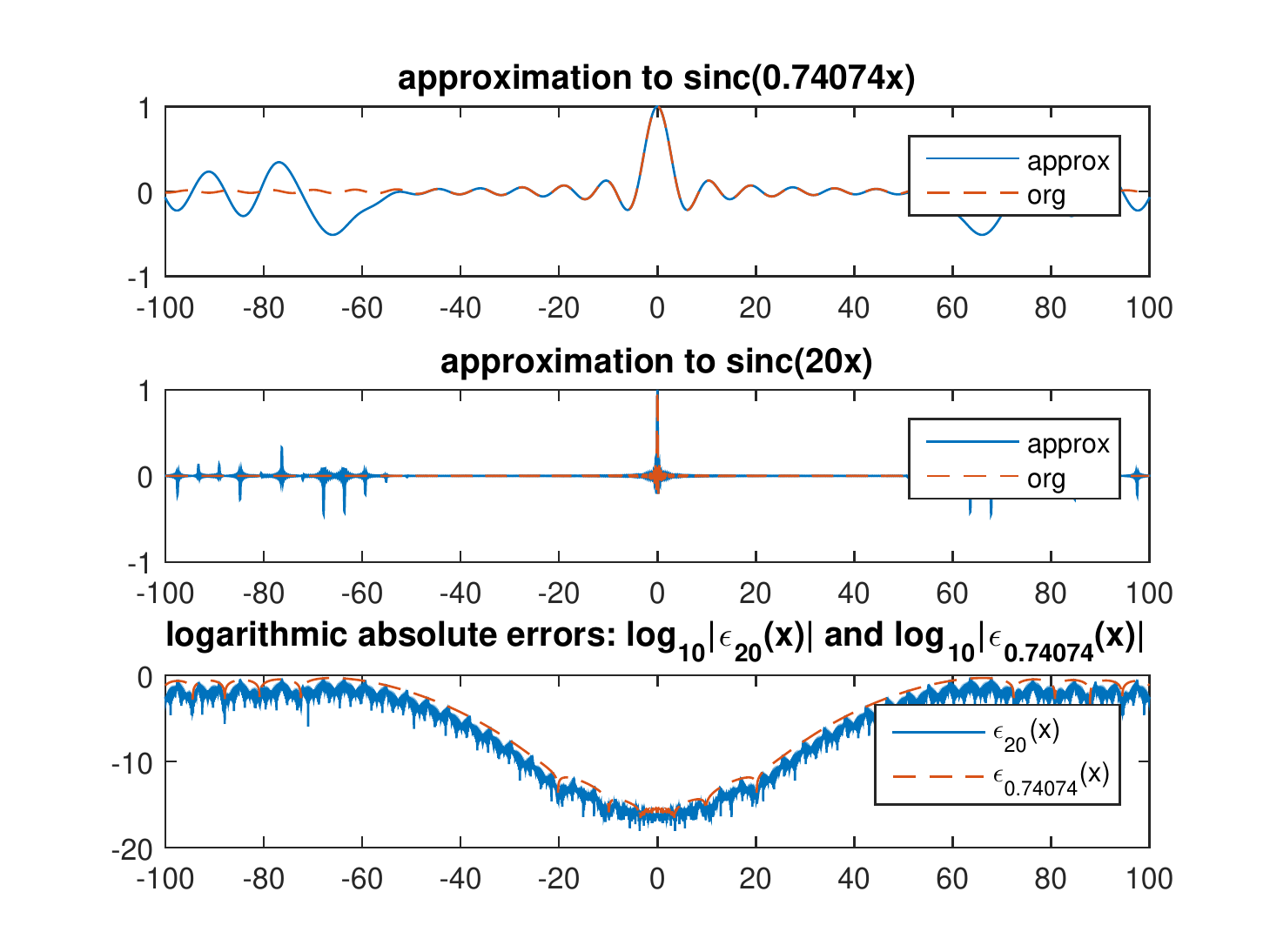}\caption{Approximation of $\mbox{sinc}\left(Bx\right)$ by (\ref{eq:sinc_in_cos_approx})
using Algorithm \ref{alg:sinc-in-cos}. On top and middle plots, $\mbox{sinc}\left(Bx\right)$
and $\mbox{sinc}\left(B_{0}x\right)$ (red dashed) along with their
approximations (solid blue), for $B=3^{-3}20$ and $B_{0}=20$, respectively.
On the bottom plot, the logarithmic absolute errors for $B$ (red
dashed) and $B_{0}$(blue solid). As derived the error corresponding
to $B_{0}$ is less than that of $B$. \label{fig:Approximation-of-sinc_in_cosines}}
\end{figure}

\begin{figure}
\includegraphics[scale=0.6]{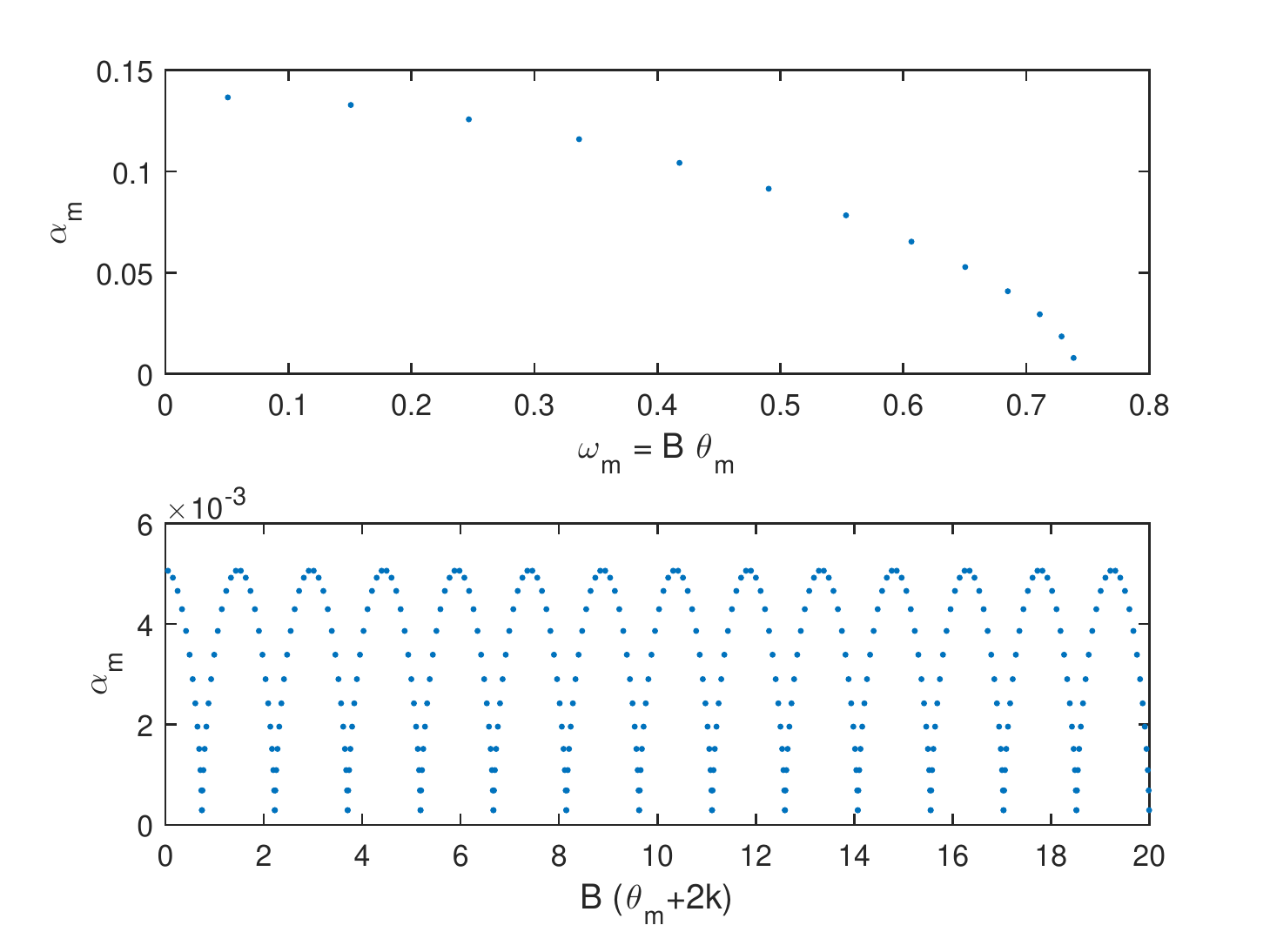}\caption{On top plot, solution $\left(\alpha_{m},\omega_{m}\right)$ to the
moment problem (\ref{eq:moment-problem}) for $B=3^{-3}B_{0}$ and
$B_{0}=20$ which is used to approximate $\mbox{sinc}\left(Bx\right)$
(see top plot in Figure \ref{fig:Approximation-of-sinc_in_cosines}).
On the bottom plot, $\left(\alpha_{m},\left|\theta_{m}+2k\right|B\right)_{k=-\left(3^{3}-1\right)/2}^{\left(3^{3}-1\right)/2}$
used to approximate$\mbox{sinc}\left(B_{0}x\right)$ (see middle plot
in Figure \ref{fig:Approximation-of-sinc_in_cosines}). \label{fig:alpha_m+omega_m-B0=00003D00003D20}}
\end{figure}

\paragraph*{How about uniform sampling?}

Consider approximation of the integral 
\begin{align}
f\left(x\right) & =\frac{1}{2B}\int_{-B}^{B}\mathrm{e}^{\mathrm{i}x\omega}d\omega=\mbox{sinc}\left(Bx\right)
\end{align}
by discretization of the integral using uniform sampling over the
interval $\left[-B,B\right]$: 
\begin{align}
\tilde{f}_{B,N}\left(x\right) & =\frac{1}{2B}\frac{2B}{2N+1}\sum_{n=-N}^{N}\mathrm{e}^{\mathrm{i}x\frac{2B}{2N+1}n}\nonumber \\
 & =\frac{1}{2N+1}\frac{\sin\left(Bx\right)}{\sin\left(\frac{Bx}{2N+1}\right)}\label{eq:sin(Bx)/sin(x)}
\end{align}
Because $\tilde{f}_{B,N}\left(x\right)$ is periodic with period $\pi B^{-1}\left(2N+1\right)$,
it is also referred to as periodic sinc function.

Using the Taylor series expansion of $\left(1-x\right)^{-1}$ and
$\mbox{sinc}\left(x\right)$ around zero \footnote{,$\left(1-x\right)^{-1}=\sum_{m=0}^{\infty}x^{m}$ and $\mbox{sinc}\left(x\right)=\sum_{n=0}^{\infty}\frac{\left(-1\right)^{n}}{\left(2n+1\right)!}x^{2n}$}
series representation of the error becomes 
\begin{align}
\epsilon_{B}\left(x\right) & =f\left(x\right)-\tilde{f}_{B,N}\left(x\right)\nonumber \\
 & =\frac{\sin\left(Bx\right)}{Bx}\left(1-\frac{1}{1-\left[1-\mbox{sinc}\left(\frac{Bx}{2N+1}\right)\right]}\right)\nonumber \\
 & =\mbox{sinc}\left(Bx\right)\!\sum_{m=1}^{\infty}\!\left(-1\right)^{m+1}\!\left[\sum_{n=1}^{\infty}\!\frac{\left(-1\right)^{n}}{\left(2n+1\right)!}\!\left(\frac{B}{2N+1}x\right)^{2n}\!\right]^{m}\label{eq:eps_B(x)_O(2N+1)^(-2)}
\end{align}
which decays like $\mathcal{O}\left(\left(2N+1\right)^{-2}\right)$
within the vicinity of zero and increases away from zero for $\left|x\right|<\pi\left(2B\right)^{-1}\left(2N+1\right)$.
Consequently, maximum absolute error is obtained at $\left|x\right|=\pi\left(2B\right)^{-1}\left(2N+1\right)$
which is 
\begin{align}
\left|\epsilon_{B}\left(\frac{\left(2N+1\right)\pi}{2B}\right)\right| & =\left|\mathrm{sinc}\left(\frac{\left(2N+1\right)\pi}{2}\right)\left(1-\mathrm{sinc}\left(\frac{\pi}{2}\right)^{-1}\right)\right|\nonumber \\
 & =\frac{2}{\left(2N+1\right)\pi}\left(\frac{\pi}{2}-1\right)\label{eq:eps_B(x)_O(2N+1)^(-1)}
\end{align}
and decays in the order of $N$. For $\left(N,B\right)=\left(13,20\times3^{-3}\right)$
and $\left(N_{0},B_{0}\right)=\left(13\times3^{3},20\right)$ we present
$\tilde{f}_{B,N}\left(x\right)$ and $\tilde{f}_{B_{0},N_{0}}\left(x\right)$
in Figure \ref{fig:Approximation-of-sinc_in_cosines-uniform-sampling}.

\begin{figure}
\includegraphics[scale=0.6]{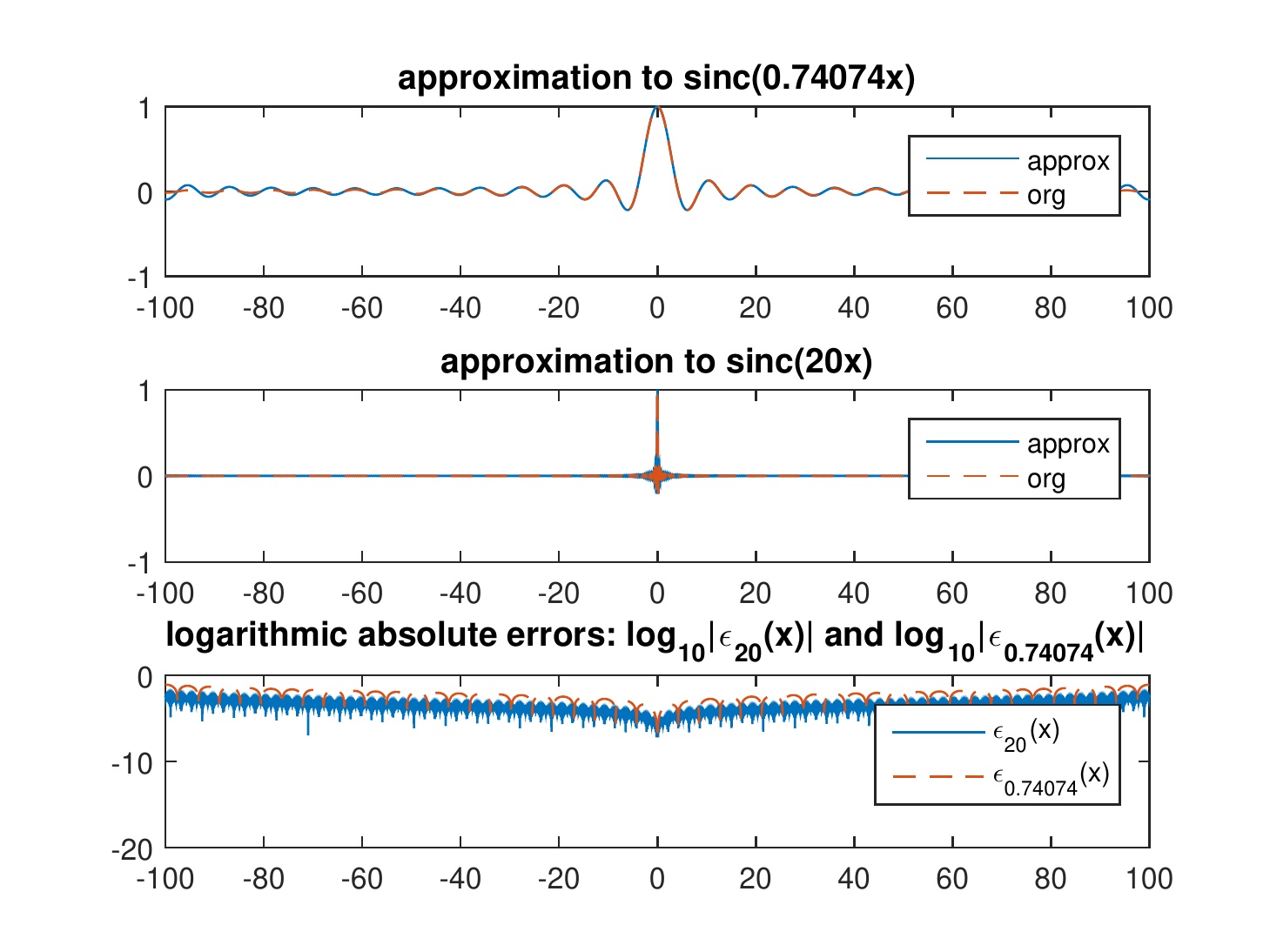}\caption{Approximation of $\mbox{sinc}\left(Bx\right)$ by (\ref{eq:sin(Bx)/sin(x)}).
On top and middle plots, $\mbox{sinc}\left(Bx\right)$ and $\mbox{sinc}\left(B_{0}x\right)$
(red dashed) along with their approximations $\tilde{f}_{B,N}\left(x\right)$
and $\tilde{f}_{B_{0},N_{0}}\left(x\right)$ (solid blue) using uniform
sampling, for $\left(B,N\right)=\left(3^{-3}20,13\right)$ and $\left(B_{0},N_{0}\right)=\left(20,3^{3}13\right)$,
respectively. On the bottom plot, the logarithmic absolute errors
for $B$ (red dashed) and $B_{0}$(blue solid). As derived the error
corresponding to $B_{0}$ is less than that of $B$. \label{fig:Approximation-of-sinc_in_cosines-uniform-sampling}}
\end{figure}

Lemma \ref{lem:sinc_in_cos} can be generalized to any arbitrary function: 
\begin{thm}
\label{thm:sinc_in_f} Let 
\begin{align}
\epsilon_{B}\left(x\right) & =\mathrm{sinc}\left(Bx\right)-f\left(x\right).\label{eq:epsilon_B-2}
\end{align}
Then 
\begin{align}
\left|\mathrm{sinc}\left(3^{n}Bx\right)-\frac{1}{3^{n}}\left[\sum_{l=1}^{\left(3^{n}-1\right)/2}2\cos\left(2Blx\right)+1\right]f\left(x\right)\right| & \le\left|\epsilon_{B}\left(x\right)\right|,\,\mbox{for \ensuremath{n\ge0}.}
\end{align}
\end{thm}
\begin{proof}
The proof follows from arguments that are similar to those in Lemma
\ref{lem:sinc_in_cos}. 
\end{proof}
\begin{cor}
\label{cor:sinc_in_chirplets} Let 
\begin{align}
\epsilon_{B}\left(x\right) & =\mathrm{sinc}\left(Bx\right)-\sum_{m}\alpha_{m}\exp\left(-\gamma_{m}x^{2}\right),\label{eq:epsilon_B-1}
\end{align}
such that $\mathrm{Re}\left\{ \gamma_{m}\right\} >0$. Then 
\begin{align}
\left|\mathrm{sinc}\left(3^{n}Bx\right)-\frac{1}{3^{n}}\sum_{m}\sum_{l=-\left(3^{n}-1\right)/2}^{\left(3^{n}-1\right)/2}a_{m,l}g_{m,l}\left(x\right)\right| & \le\left|\epsilon_{B}\left(x\right)\right|,\,\mbox{for \ensuremath{n\ge0}.}
\end{align}
where $a_{m,l}=\alpha_{m}\exp\left(\mathrm{i}\frac{\left(Bl\right)^{2}}{\mathrm{Im}\left\{ \gamma_{m}\right\} }\right)$
and 
\begin{align}
g_{m,l}\left(x\right) & =\exp\left(-\mathrm{Re}\left\{ \gamma_{m}\right\} x^{2}\right)\nonumber \\
 & \qquad\times\exp\left(-\mathrm{i}\mathrm{Im}\left\{ \gamma_{m}\right\} \left(x-\frac{Bl}{\mathrm{Im}\left\{ \gamma_{m}\right\} }\right)^{2}\right).
\end{align}
\end{cor}
\begin{proof}
This is a direct consequence of Theorem \ref{thm:sinc_in_f} and identities
\begin{align}
\left[\sum_{l=1}^{\left(3^{n}-1\right)/2}2\cos\left(2Blx\right)+1\right]\exp\left(-\gamma_{m}x^{2}\right) & =\sum_{l=-\left(3^{n+1}-1\right)/2}^{\left(3^{n+1}-1\right)/2}\exp\left(-\gamma_{m}x^{2}+\mathrm{i}2Blx\right)
\end{align}

\begin{align}
 & \exp\left(-\gamma_{m}x^{2}+\mathrm{i}2Blx\right)\nonumber \\
 & =\exp\left(-\mathrm{Re}\left\{ \gamma_{m}\right\} x^{2}\right)\exp\left(-\mathrm{i}\mathrm{Im}\left\{ \gamma_{m}\right\} \left(x-\frac{Bl}{\mathrm{Im}\left\{ \gamma_{m}\right\} }\right)^{2}\right)\exp\left(\mathrm{i}\frac{\left(Bl\right)^{2}}{\mathrm{Im}\left\{ \gamma_{m}\right\} }\right)
\end{align}
\end{proof}
Corollary \ref{cor:sinc_in_chirplets} says that the sinc function
can be approximated as a sum of shifted, Gaussian tapered chirps.
One can determine $\left(\alpha_{m},\gamma_{m}\right)$ using the
method in Appendix \ref{sec:Generalization-of-Pad=00003D0000E9} by
solving the appropriate moment problem (see Step 3 of Algorithm \ref{alg:sinc-in-chirplets}).
This type of approximations of $\mbox{sinc}\left(x\right)$ can be
used to construct a multiresolution scheme for band-limited function
as an alternative to existing multiscale approaches. It is important
to point out that unlike chirplet decomposition methods presented
in \cite{mann1995chirplet,candes2002multiscale}, the moment problem
provides an explicit solution for $\left(\alpha_{m},\gamma_{m}\right)$
while coupling the real and imaginary part of the complex Gaussian
parameters $\gamma_{m}$. Algorithm \ref{alg:sinc-in-chirplets} outlines
approximating a sinc of arbitrary bandwidth as a sum of scaled cosines
based on the moment problem and Corollary \ref{cor:sinc_in_chirplets}.
A corresponding example is presented in Figure \ref{fig:Approximation-of-sinc_in_chirplets}.

\begin{algorithm*}
Given $0\le B_{0}\in\mathbb{R}$ 
\begin{enumerate}
\item Compute $n=\log_{3}\left\lfloor B_{0}\right\rfloor +1$ 
\item Set $B=B_{0}3^{-n}$. 
\item Solve the moment problem 
\begin{align*}
h_{n} & =B^{2n}\frac{n!}{\left(2n+1\right)!}=\sum_{m}\alpha_{m}\gamma_{m}^{n}
\end{align*}
for $\left(\alpha_{m},\gamma_{m}\right)$ using the method of \cite{YF2014}\textbf{
}(see Appendix \ref{sec:Generalization-of-Pad=00003D0000E9}) 
\item Form the approximation 
\begin{align*}
\mbox{sinc}\left(B_{0}x\right) & \approx\frac{1}{3^{n}}\sum_{m}\sum_{l=-\left(3^{n}-1\right)/2}^{\left(3^{n}-1\right)/2}\hspace{-0.5cm}\alpha_{m}\left[\begin{array}{l}
\exp\left(\mathrm{i}\frac{\left(Bl\right)^{2}}{\mathrm{Im}\left\{ \gamma_{m}\right\} }\right)\exp\left(-\mathrm{Re}\left\{ \gamma_{m}\right\} x^{2}\right)\\
\times\exp\left(\begin{array}{l}
-\mathrm{i}\,\mathrm{Im}\left\{ \gamma_{m}\right\} \left(x-\frac{Bl}{\mathrm{Im}\left\{ \gamma_{m}\right\} }\right)^{2}\end{array}\right)
\end{array}\right]
\end{align*}
\end{enumerate}
\caption{\label{alg:sinc-in-chirplets}Representation of $\mbox{sinc}\left(B_{0}x\right)$
as a sum of chirplets}
\end{algorithm*}

\begin{figure}
\includegraphics[scale=0.6]{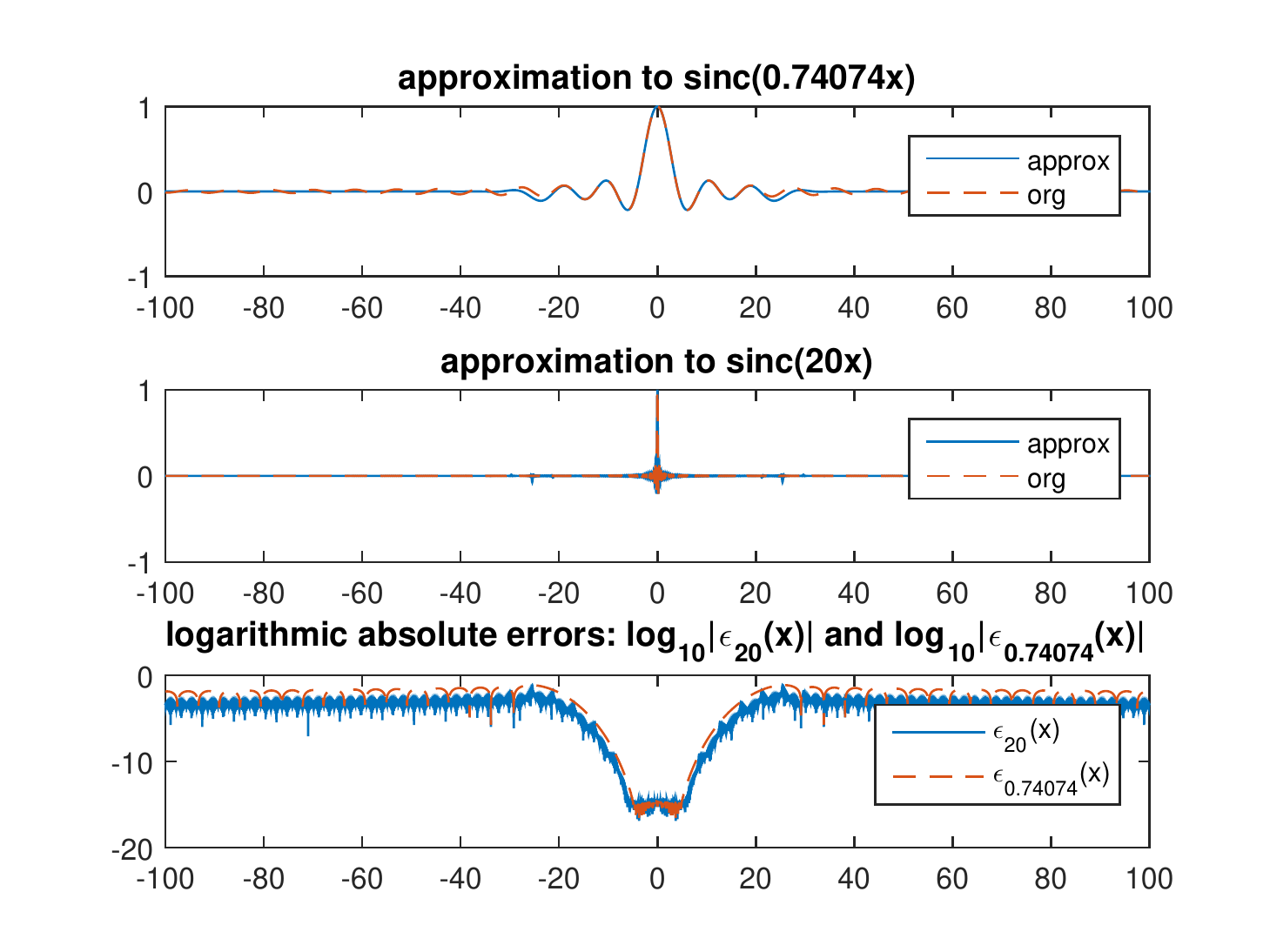}\caption{Approximation of $\mbox{sinc}\left(Bx\right)$ as a sum of chirplets
(see Corollary \ref{cor:sinc_in_chirplets}) using Algorithm \ref{alg:sinc-in-chirplets}.
On top and middle plots, $\mbox{sinc}\left(B\,x\right)$ and $\mbox{sinc}\left(B_{0}x\right)$
(red dashed) along with their approximations (solid blue), for $B=3^{-3}20$
and $B_{0}=20$, respectively. On the bottom plot, the logarithmic
absolute errors for $B$ (red dashed) and $B_{0}$(blue solid). As
derived the error corresponding to $B_{0}$ is less than that of $B$.
\label{fig:Approximation-of-sinc_in_chirplets}}
\end{figure}

\section{$T$-limited functions\label{sec:Kernel-for-T-bandlimited-functions}}

In this section we consider triangle and tetrahedral limited functions
both of whom are referred to as T-limited functions. The distinction
of two class of T-limited function should be clear from the dimensions
of their variables.

\subsection{Triangle-limited functions. }

We say a function $f\left(x,y\right)$, $\left(x,y\right)\in\mathbb{R}^{2}$
is triangle limited if its Fourier transform $\hat{f}\left(k_{x},k_{y}\right)$
is supported within a triangular region $T\subset\mathbb{R}^{2}$.
Without loss of generality let $T\subset\mathbb{R}^{2}$ be parametrize
by 
\begin{align}
T & =\left\{ \left(k_{x},k_{y}\right)\,|\,0\le k_{x}\le\Delta p,\,\left|k_{y}\right|\le k_{x}s,\right\} \label{eq:T}
\end{align}
for some $\Delta p,s\in\mathbb{R}^{+}$. Define $K_{\triangle}\left(x,y\right)$
to be 
\begin{align}
 & K_{\triangle}\left(x,y\right)=\int_{0}^{\Delta p}\int_{-k_{x}s}^{k_{x}s}\mathrm{e}^{\mathrm{i}2\pi\left(k_{x}x+k_{y}y\right)}dk_{y}dk_{x}\nonumber \\
 & =\hspace{-0.1cm}\frac{\Delta p}{2\pi y}\hspace{-0.1cm}\left[\hspace{-0.2cm}\begin{array}{l}
\mbox{cosinc}\left(2\pi\Delta p\left(x+sy\right)\right)-\mbox{cosinc}\left(2\pi\Delta p\left(x-sy\right)\right)\\
-\mathrm{i}\left[\mbox{sinc}\left(2\pi\Delta p\left(x+sy\right)\right)-\mbox{sinc}\left(2\pi\Delta p\left(x-sy\right)\right)\right]
\end{array}\hspace{-0.2cm}\right]\label{eq:K-triangle}
\end{align}
We refer to $K_{\triangle}$ as the kernel for $T$-limited functions.
Let us consider the kernel 
\begin{align*}
K\left(x,y\right) & =K_{\triangle}\left(x,y\right)\exp\left(-\mathrm{i}2\pi\Delta p\frac{2}{3}x\right)
\end{align*}
for $s=\sqrt{3}^{-1}$. Fourier transform of $K\left(x,y\right)$
is the characteristic function over an equilateral triangle whose
center of mass is at the origin. For $\Delta p=75$, we present the
kernel $K\left(x,y\right)$ and its Fourier transform in Figure \ref{fig:K-triangle}.
Next we will give two ways to construct quadratures for discrete Fourier
approximation of $K_{\triangle}\left(x,y\right)$ which can be generalized
to construct quadratures for simplexes in higher dimensions, too.

\begin{figure}
\center\includegraphics[scale=0.5]{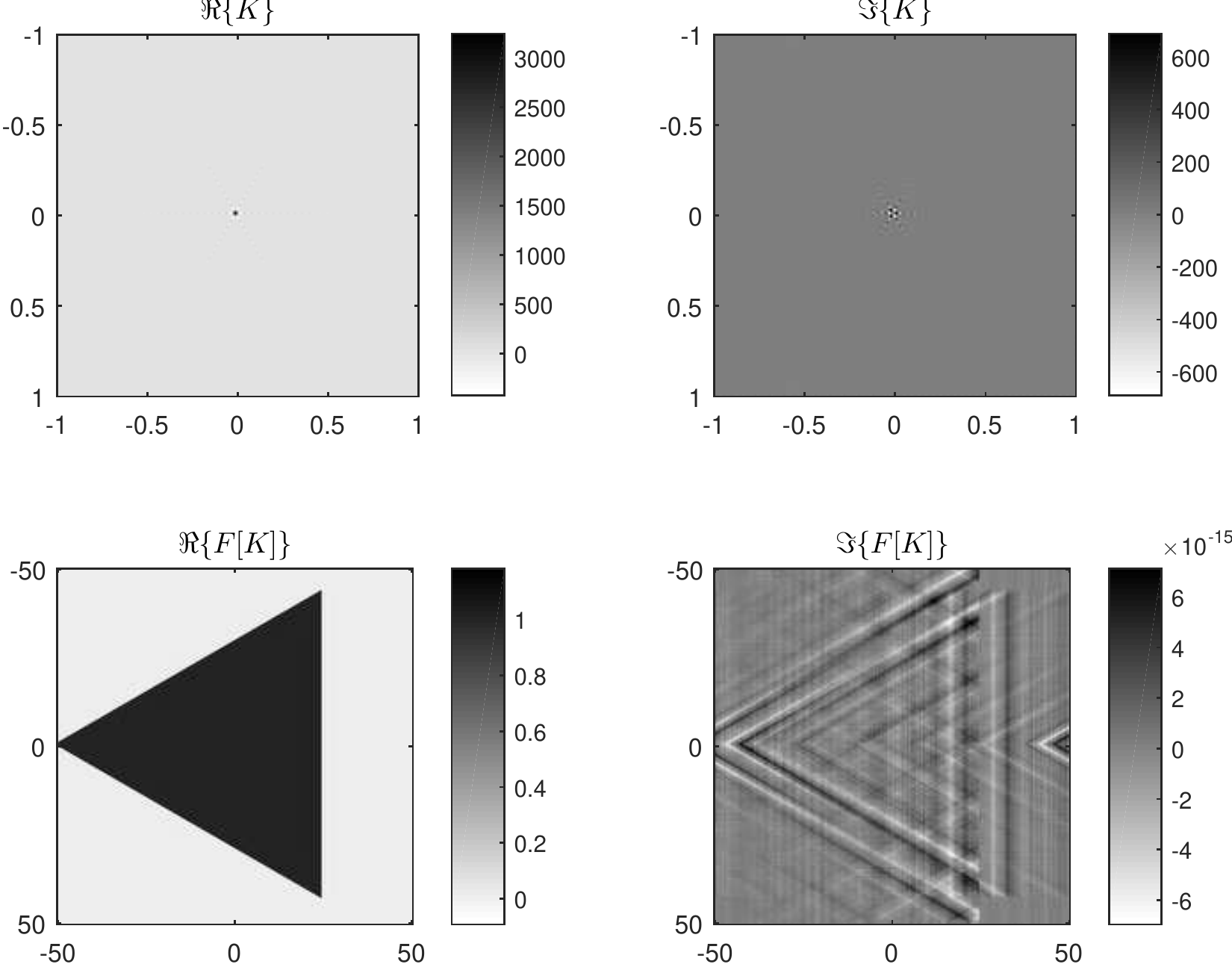}\caption{Real, imaginary parts of $K\left(x,y\right)=K_{\triangle}\left(x,y\right)\exp\left(-\mathrm{i}2\pi\Delta p\frac{3}{2}x\right)$
and its Fourier transform for $\Delta p=75$ and $s=\sqrt{3}^{-1}$.
The horizontal axis are $x$ and $k_{x}$ axis and the vertical axis
are $y$ and $k_{y}$, appropriately. \label{fig:K-triangle}}
\end{figure}

\begin{prop}
\label{prop:K_m_2_K_=00003D00007Bm-1=00003D00007D} $K_{\triangle}\left(x,y\right)$
satisfies the following scaling property 
\begin{align}
K_{\triangle}\left(x,y\right) & =\frac{1}{4}\left[\begin{array}{l}
K_{\triangle}\left(\frac{x}{2},\frac{y}{2}\right)\left(1+2\mathrm{e}^{\mathrm{i}\pi\Delta px}\cos\left(\pi\Delta psy\right)\right)\\
+\mathrm{e}^{\mathrm{i}2\pi\Delta px}K_{\triangle}\left(-\frac{x}{2},\frac{y}{2}\right)
\end{array}\right]\label{eq:scaling_prop-4-K_triangle}
\end{align}
Consequently, let $K_{\triangle,m}\left(x,y\right)=K_{\triangle}\left(2^{m}x,2^{m}y\right)$.
Then 
\begin{align}
K_{\triangle,m}\left(x,y\right) & =\frac{1}{4}\left[\begin{array}{l}
K_{\triangle,m-1}\left(x,y\right)\left(1+2\mathrm{e}^{\mathrm{i}\pi2^{m}\Delta px}\cos\left(\pi2^{m}\Delta psy\right)\right)\\
+\mathrm{e}^{\mathrm{i}\pi2^{m+1}\Delta px}K_{\triangle,m-1}\left(-x,y\right)
\end{array}\right]\label{eq:K_m_2_K_=00003D00007Bm-1=00003D00007D}
\end{align}
\end{prop}
\begin{proof}
This is a direct consequence self similarity of isosceles triangle
which is used to decompose integral representation of $K_{\triangle}\left(x,y\right)$
using the identity 
\begin{multline}
\int_{0}^{\Delta p}\int_{-k_{x}s}^{k_{x}s}dk_{y}dk_{x}\\
=\left[\begin{array}{l}
\int_{0}^{\Delta p/2}\int_{-k_{x}s}^{k_{x}s}+\int_{\Delta p/2}^{\Delta p}\int_{-\left(\Delta p-k_{x}\right)s}^{\left(\Delta p-k_{x}\right)s}\\
+\int_{\Delta p/2}^{\Delta p}\int_{\left(\Delta p-k_{x}\right)s}^{k_{x}s}+\int_{\Delta p/2}^{\Delta p}\int_{-k_{x}s}^{\left(k_{x}-\Delta p\right)s}
\end{array}\right]dk_{y}dk_{x}
\end{multline}
and $\frac{1}{4}K_{\triangle}\left(\frac{x}{2},\frac{y}{2}\right)=\int_{0}^{\Delta p/2}\int_{-k_{x}s}^{k_{x}s}\mathrm{e}^{\mathrm{i}2\pi\left(k_{x}x+k_{y}y\right)}dk_{y}dk_{x}$ 
\end{proof}
\begin{cor}
\label{cor:K_=00003D00007Bm-1=00003D00007D_2_K_m} Let $K_{\triangle,m}\left(x,y\right)=K_{\triangle}\left(2^{m}x,2^{m}y\right)$.
Then 
\begin{align}
K_{\triangle,m-1}\left(x,y\right) & =\left[\frac{\left(\begin{array}{l}
K_{\triangle,m}\left(x,y\right)\left(1+2\mathrm{e}^{-\mathrm{i}\pi2^{m}\Delta px}\cos\left(\pi2^{m}\Delta psy\right)\right)\\
-\mathrm{e}^{\mathrm{i}\pi2^{m+1}\Delta px}K_{\triangle,m}\left(-x,y\right)
\end{array}\right)}{\cos^{2}\left(\pi2^{m}\Delta psy\right)+\cos\left(\pi2^{m}\Delta px\right)\cos\left(\pi2^{m}\Delta psy\right)}\right]\label{eq:K_=00003D00007Bm-1=00003D00007D_2_K_m}
\end{align}
\end{cor}
\begin{proof}
Considering $K_{\triangle,m}\left(x,y\right)$ and $K_{\triangle,m}\left(-x,y\right)$,
computation of 
\begin{multline}
\frac{K_{\triangle,m}\left(x,y\right)}{\mathrm{e}^{\mathrm{i}\pi2^{m+1}\Delta px}}-\frac{K_{\triangle,m}\left(-x,y\right)}{\left(1+2\mathrm{e}^{-\mathrm{i}\pi2^{m}\Delta px}\cos\left(\pi2^{m}\Delta psy\right)\right)}\\
=\frac{1}{4}\left[\begin{array}{l}
K_{\triangle,m-1}\left(x,y\right)\\
\times\left[\frac{\left|1+2\mathrm{e}^{\mathrm{i}\pi2^{m}\Delta px}\cos\left(\pi2^{m}\Delta psy\right)\right|^{2}-1}{\mathrm{e}^{\mathrm{i}\pi2^{m+1}\Delta px}\left(1+2\mathrm{e}^{-\mathrm{i}\pi2^{m}\Delta px}\cos\left(\pi2^{m}\Delta psy\right)\right)}\right]
\end{array}\right]
\end{multline}
leads to (\ref{eq:K_=00003D00007Bm-1=00003D00007D_2_K_m}). 
\end{proof}

\subsubsection{Discrete Fourier approximation of $K_{\triangle}\left(x,y\right)$
\label{subsec:Discrete-Fourier-approximation-T-kernel}}

For a discrete representation of the kernel let us consider a bounded
region $\left(x,y\right)\in S$. $T$-limited projection operator
$P_{\triangle}$and $T$-limited projection $f_{\triangle}$ of $f$
restricted to the region $S$ are defined by 
\begin{align}
P_{\triangle}\left[f\right]\left(x,y\right) & =f_{\triangle}\left(x,y\right)\nonumber \\
 & =\int_{S}f\left(x',y'\right)K_{\triangle}\left(x-x',y-y'\right)dx'dy'
\end{align}
Noticing that argument of $K_{\triangle}$ ranges over 
\begin{align}
S+S & =\left\{ \left(x,y\right)|\left(x,y\right)=\left(x_{1},y_{1}\right)+\left(x_{2},y_{2}\right),\,\left(x_{n},y_{n}\right)_{n=1,2}\in S\right\} ,
\end{align}
in order to compute $f_{\triangle}\left(x,y\right)$ accurately over
the region $S$, one should have an accurate representation of $K_{\triangle}$
inside $S+S$ .

Recalling (\ref{eq:K-triangle}),

\begin{align}
K_{\triangle}\left(x,y\right) & =\int_{0}^{\Delta p}\int_{-k_{x}s}^{k_{x}s}\mathrm{e}^{\mathrm{i}2\pi\left(k_{x}x+k_{y}y\right)}\,dk_{y}dk_{x}\nonumber \\
 & =\frac{\Delta p}{2\pi\mathrm{i}}s\int_{-1}^{1}\partial_{x}\left[\int_{0}^{1}\mathrm{e}^{\mathrm{i}2\pi\Delta p\left(x+sk_{y}y\right)k_{x}}dk_{x}\right]\,dk_{y}\label{eq:K-step1}\\
 & \approx\frac{\Delta p}{2\pi\mathrm{i}}s\int_{-1}^{1}\partial_{x}\left[\sum_{m=1}^{M}\alpha_{m}\mathrm{e}^{\mathrm{i}2\pi\Delta p\left(x+sk_{y}y\right)k_{x}\left[m\right]}\right]\,dk_{y}\nonumber \\
 & =\Delta p^{2}\,s\sum_{m=1}^{M}\left(\begin{array}{l}
\alpha_{m}k_{x}\left[m\right]\mathrm{e}^{\mathrm{i}2\pi\Delta p\,x\,k_{x}\left[m\right]}\\
\times\left[\int_{-1}^{1}\mathrm{e}^{\mathrm{i}2\pi\Delta p\,s\,k_{x}\left[m\right]\,y\,k_{y}}\,dk_{y}\right]
\end{array}\right)\label{eq:K-step2}\\
 & \approx2\Delta p^{2}\,s\sum_{m=1}^{M}\left(\begin{array}{l}
\alpha_{m}k_{x}\left[m\right]\mathrm{e}^{\mathrm{i}2\pi\Delta p\,x\,k_{x}\left[m\right]}\mathrm{e}^{-\mathrm{i}2\pi2\Delta p\,s\,k_{x}\left[m\right]\,y}\\
\times\left[\sum_{n=1}^{N\left(m\right)}\beta_{m,n}\mathrm{e}^{\mathrm{i}2\pi\Delta p\,s\,k_{x}\left[m\right]\,y\,2k_{y}\left[m,n\right]}\right]
\end{array}\right)\nonumber 
\end{align}
where $\left(\alpha_{m},k_{x}\left[m\right]\right)$ and $\left(\beta_{m,n},k_{y}\left[m,n\right]\right)$
are quadratures for approximating sinc as a sum of cosines (See (\ref{eq:sinc_in_cos_approx})
and consider $B$ equal to $2\pi\Delta p\left(X+sY\right)$ and $4\pi\Delta p\,s\,k_{x}\left[m\right]Y$,
respectively, where $X=\max_{\left(x,y\right)\in S+S}\left|x\right|$
and $Y=\max_{\left(x,y\right)\in S+S}\left|y\right|$.) for a desired
accuracy, which determines the accuracy of approximating $K_{\triangle}$
inside $S+S$. 
\begin{cor}
\label{cor:scaling-property-T-limited-error} Let $\tilde{K}_{\triangle,m}\left(x,y\right)$
be an approximation of $K_{\triangle,m}\left(x,y\right)$ and $\epsilon_{\triangle,m}\left(x,y\right)=K_{\triangle,m}\left(x,y\right)-\tilde{K}_{\triangle,m}\left(x,y\right)$
be the associated error. Then 
\begin{align}
\left|\epsilon_{\triangle,m}\left(x,y\right)\right| & \le\frac{1}{4}\left[3\left|\epsilon_{\triangle,m-n}\left(x,y\right)\right|+\left|\epsilon_{\triangle,m-n}\left(-x,y\right)\right|\right].
\end{align}
Furthermore, if $\left|\epsilon_{\triangle,m_{0}}\left(x,y\right)\right|=\left|\epsilon_{\triangle,m_{0}}\left(-x,y\right)\right|$
for an $m_{0}\in\mathbb{Z}$, then $\left|\epsilon_{\triangle,m}\left(x,y\right)\right|\le\left|\epsilon_{\triangle,m_{0}}\left(x,y\right)\right|$
for all $m\ge m_{0}$. Consequently, 
\begin{align}
\left|K_{\triangle}\left(x,y\right)-\tilde{K}_{\triangle,m}\left(2^{-m}x,2^{-m}y\right)\right| & =\left|\epsilon_{\triangle,m}\left(2^{-m}x,2^{-m}y\right)\right|\nonumber \\
 & \le\left|\epsilon_{\triangle}\left(2^{-m}x,2^{-m}y\right)\right|
\end{align}
\end{cor}
\begin{proof}
Proof by induction using the identity 
\begin{align}
\epsilon_{\triangle,m}\left(x,y\right) & =\frac{1}{4}\hspace{-0.25cc}\left[\hspace{-0.25cc}\begin{array}{l}
\epsilon_{\triangle,m-1}\left(x,y\right)\left(1+2\mathrm{e}^{\mathrm{i}2\pi\Delta px}\cos\left(2\pi\Delta psy\right)\right)\\
+\mathrm{e}^{\mathrm{i}4\pi\Delta px}\epsilon_{\triangle,m-1}\left(-x,y\right)
\end{array}\hspace{-0.5cc}\right]
\end{align}
obtained from (\ref{eq:K_m_2_K_=00003D00007Bm-1=00003D00007D}). 
\end{proof}
Corollary \ref{cor:scaling-property-T-limited-error} says that for
approximating $K_{\triangle,m}\left(x,y\right)$ with a desired error
bound $\epsilon$ over a desired region centered around zero, it is
sufficient to find an approximation to $K_{\triangle,m_{0}}\left(x,y\right)$
for any integer (including negative integers) $m_{0}<m$ whose error
is less than or equal to $\epsilon$ within the vicinity of zero.
We present $\tilde{K}_{\triangle,m}\left(2^{-m}x,2^{-m}y\right)$
and $\left|\epsilon_{\triangle,m}\left(2^{-m}x,2^{-m}y\right)\right|$
for $m_{0}=0$ and $m=0,1,\ldots,7$ in Figures \ref{fig:RNS-quad-approx-K_triangle}
and \ref{fig:RS-quad-approx-K_triangle}.

\subsubsection{Nodes capturing rotational invariance of equilateral triangle}

Let us consider the equilateral triangle $T_{E}$ with each side equal
to $1$ and center of mass at the origin. $T_{E}$ is equivalent to
the triangle $T$ of (\ref{eq:T}) with $\Delta p=\sqrt{3}/2$ and
$s=\sqrt{3}^{-1}$. While we can construct discrete Fourier approximation
of the kernel for the equilateral triangle using the method in Appendix
\ref{subsec:Discrete-Fourier-approximation-T-kernel}, constructed
nodes $\left(k_{x}\left[m\right],k_{y}\left[m,n\right]\right)$ do
not necessarily satisfy the rotational invariance of equilateral triangle
(see Figure \ref{fig:Nodes-for-equilateral}).

\begin{figure}
\center

\begin{tabular}{|c|c|c|}
\hline 
{\scriptsize{}{}$m$}  &
\textbf{\scriptsize{}{}RNI}{\scriptsize{} } &
\textbf{\scriptsize{}{}RI}\tabularnewline
\hline 
\multicolumn{1}{|>{\centering}m{0.5cm}|}{{\scriptsize{}{}0}} &
\includegraphics[scale=0.15]{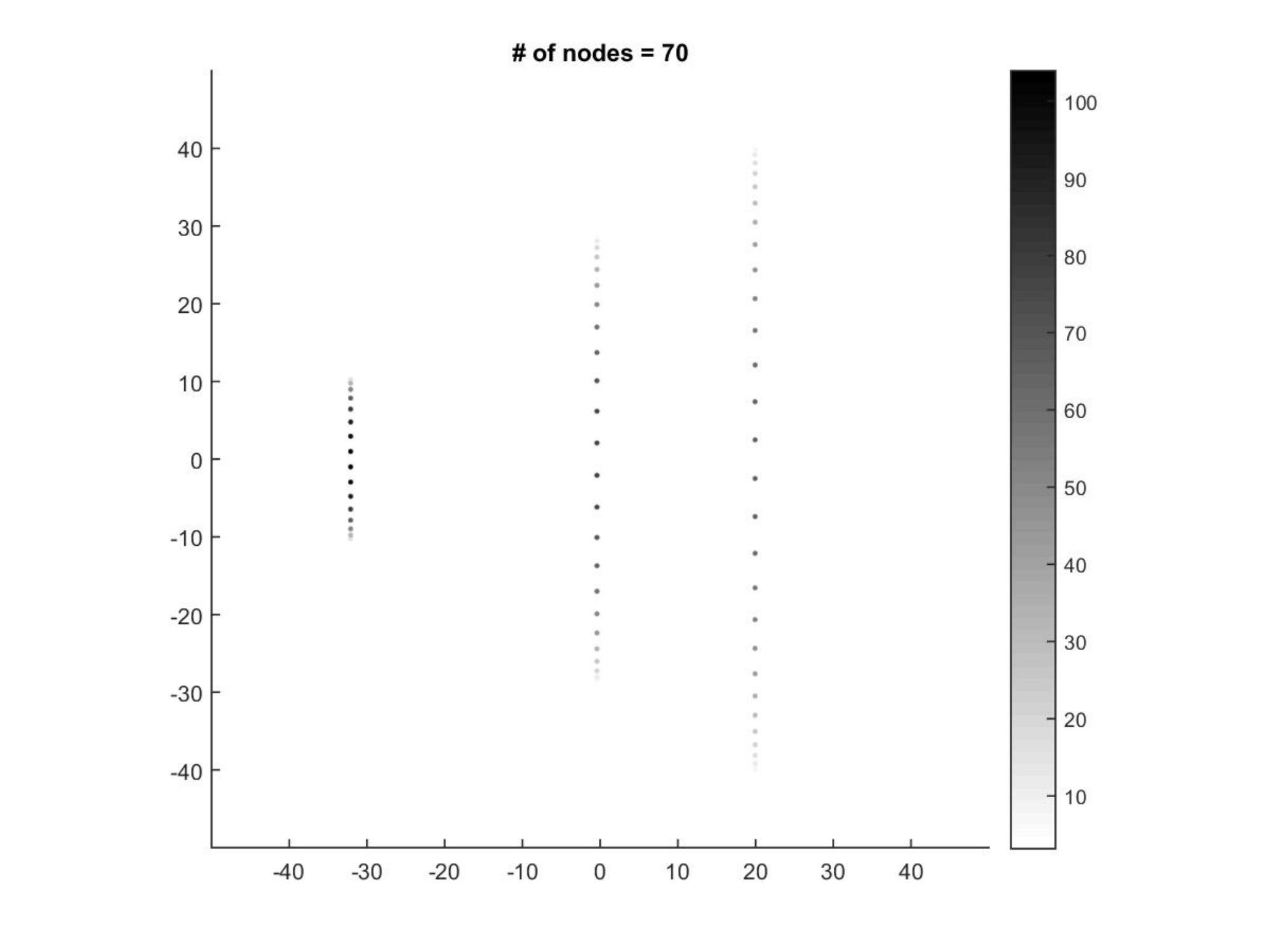}  &
\includegraphics[scale=0.15]{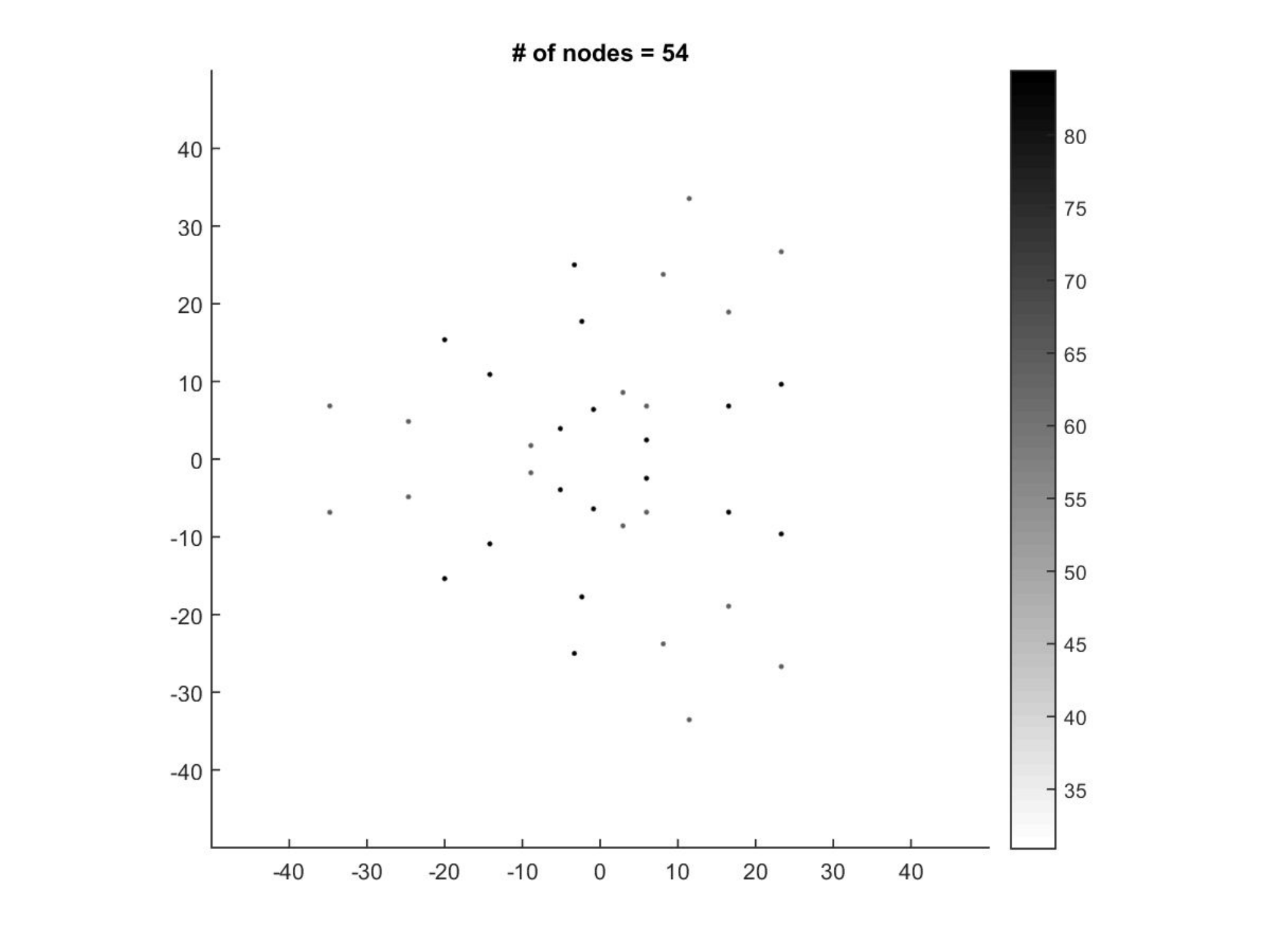}\tabularnewline
\hline 
{\scriptsize{}{}1}  &
\includegraphics[scale=0.15]{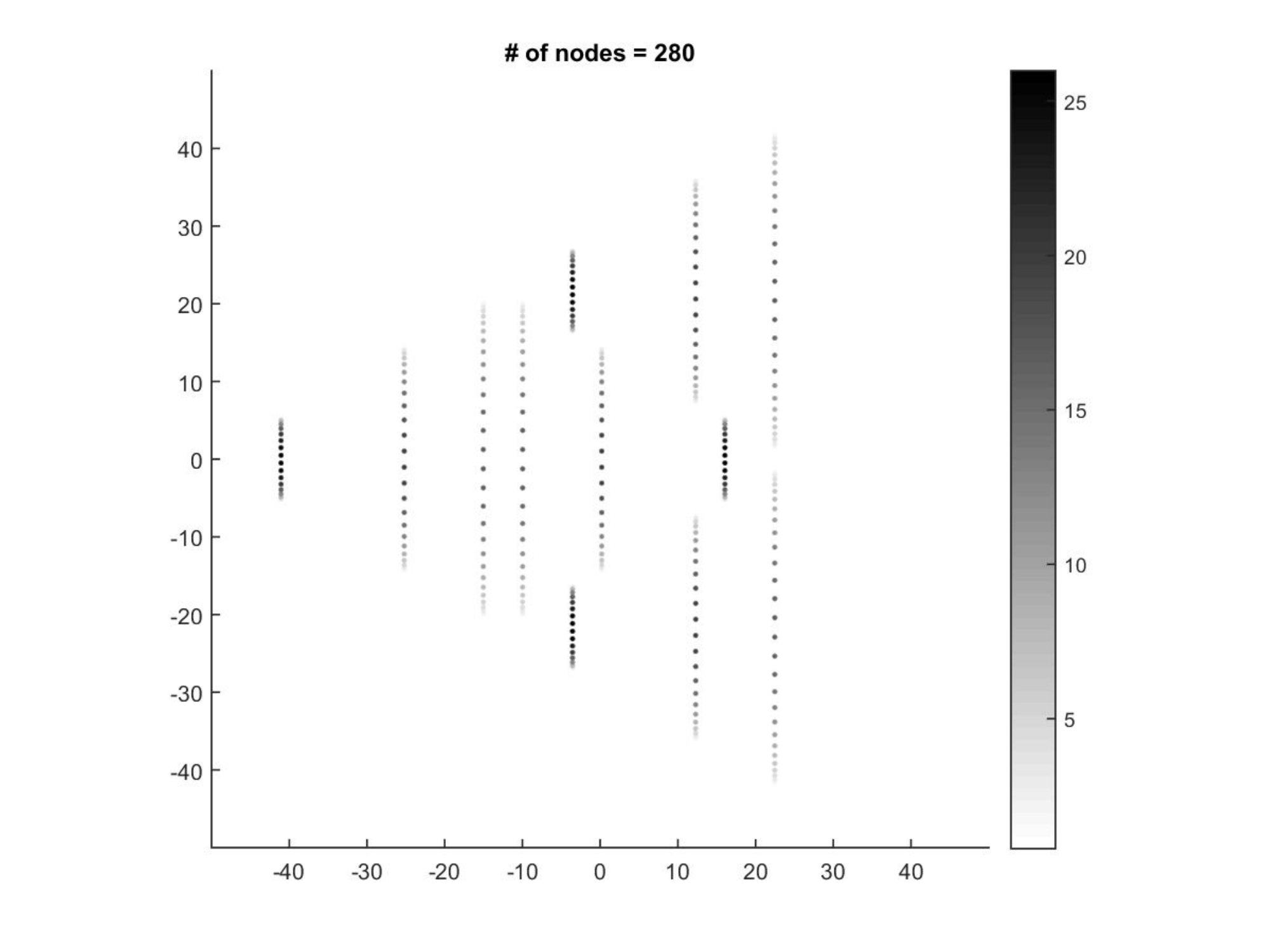}  &
\includegraphics[scale=0.15]{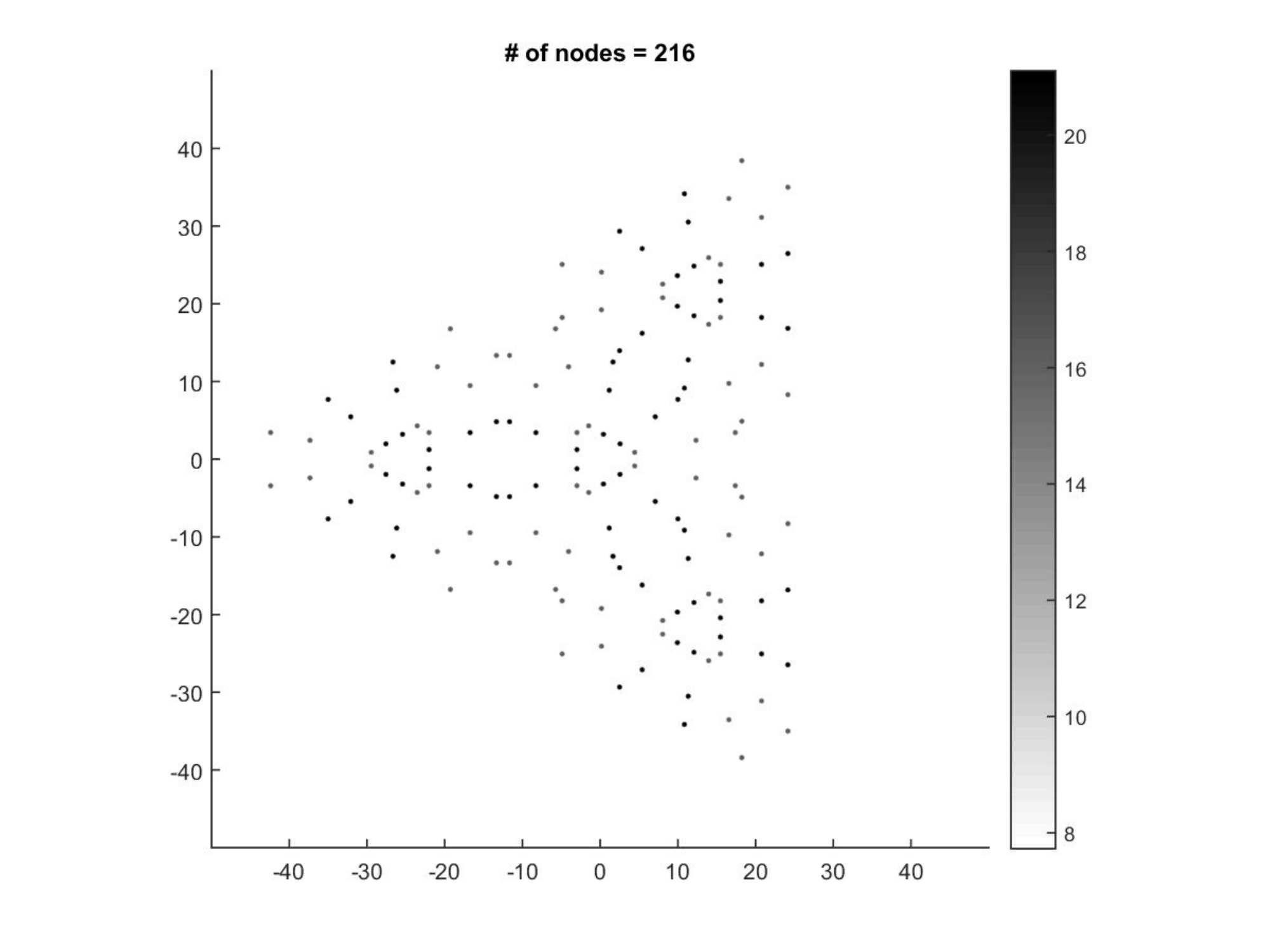}\tabularnewline
\hline 
{\scriptsize{}{}2}  &
\includegraphics[scale=0.15]{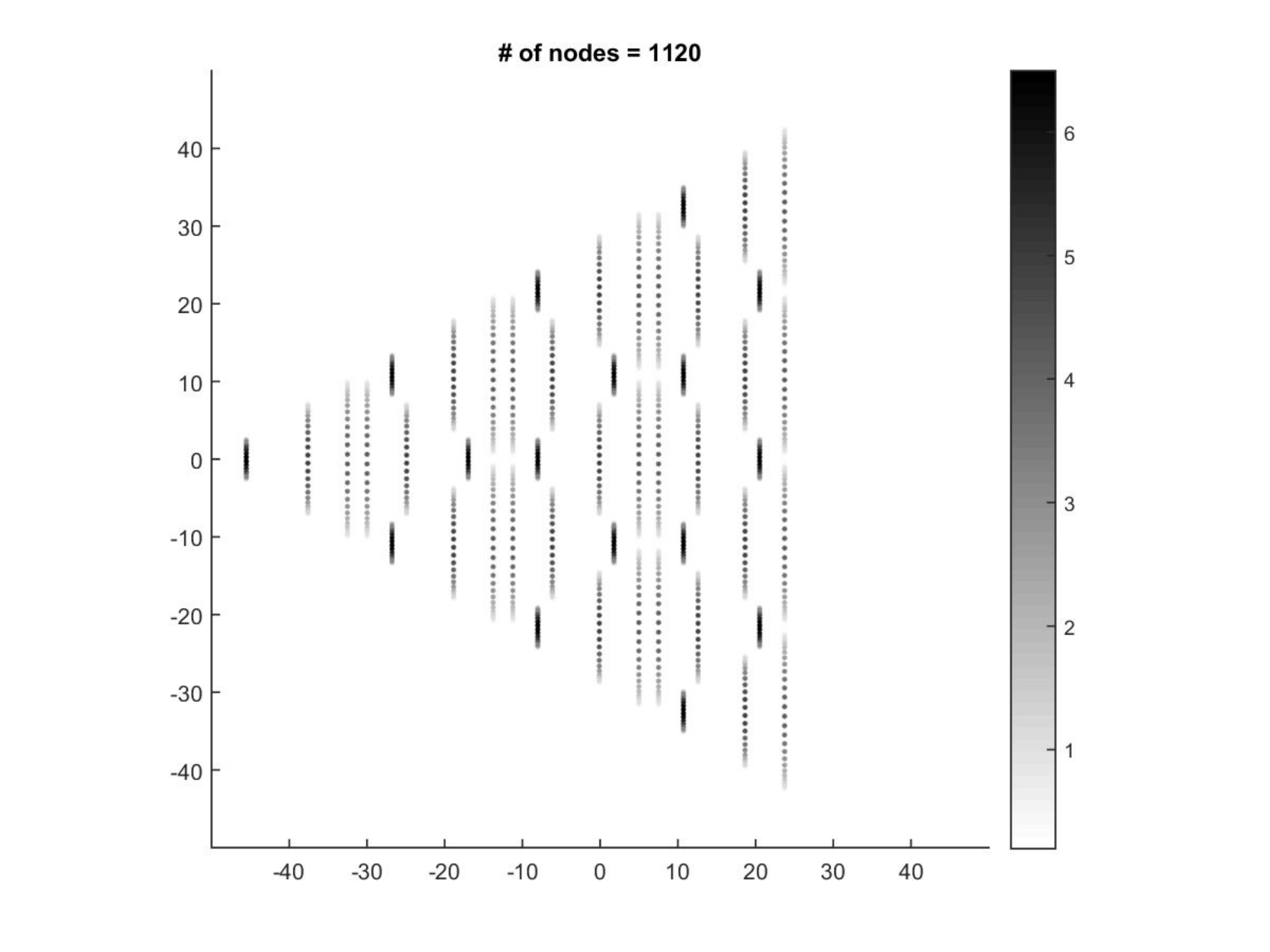}  &
\includegraphics[scale=0.15]{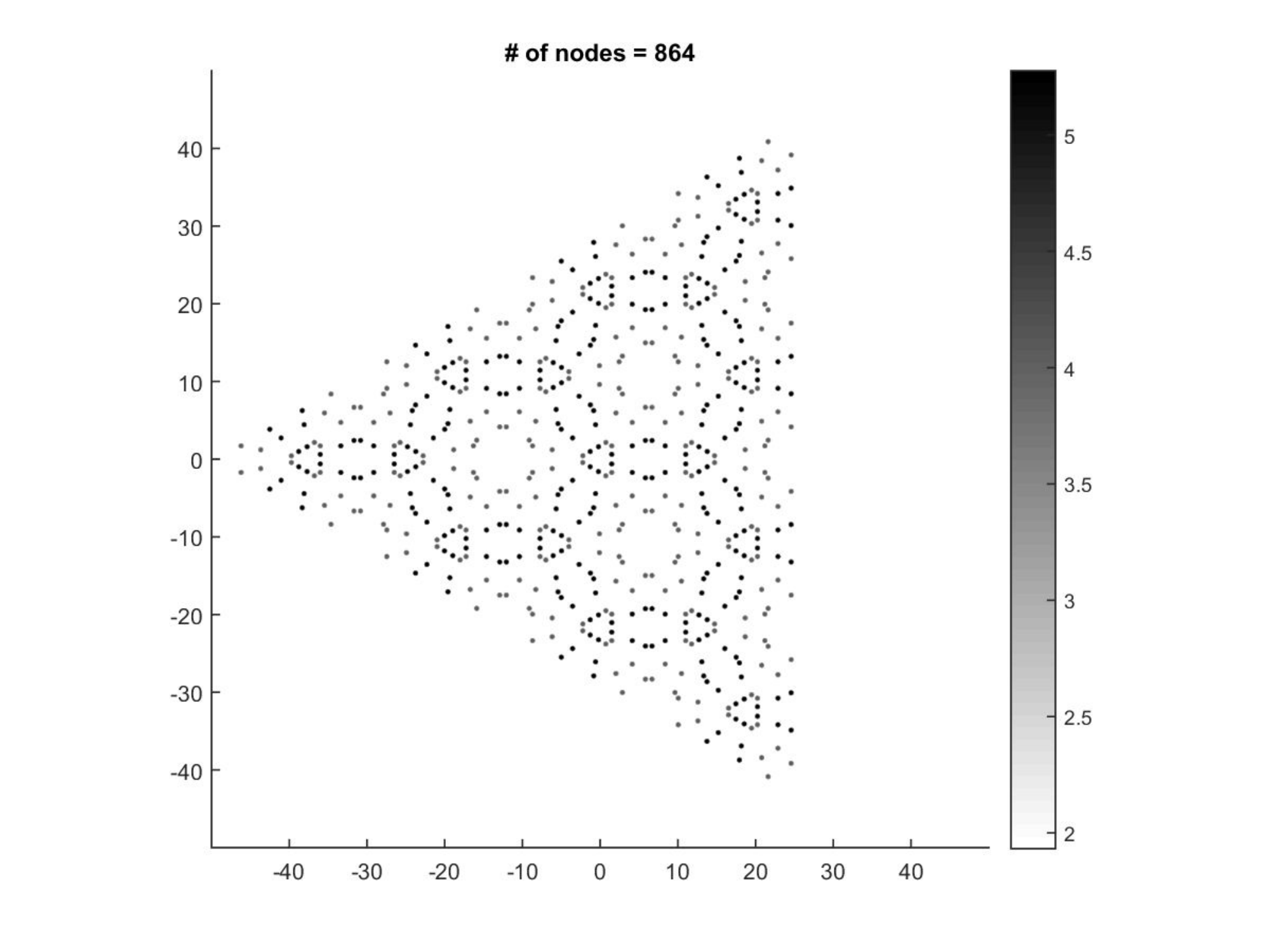}\tabularnewline
\hline 
{\scriptsize{}{}3}  &
\includegraphics[scale=0.15]{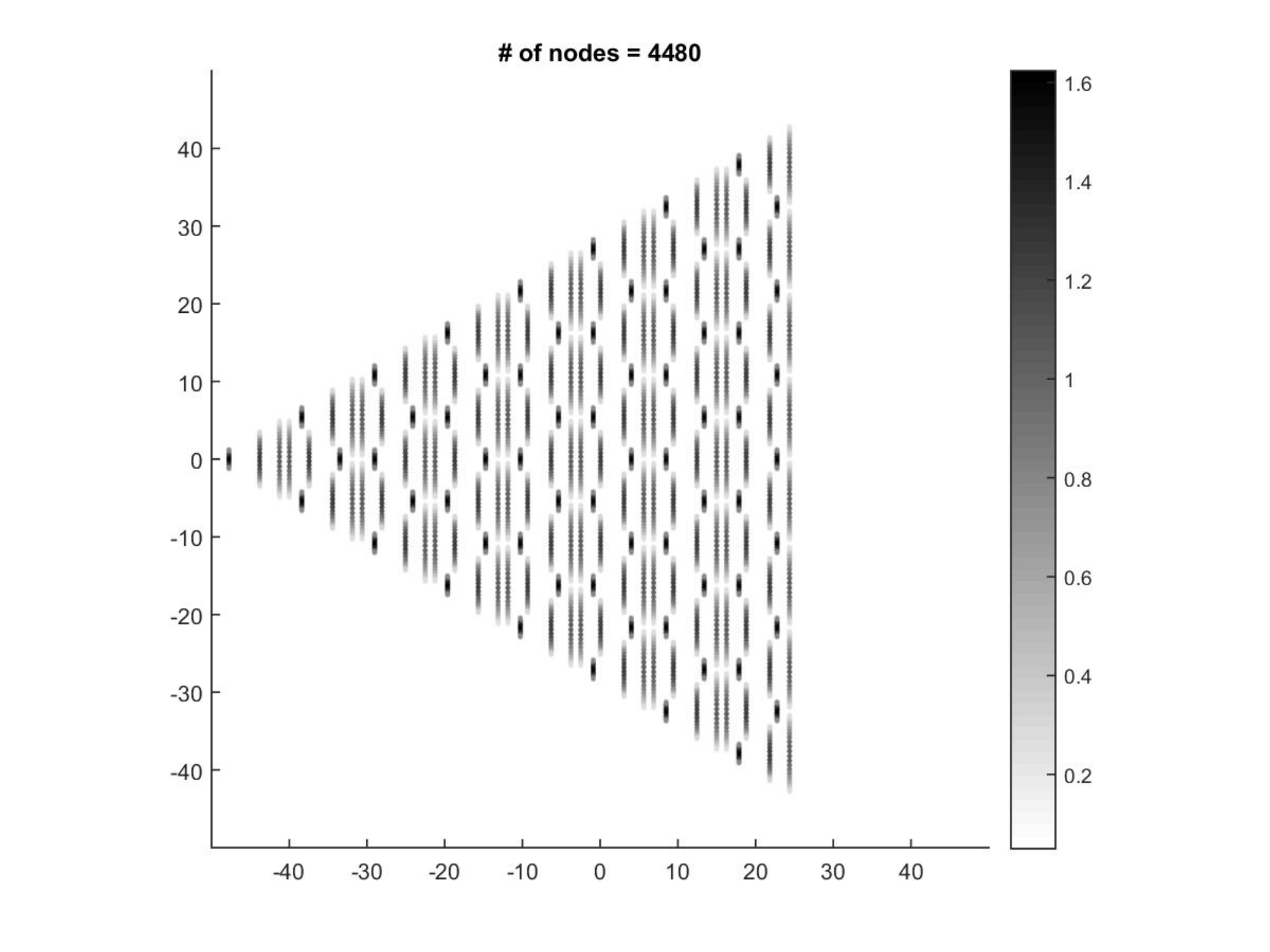}  &
\includegraphics[scale=0.15]{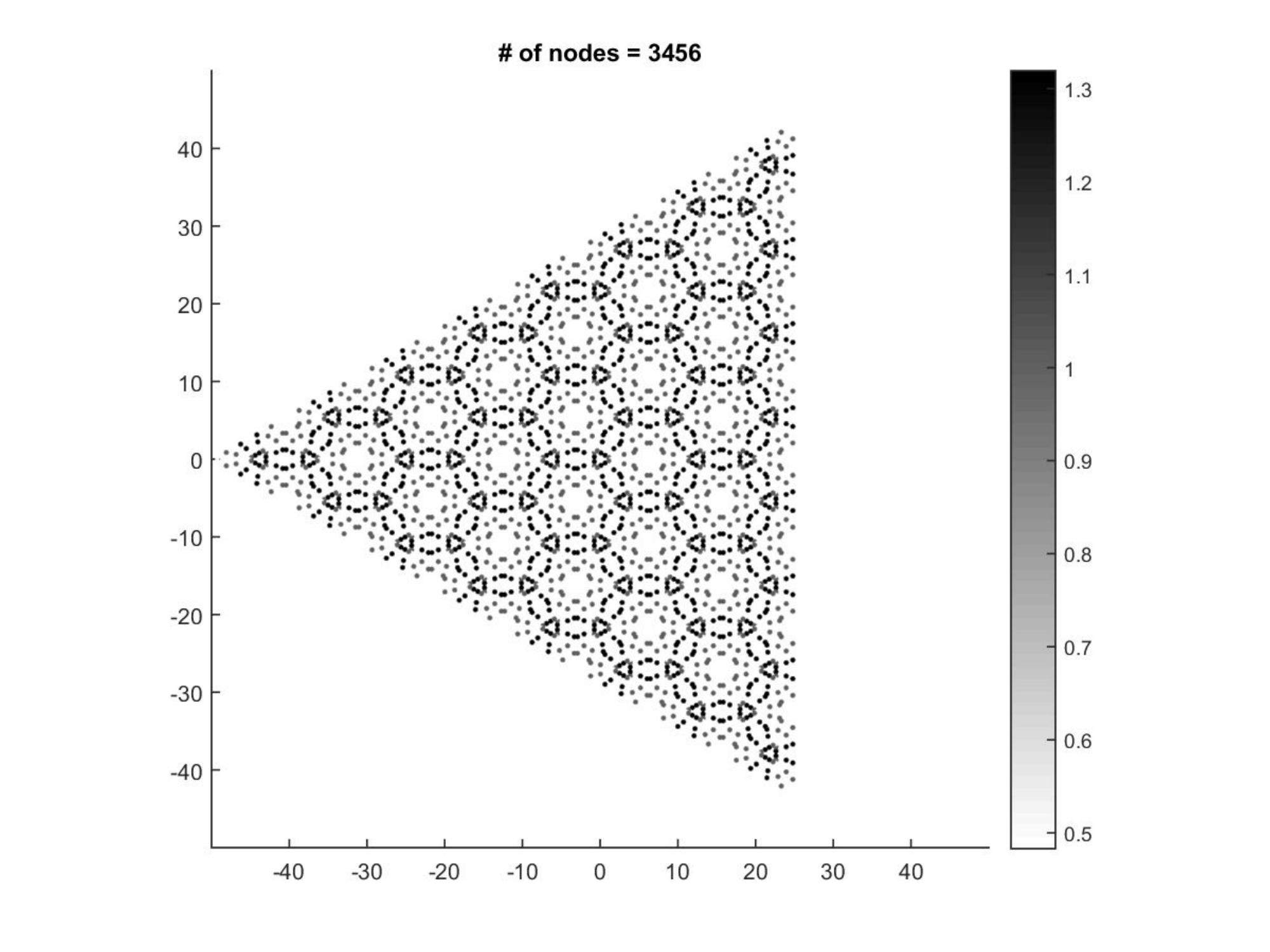}\tabularnewline
\hline 
{\scriptsize{}{}4}  &
\includegraphics[scale=0.15]{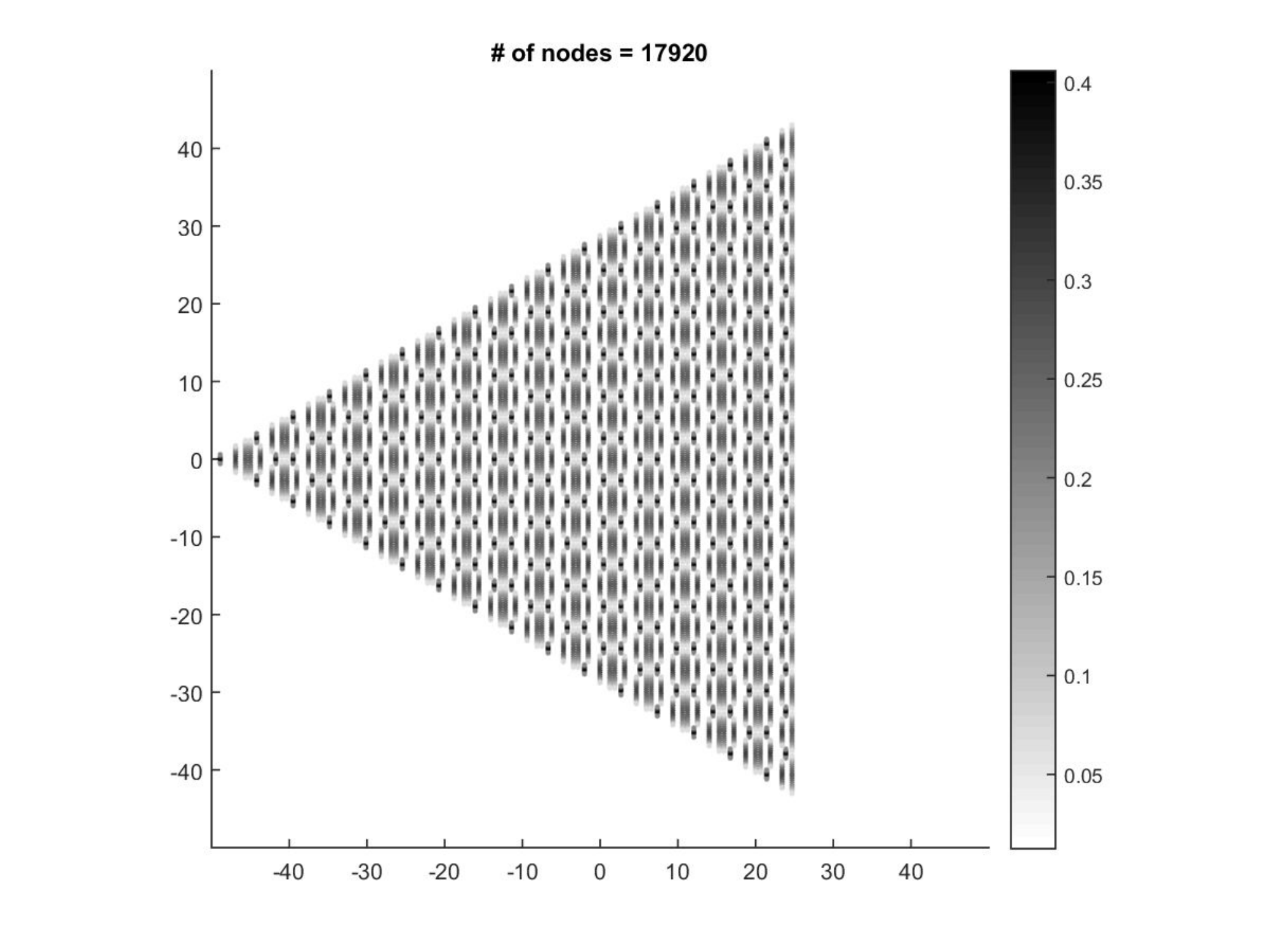}  &
\includegraphics[scale=0.15]{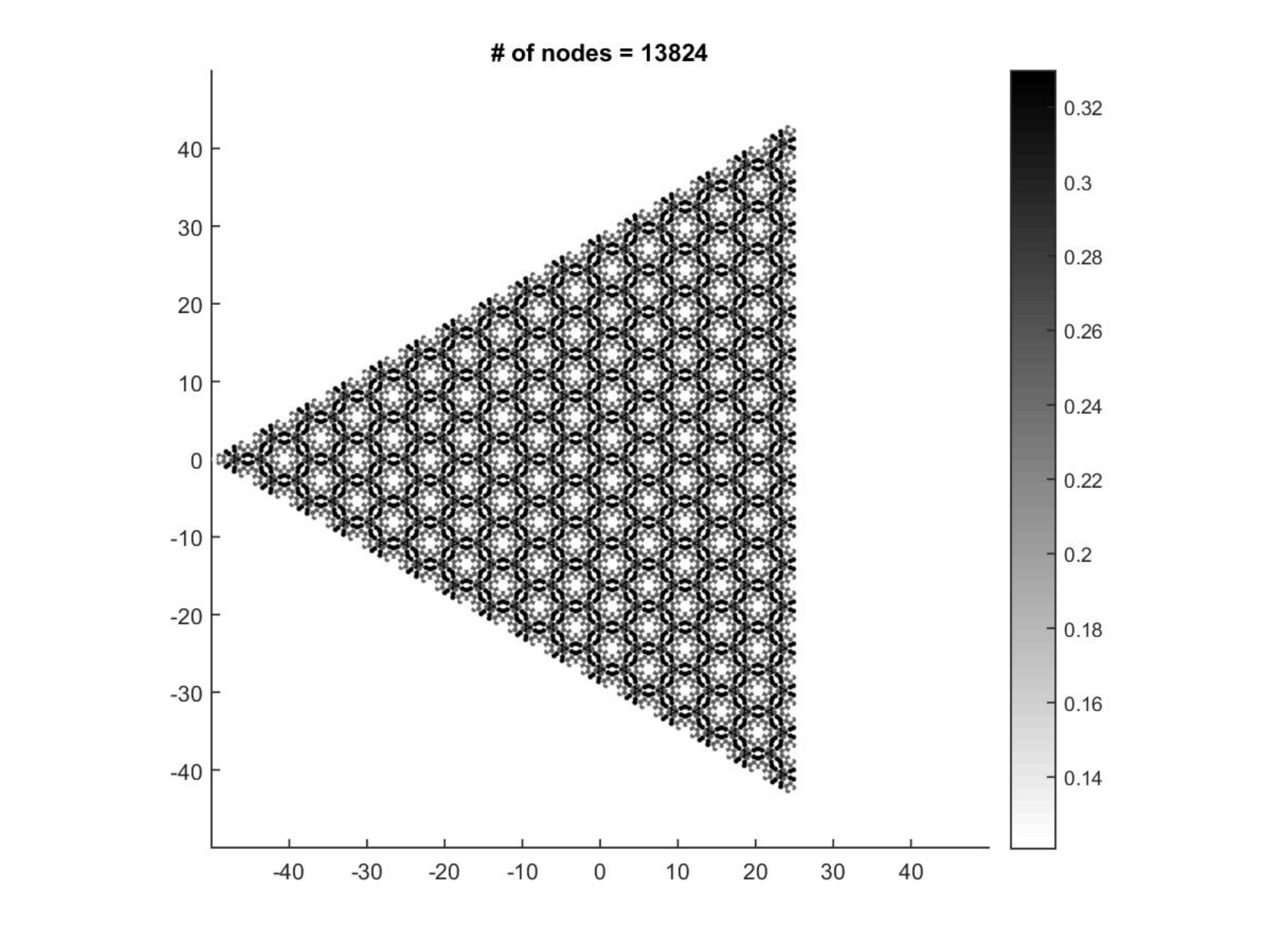}\tabularnewline
\hline 
{\scriptsize{}{}5}  &
\includegraphics[scale=0.15]{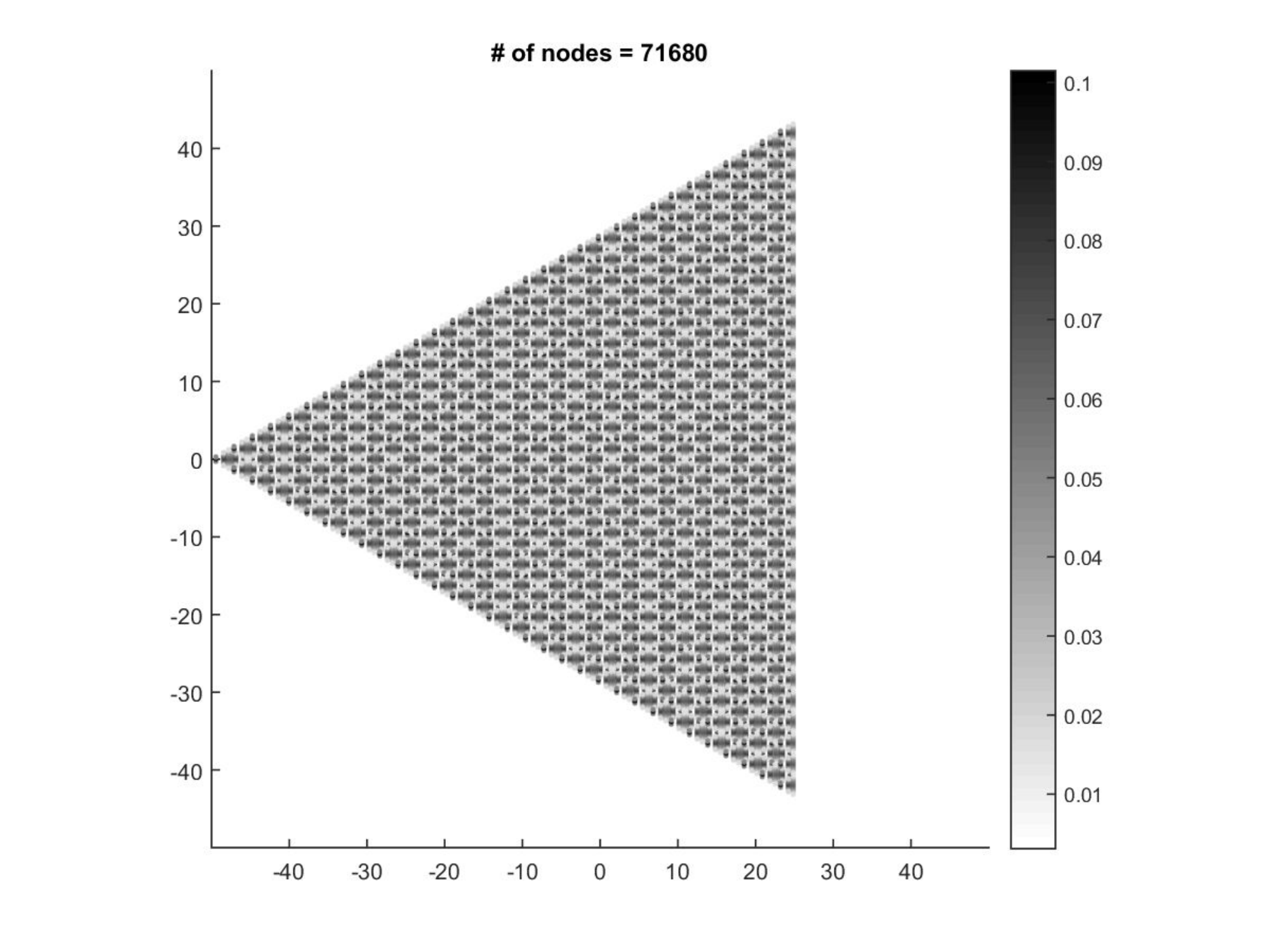}  &
\includegraphics[scale=0.15]{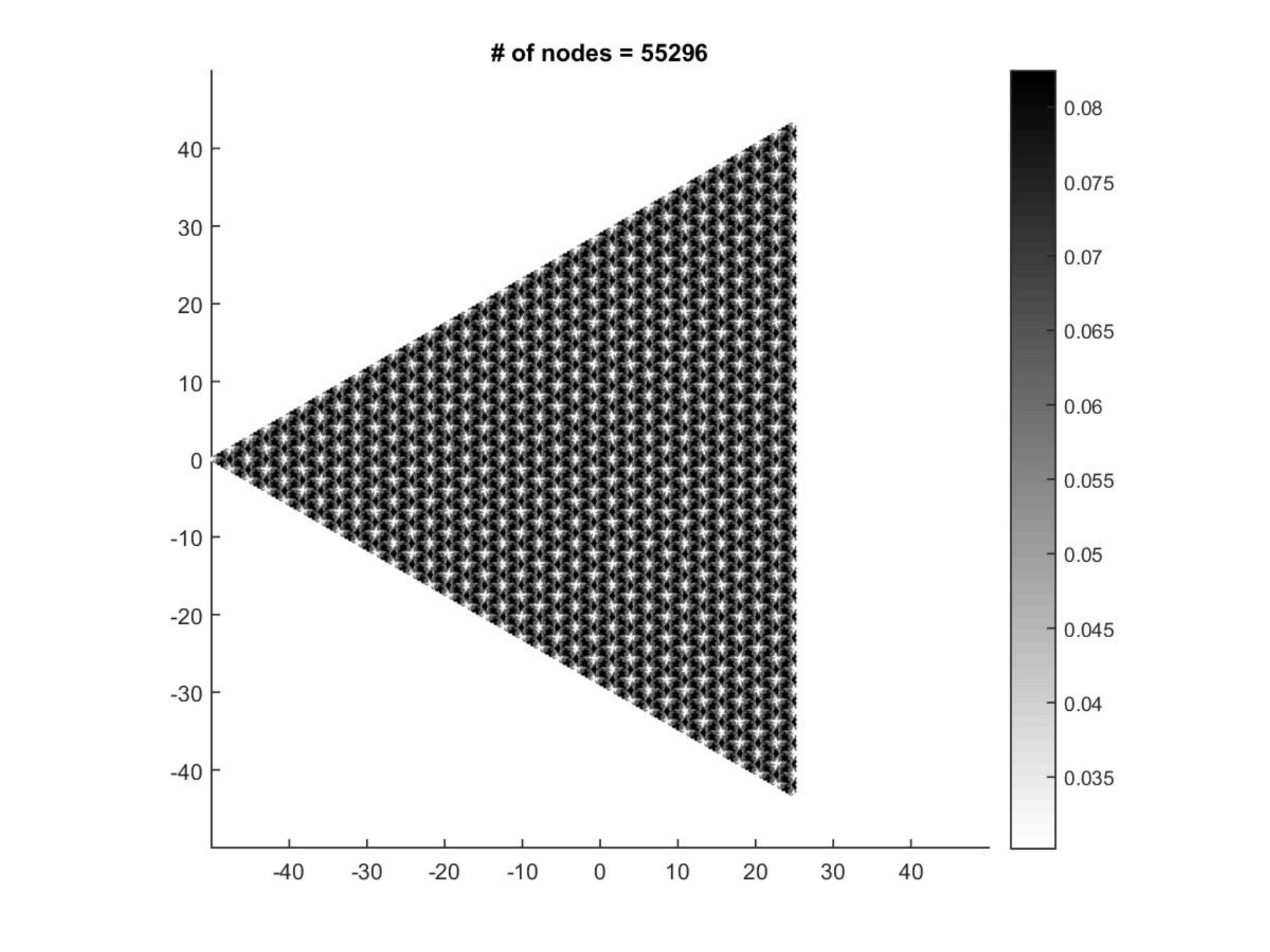}\tabularnewline
\hline 
{\scriptsize{}{}6}  &
\includegraphics[scale=0.15]{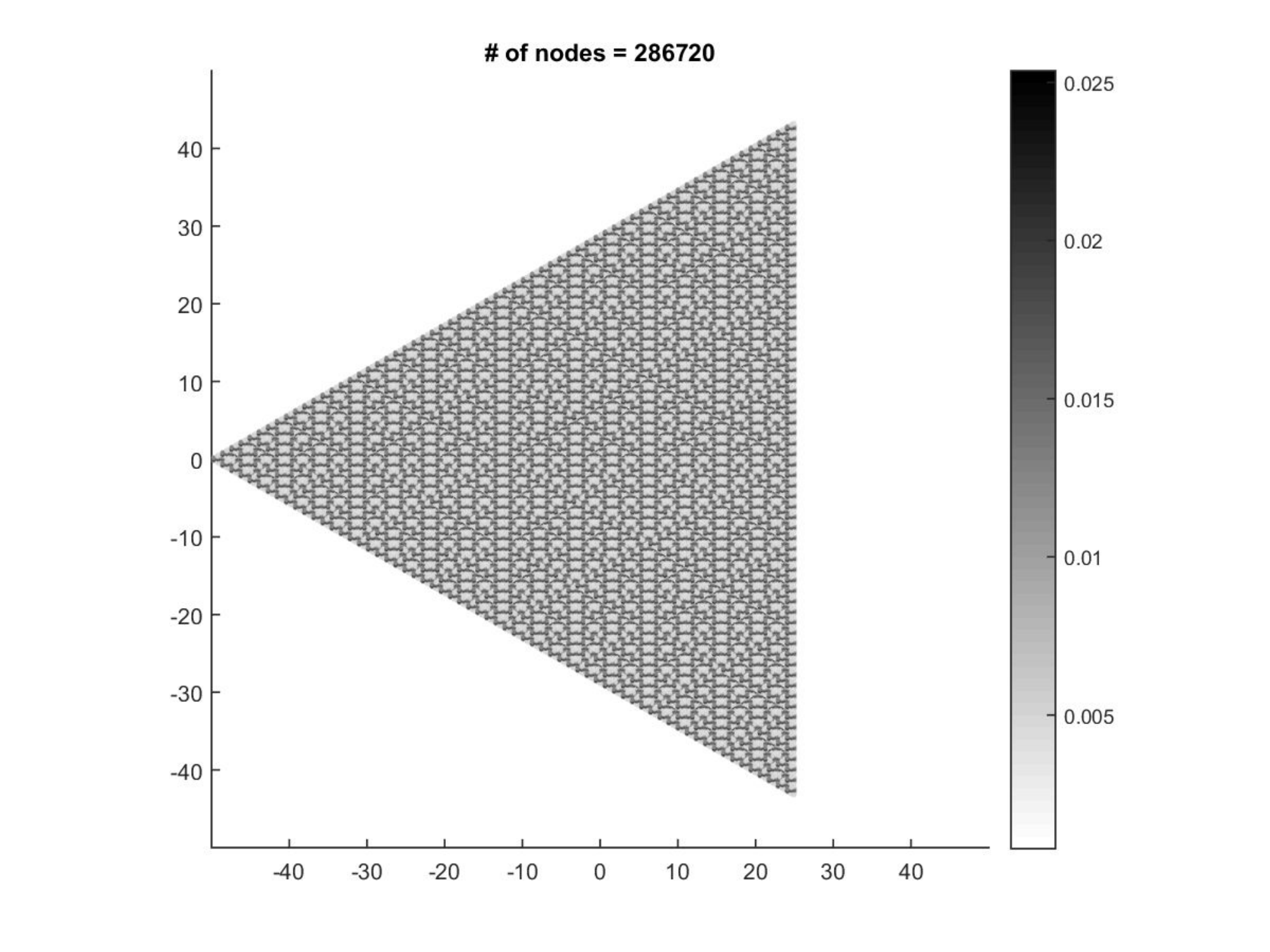}  &
\includegraphics[scale=0.15]{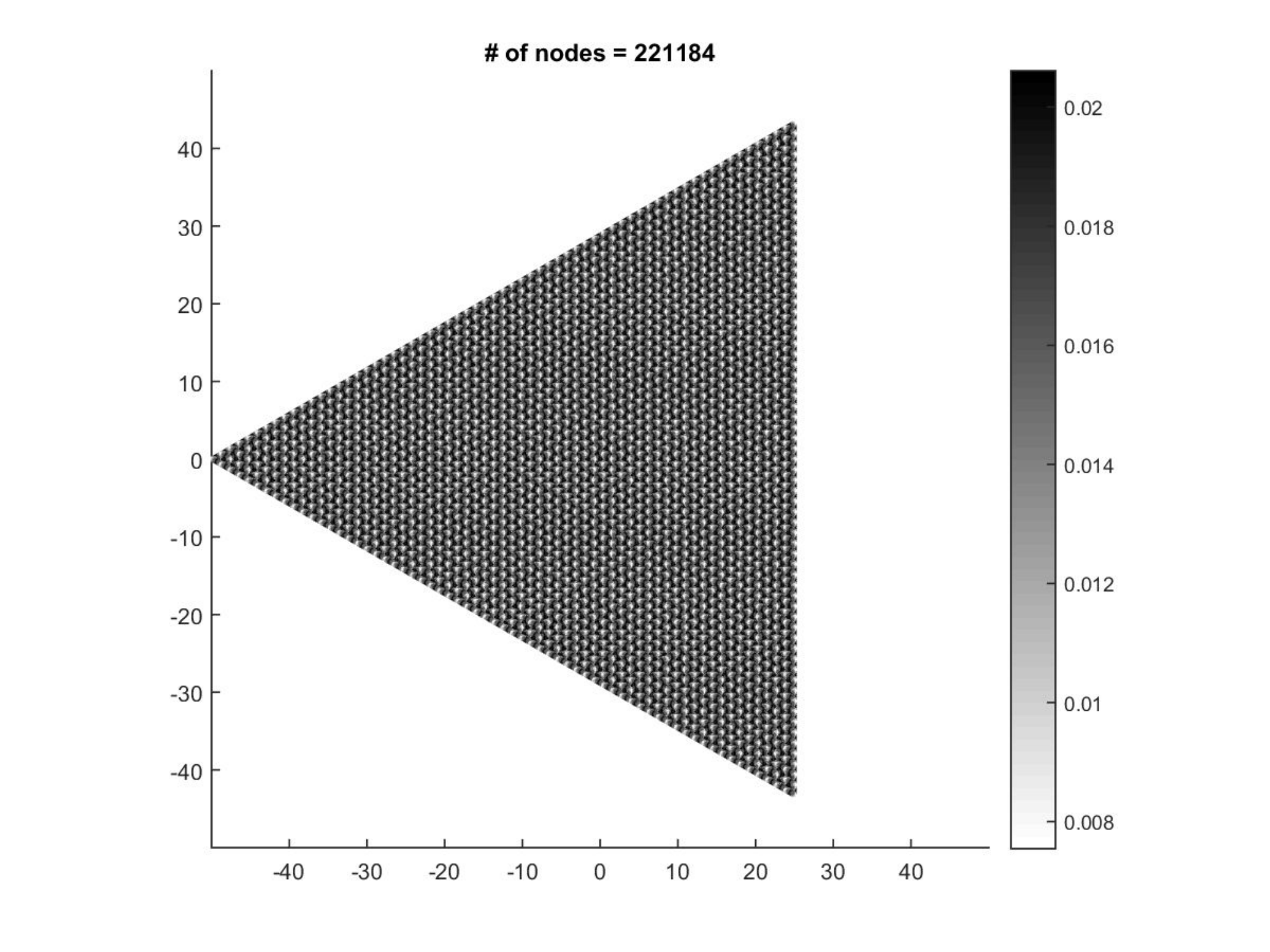}\tabularnewline
\hline 
{\scriptsize{}{}7}  &
\includegraphics[scale=0.15]{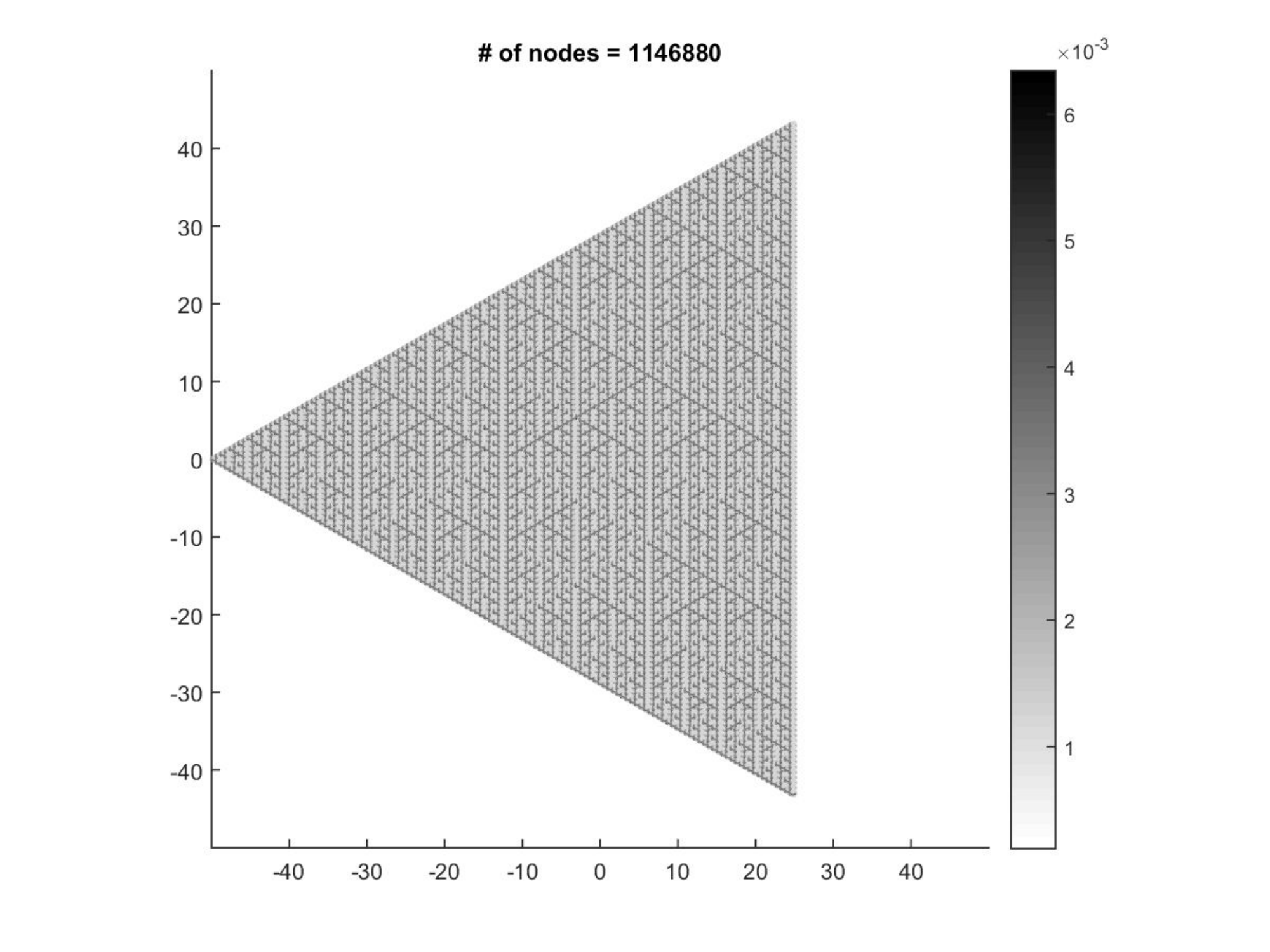}  &
\includegraphics[scale=0.15]{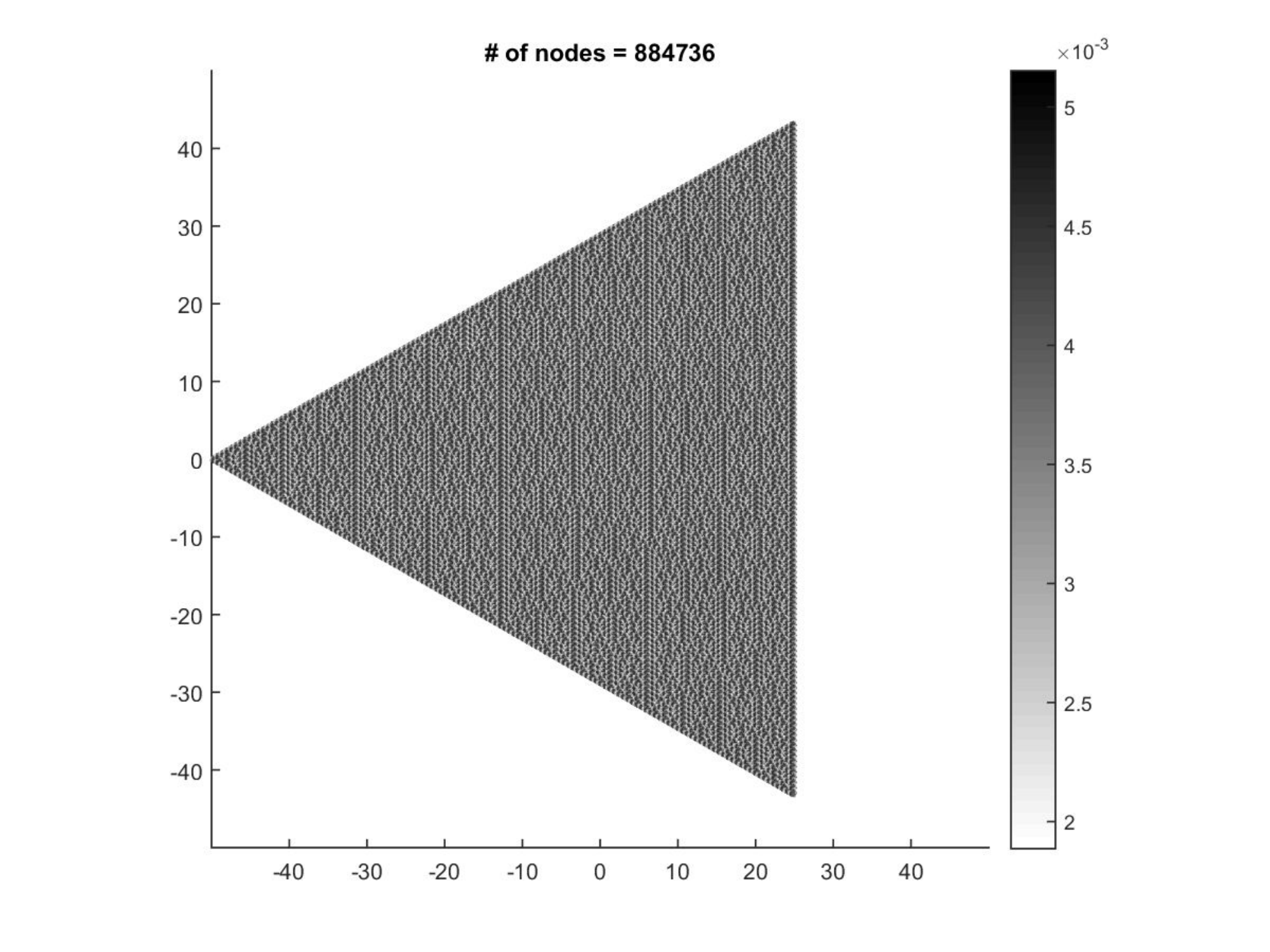}\tabularnewline
\hline 
\end{tabular}

\caption{Rotationally not invariant (RNI column) and rotationally invariant
(RI column) Fourier quadratures for $K_{\triangle}\left(x,y\right)$.
The horizontal axis is $k_{x}$ and the vertical axis is $k_{y}$.\label{fig:Nodes-for-equilateral}}
\end{figure}

In order to preserve the rotational invariance of the equilateral
triangle among the nodes, one can construct nodes for the isosceles
triangle $T_{I}$ which is triangle $T$ of (\ref{eq:T}) with $\Delta p=\sqrt{3}/6$
and $s=\sqrt{3}^{-1}$, then rotate these nodes by $2\pi/3$ and $4\pi/3$
to construct nodes satisfying the rotational invariance of the equilateral
triangle (see Figure \ref{fig:Nodes-for-equilateral}) as illustrated
in Figure \ref{fig:rotational-invariance-of-nodes}. Our observation
is, for the same or less number of nodes, the nodes with rotational
symmetry provide a more accurate discretization of the Fourier approximation
of the kernel $K_{\triangle}\left(x,y\right)$ compared to nodes without
rotational symmetry (see Figures \ref{fig:RNS-quad-approx-K_triangle}
and \ref{fig:RS-quad-approx-K_triangle}).

\begin{figure}
\center

\includegraphics[scale=0.4]{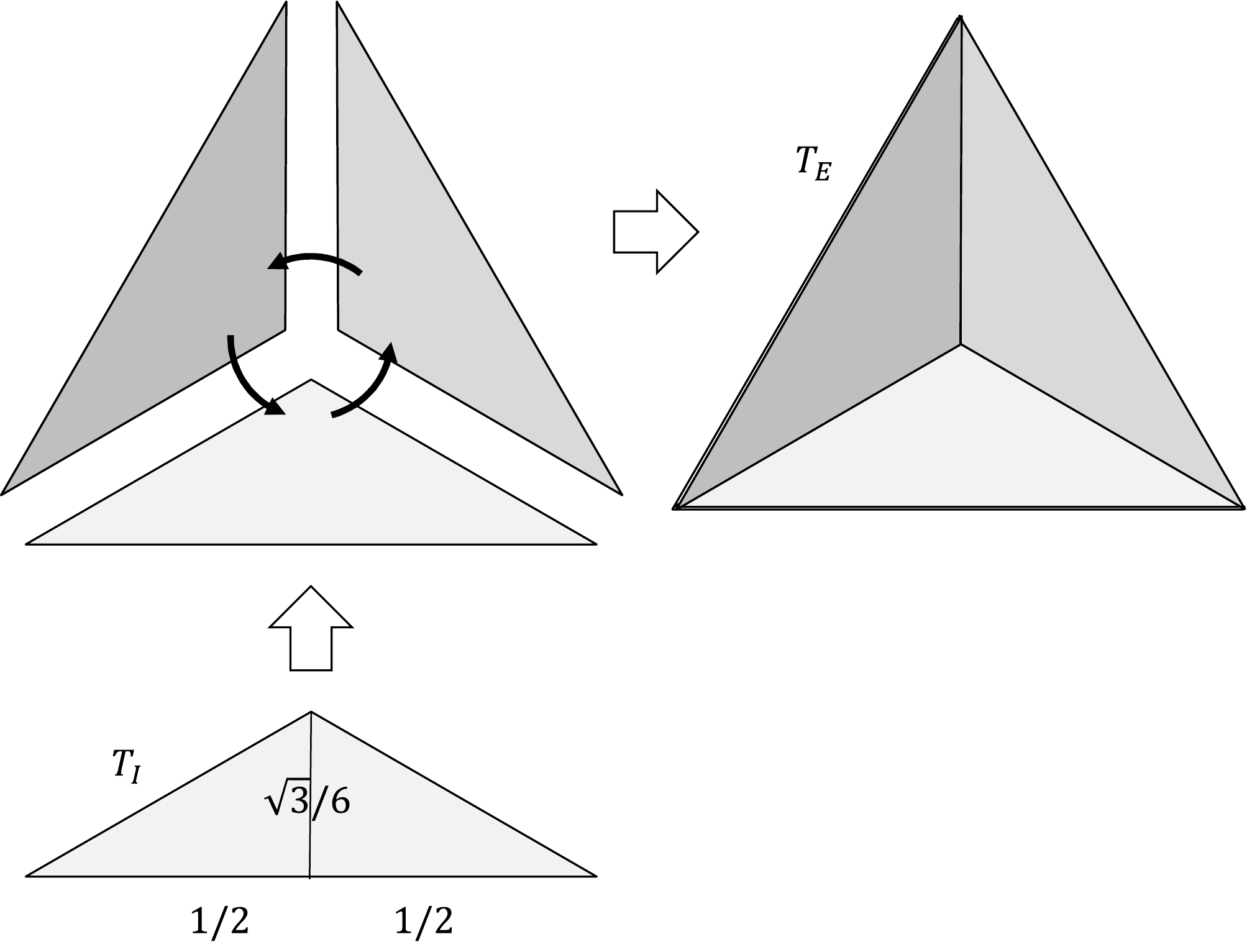}\caption{Construction of nodes that satisfy the rotational invariance of equilateral
triangle can be achieved by first constructing nodes for the lightest
gray triangle followed by rotating and accumulating the constructed
nodes. \label{fig:rotational-invariance-of-nodes}}
\end{figure}

\begin{figure}
\center

\begin{tabular}{|>{\centering}m{5mm}|c|}
\hline 
{\scriptsize{}{}$m$}  &
{\scriptsize{}{}$\tilde{K}_{m}\left(2^{-m}x,2^{-m}y\right)$}\tabularnewline
\hline 
\hline 
{\scriptsize{}{}0}  &
\includegraphics[scale=0.45]{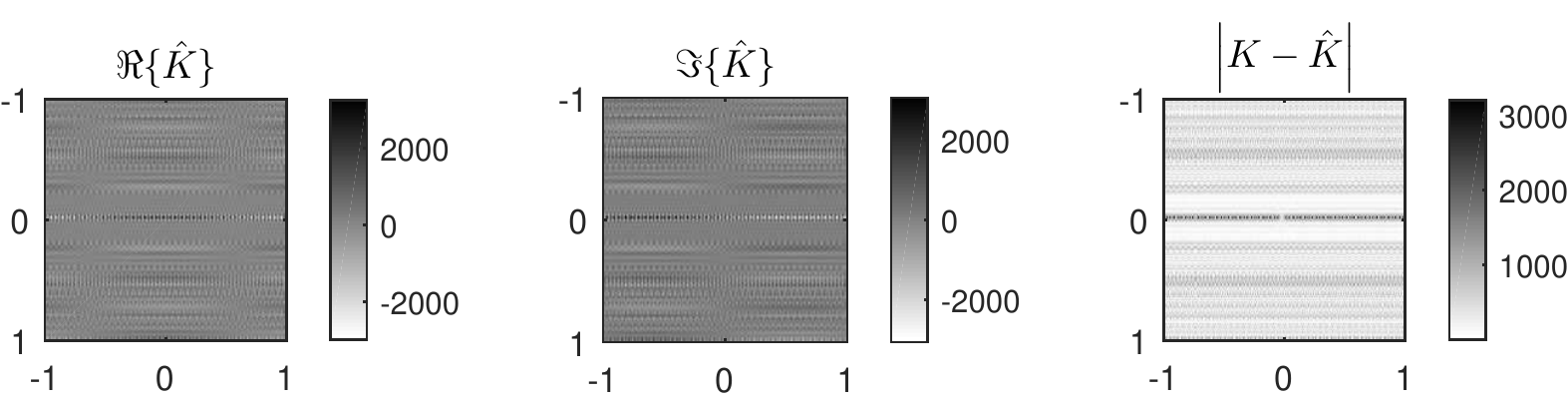}\tabularnewline
\hline 
{\scriptsize{}{}1}  &
\includegraphics[scale=0.45]{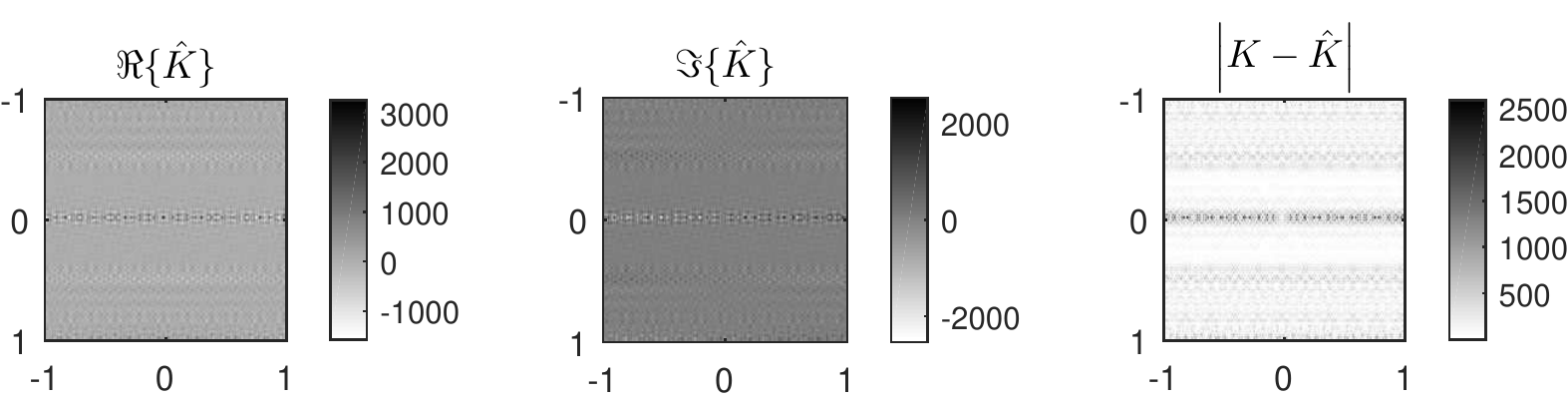}\tabularnewline
\hline 
{\scriptsize{}{}2}  &
\includegraphics[scale=0.45]{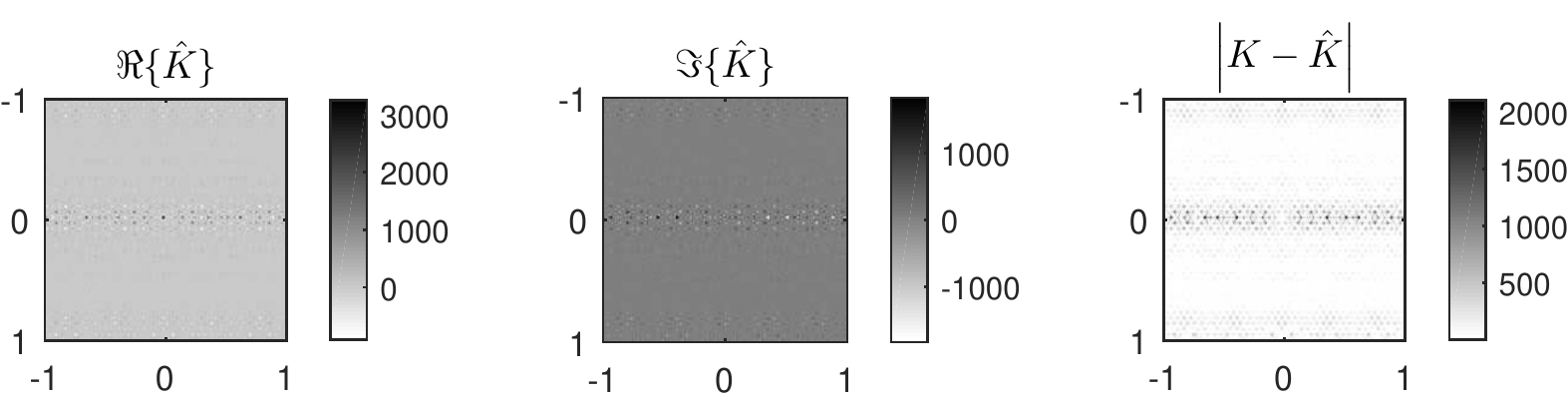}\tabularnewline
\hline 
{\scriptsize{}{}3}  &
\includegraphics[scale=0.45]{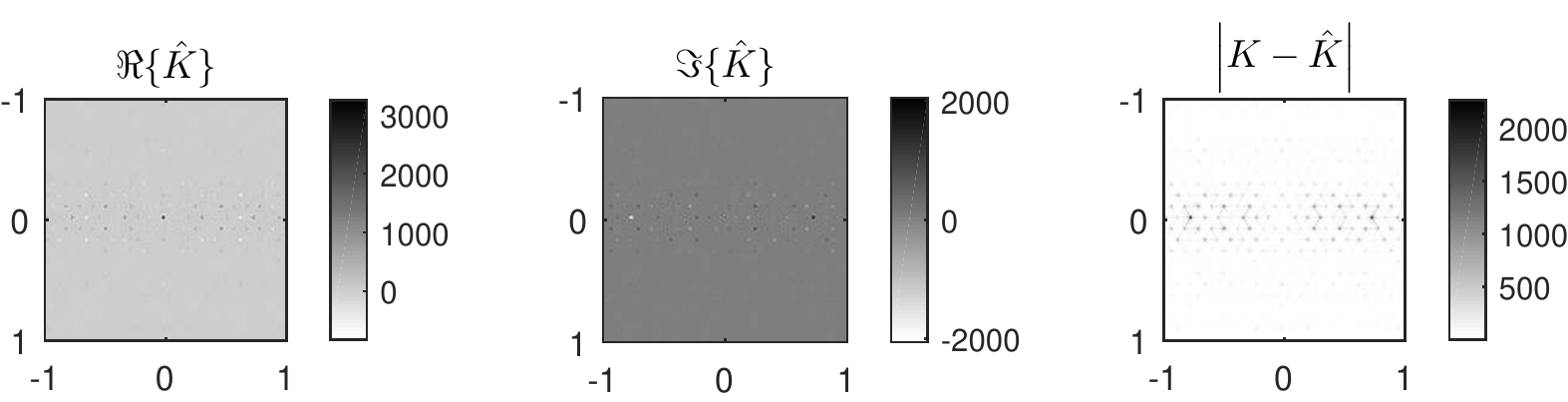}\tabularnewline
\hline 
{\scriptsize{}{}4}  &
\includegraphics[scale=0.45]{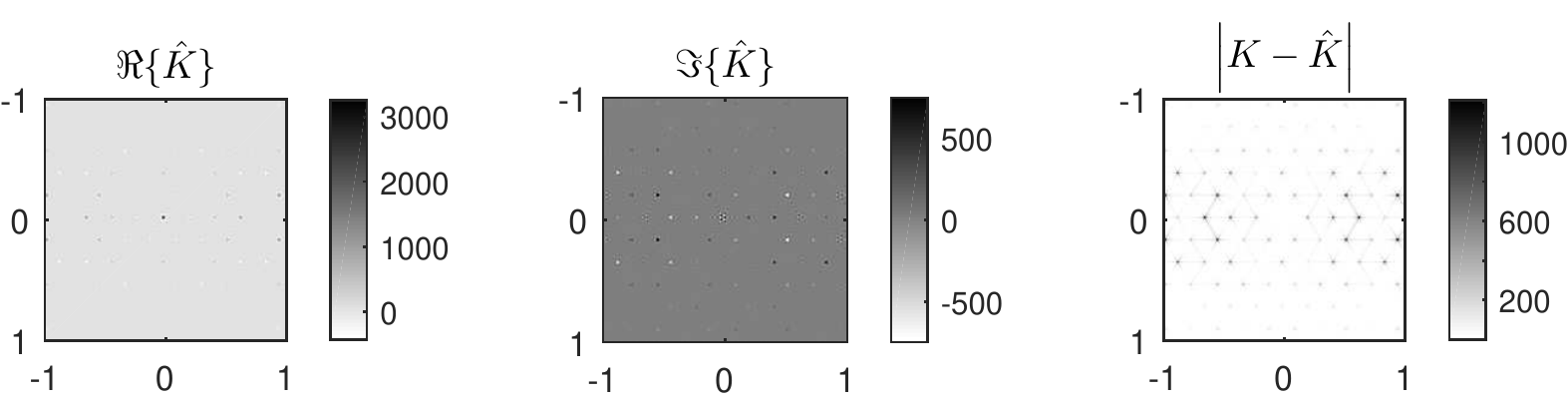}\tabularnewline
\hline 
{\scriptsize{}{}5}  &
\includegraphics[scale=0.45]{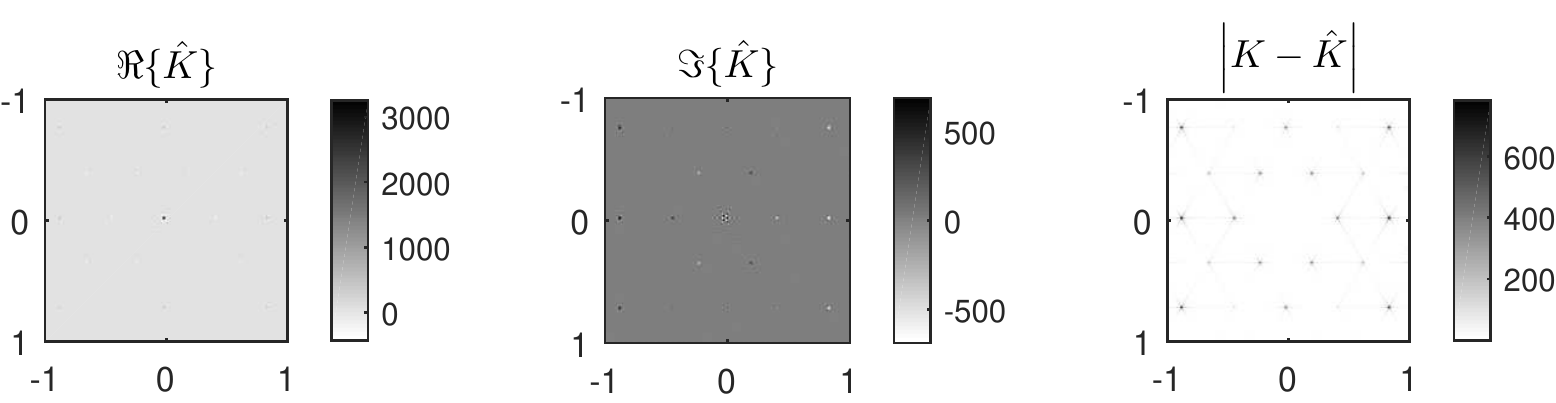}\tabularnewline
\hline 
{\scriptsize{}{}6}  &
\includegraphics[scale=0.45]{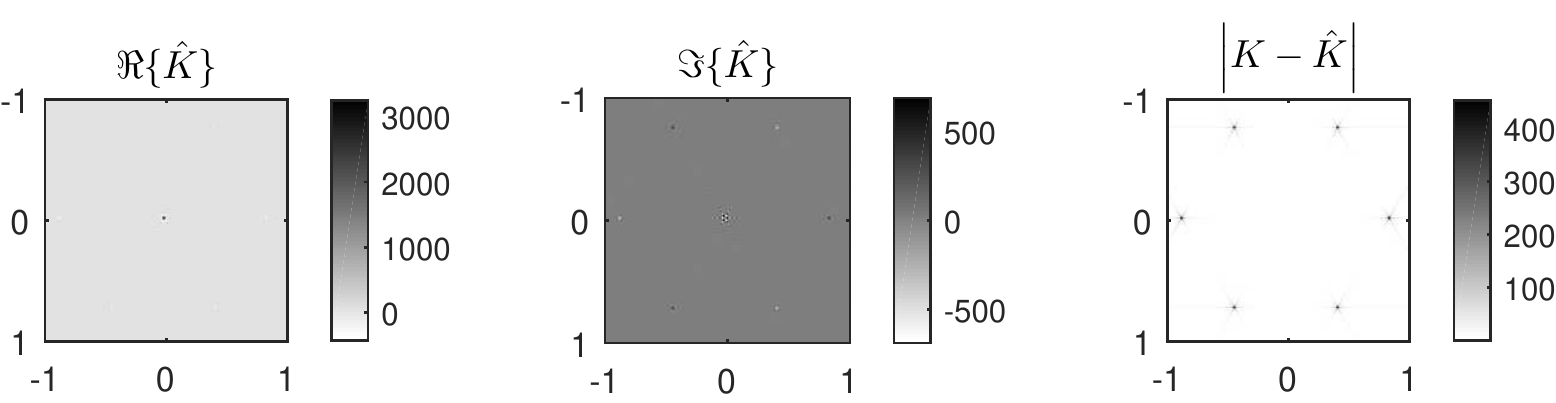}\tabularnewline
\hline 
{\scriptsize{}{}7}  &
\includegraphics[scale=0.45]{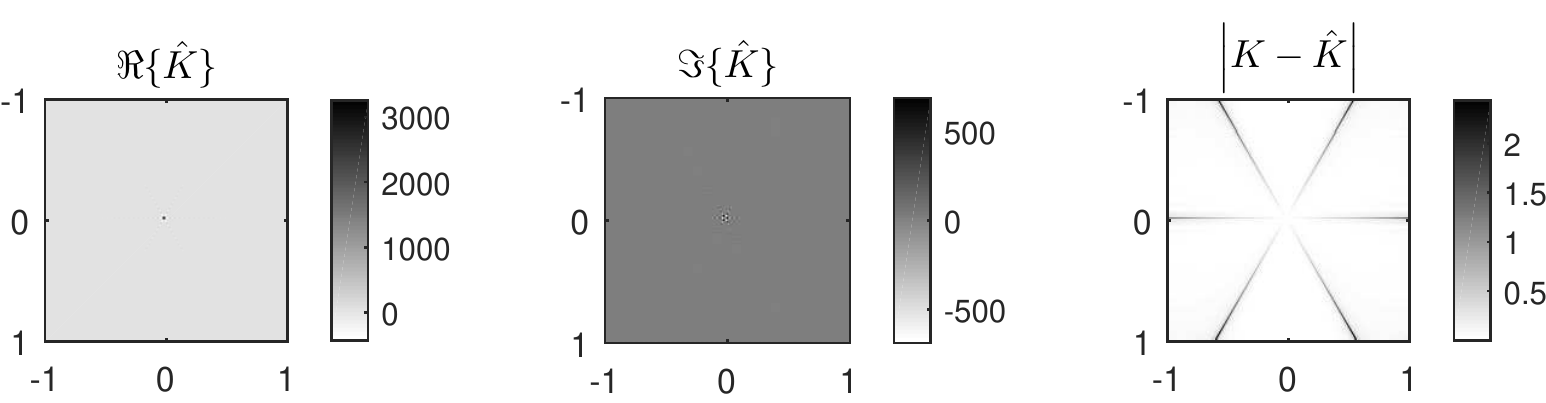}\tabularnewline
\hline 
\end{tabular}

\caption{Approximation of the kernel $K_{\triangle}\left(x,y\right)$ by $\tilde{K}_{\triangle,m}\left(2^{-m}x,2^{-m}y\right)$
using the RNI quadratures in Figure \ref{fig:Nodes-for-equilateral}.
The horizontal axis is $x$ and the vertical axis is $y$. \label{fig:RNS-quad-approx-K_triangle}}
\end{figure}

\begin{figure}
\center

\begin{tabular}{|c|c|}
\hline 
{\scriptsize{}{}$m$}  &
{\scriptsize{}{}$\tilde{K}_{\triangle,m}\left(2^{-m}x,2^{-m}y\right)$}\tabularnewline
\hline 
\hline 
{\scriptsize{}{}0}  &
\includegraphics[scale=0.45]{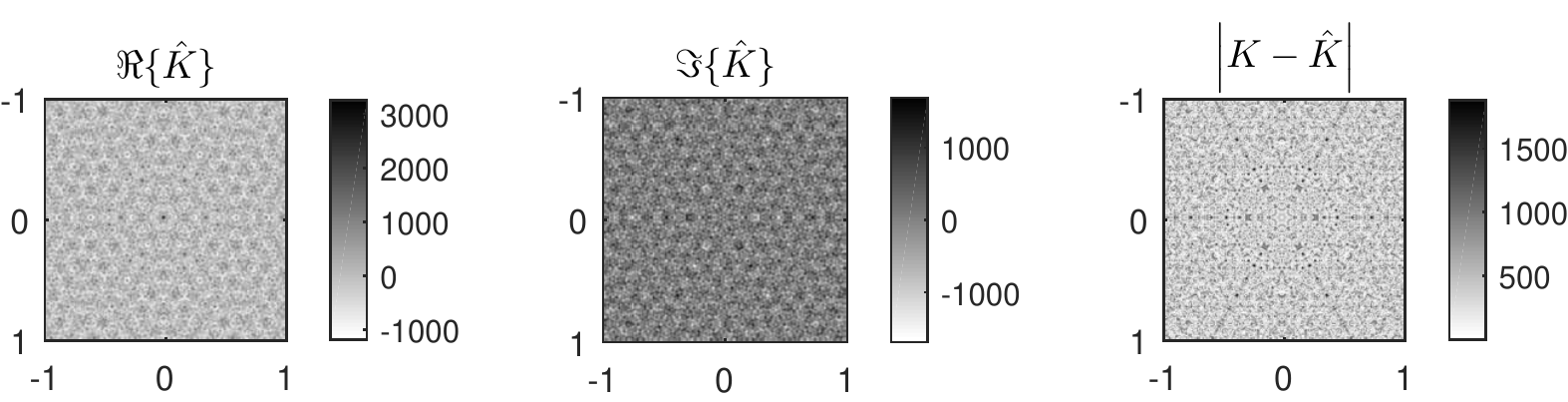}\tabularnewline
\hline 
{\scriptsize{}{}1}  &
\includegraphics[scale=0.45]{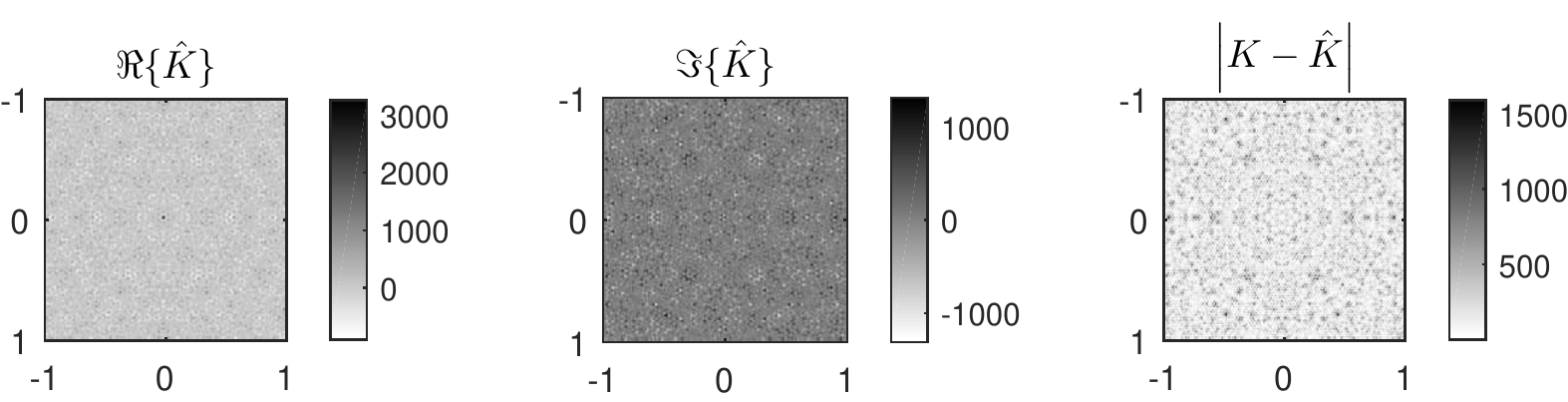}\tabularnewline
\hline 
{\scriptsize{}{}2}  &
\includegraphics[scale=0.45]{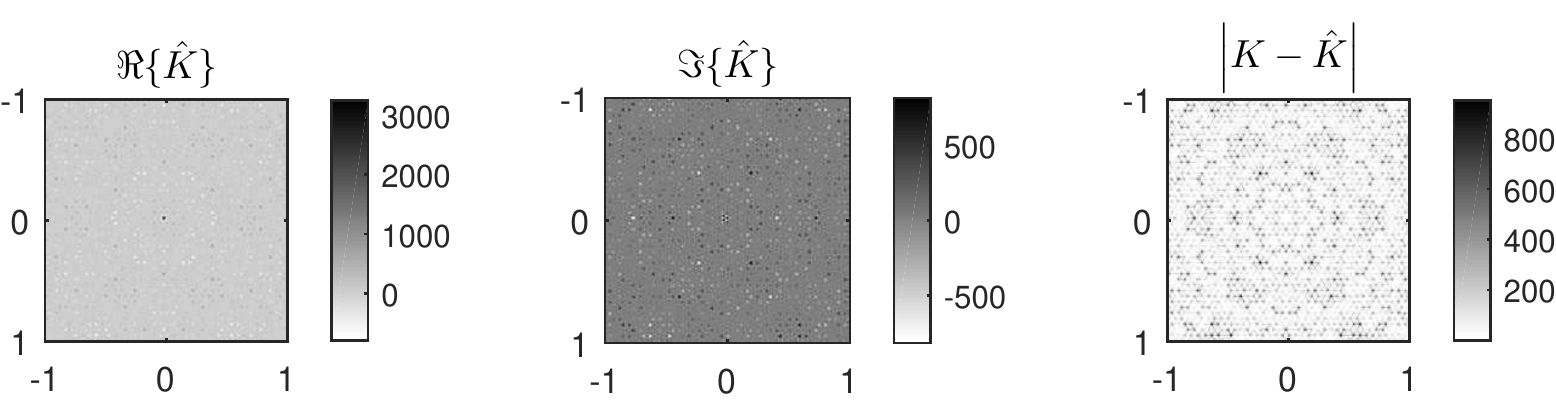}\tabularnewline
\hline 
{\scriptsize{}{}3}  &
\includegraphics[scale=0.45]{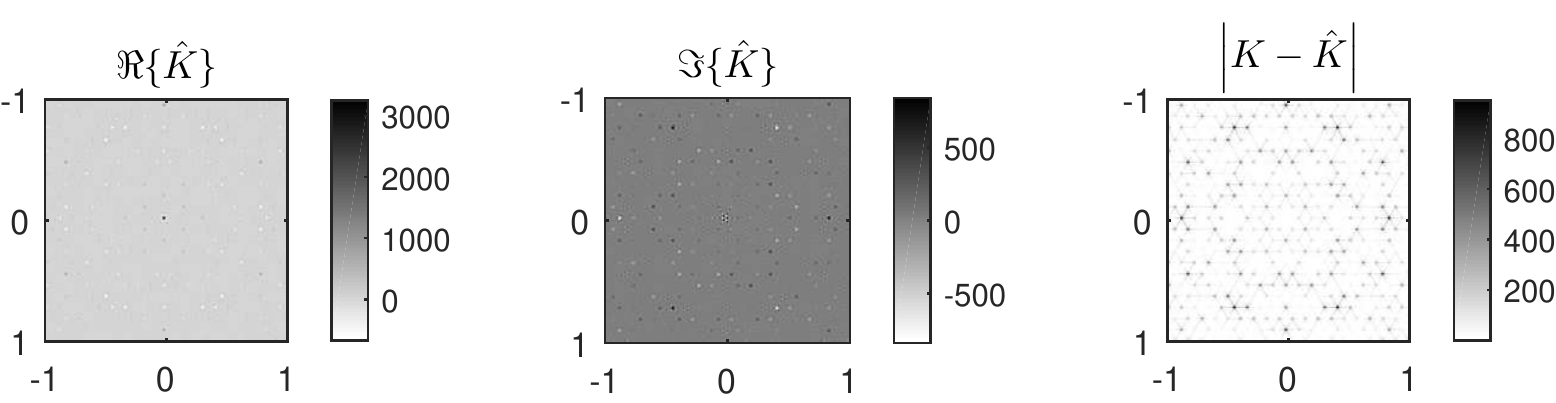}\tabularnewline
\hline 
{\scriptsize{}{}4}  &
\includegraphics[scale=0.45]{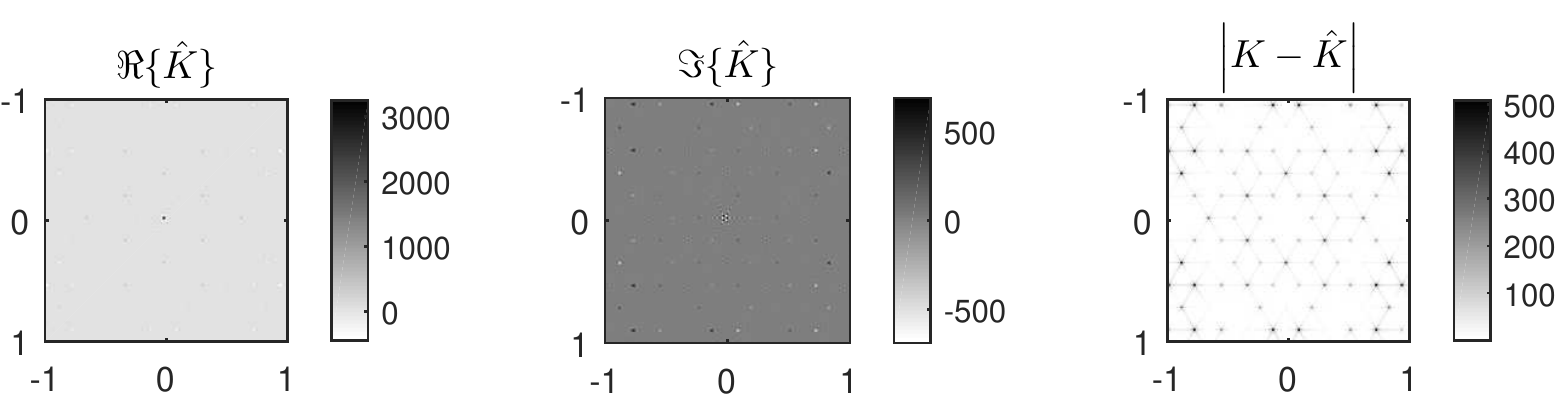}\tabularnewline
\hline 
{\scriptsize{}{}5}  &
\includegraphics[scale=0.45]{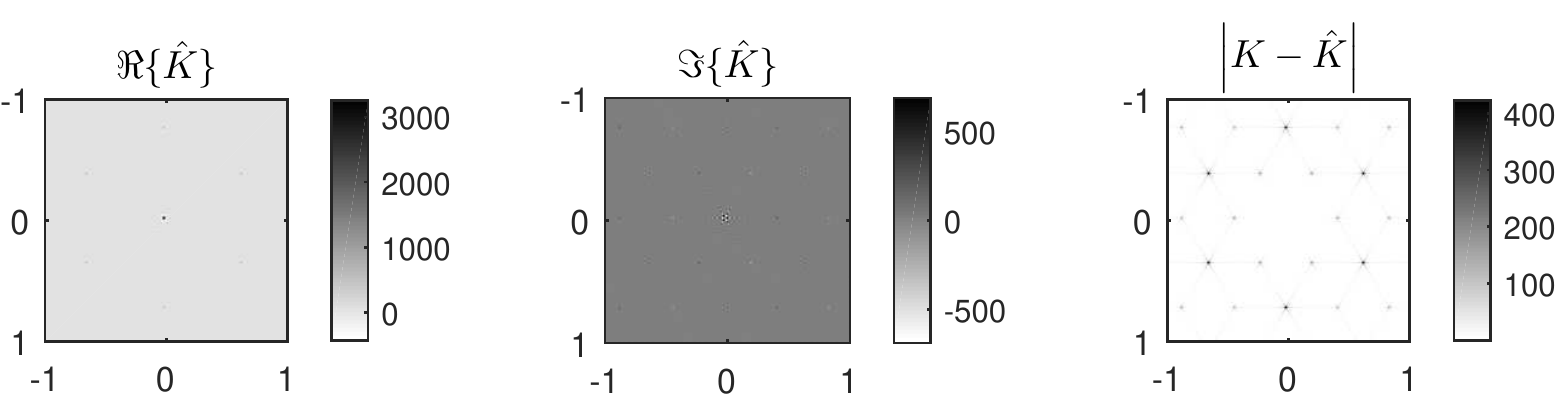}\tabularnewline
\hline 
{\scriptsize{}{}6}  &
\includegraphics[scale=0.45]{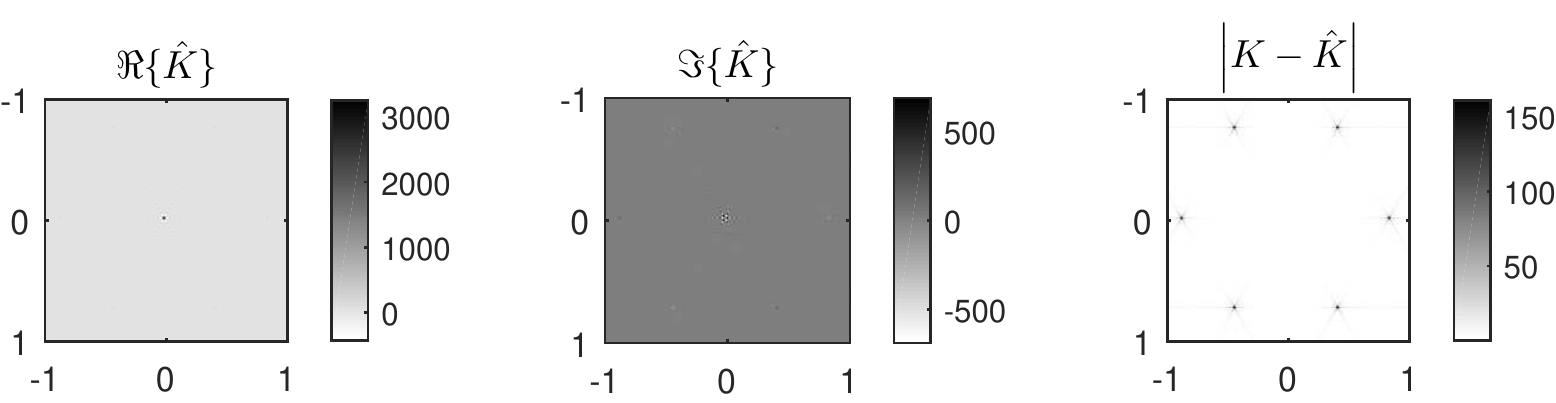}\tabularnewline
\hline 
{\scriptsize{}{}7}  &
\includegraphics[scale=0.45]{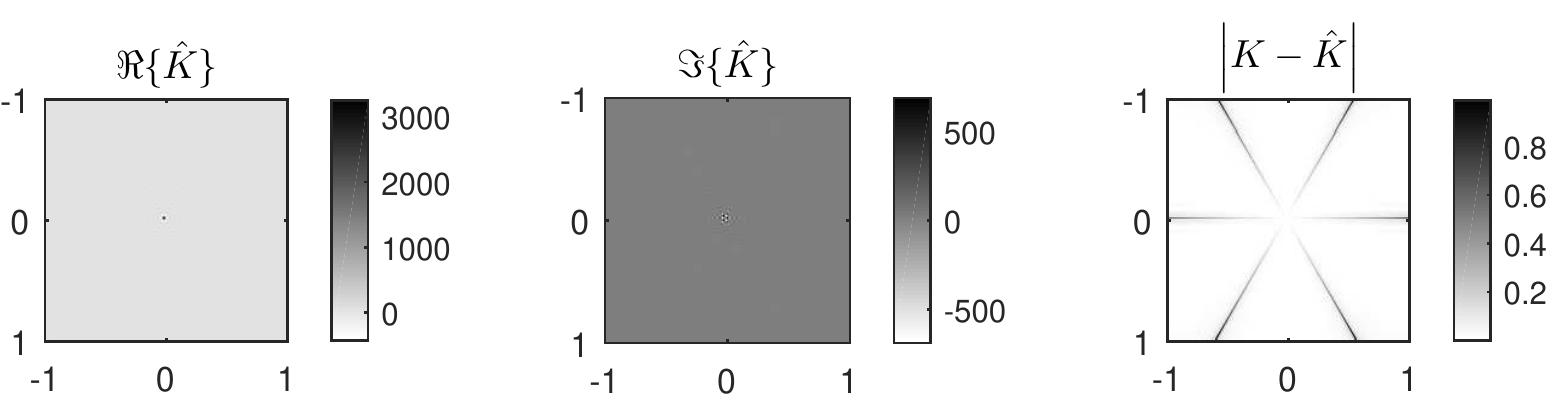}\tabularnewline
\hline 
\end{tabular}

\caption{Approximation of the kernel $K_{\triangle}\left(x,y\right)$ by $\tilde{K}_{\triangle,m}\left(2^{-m}x,2^{-m}y\right)$
using the RI quadratures in Figure \ref{fig:Nodes-for-equilateral}.
The horizontal axis is $x$ and the vertical axis is $y$. \label{fig:RS-quad-approx-K_triangle}}
\end{figure}

\subsection{Tetrahedron limited functions}

We say a function $f\left(x,y,z\right)$, $\left(x,y,z\right)\in\mathbb{R}^{3}$
is tetrahedron limited if its Fourier transform $\hat{f}\left(k_{x},k_{y},k_{z}\right)$
is supported within a tetrahedral region $T\subset\mathbb{R}^{3}$:
\begin{align}
f\left(x,y\right) & =\int_{T}\hat{f}\left(k_{x},k_{y},k_{z}\right)\mathrm{e}^{2\pi\mathrm{i}\left(k_{x}x+k_{y}y+k_{z}z\right)}dk_{z}dk_{y}dk_{x}\label{eq:fourier-representation-1}
\end{align}
Without loss of generality, let $T$ be parametrized by 
\begin{align}
T & =\left\{ \left(k_{x},k_{y},k_{z}\right)\,|\,0\le k_{z}\le h,\,0\le k_{y}\le\Delta pk_{z},\,\left|k_{x}\right|\le k_{y}s,\right\} 
\end{align}
for some $\Delta p,s,h\in\mathbb{R}^{+}$. Define $K_{\triangleleft}\left(x,y,z\right)$
to be 
\begin{align}
 & K_{\triangleleft}\left(x,y,z\right)\nonumber \\
 & =\int_{0}^{h}\int_{0}^{k_{z}\Delta p}\int_{-k_{y}s}^{k_{y}s}\mathrm{e}^{\mathrm{i}2\pi\left(k_{x}x+k_{y}y+k_{z}z\right)}dk_{x}dk_{y}dk_{z}\nonumber \\
 & =-\frac{h^{2}\Delta p}{2\pi x}\left\{ \begin{array}{l}
\frac{\mbox{expc}\left(\mathrm{i}2\pi h\left[z+\Delta p\left(y+sx\right)\right]\right)-\mbox{expc}\left(\mathrm{i}2\pi hz\right)}{2\pi h\Delta p\left(y+sx\right)}\\
-\frac{\mbox{expc}\left(\mathrm{i}2\pi h\left[z+\Delta p\left(y-sx\right)\right]\right)-\mbox{expc}\left(\mathrm{i}2\pi hz\right)}{2\pi h\Delta p\left(y-sx\right)}
\end{array}\right\} 
\end{align}
where $\mbox{expc}\left(x\right)=\left(\exp\left(x\right)-1\right)x^{-1}$,
implying $\mbox{expc}\left(\mathrm{i}x\right)=\mbox{sinc}\left(x\right)+\mathrm{i}\mbox{cosinc}\left(x\right)$,
with $\mbox{cosinc}\left(x\right)=\left(1-\cos\left(x\right)\right)x^{-1}$.
For an equilateral tetrahedron, choose $h=\sqrt{\frac{2}{3}}$, $\Delta p=\frac{\sqrt{3}}{6}h$,
$s=\sqrt{3}$ and add the resulting kernel with its $\frac{2}{3}\pi$
and $\frac{4}{3}\pi$ rotated versions around the $z$-axis (see Figure
\ref{fig:K-tetra}).

\begin{figure}
\center\includegraphics[scale=0.45]{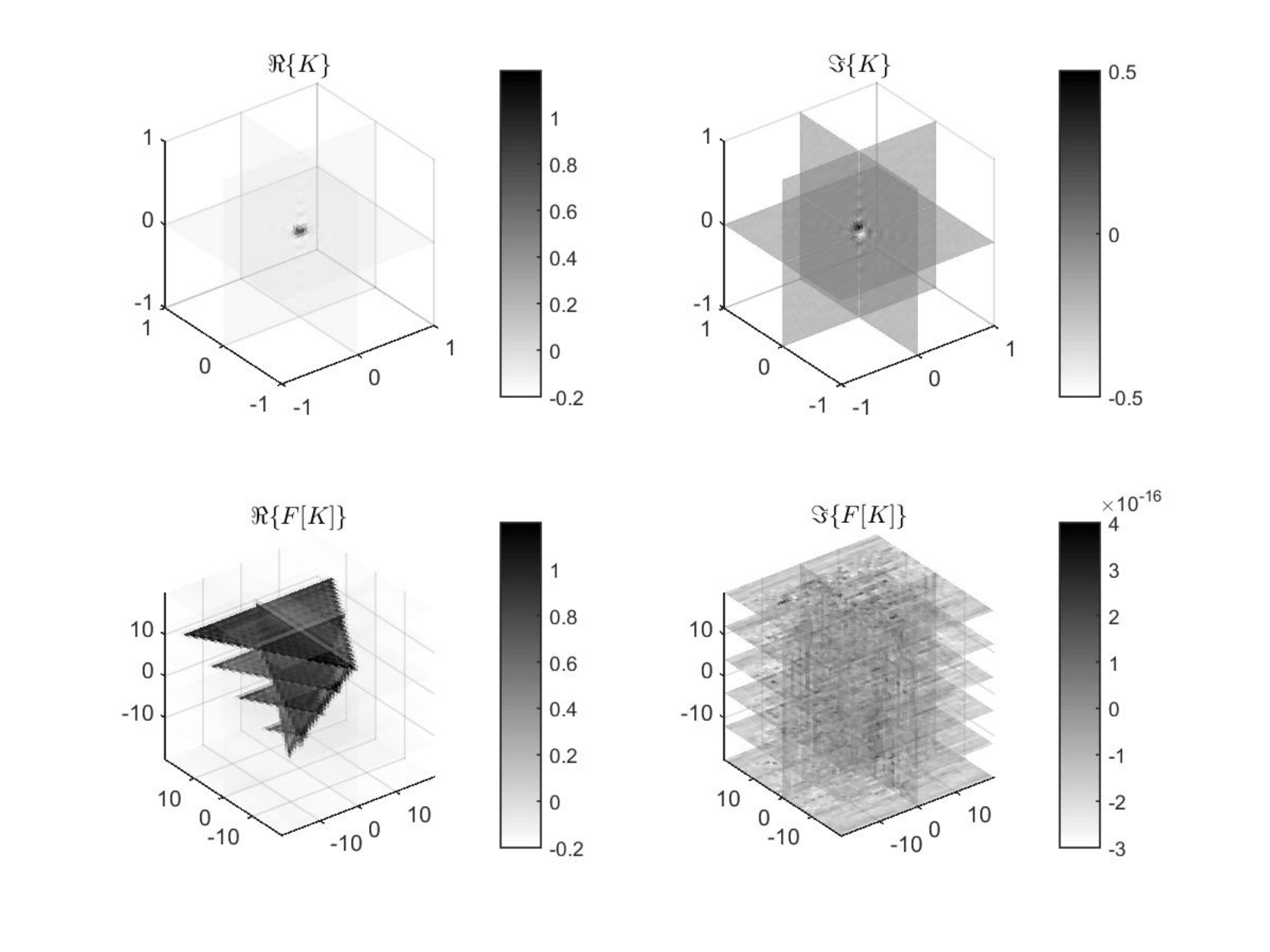}\caption{Real, imaginary parts of $K\left(x,y,z\right)=B^{-3}\sum_{n=0}^{2}K_{\triangleleft}\left(Bx_{n},By_{n},Bz_{n}\right)\exp\left(\mathrm{i}2\pi B\left[-\frac{1}{2}z+\left(\frac{1}{2}+\frac{\sqrt{3}}{3}y\right)\right]\right)$
and its Fourier transform for $h=\sqrt{\frac{2}{3}}$, $\Delta p=\frac{\sqrt{3}}{6}h$,
$s=\sqrt{3}$ and $B=20$. Here $\left[x_{n},y_{n},z_{n}\right]^{T}=R_{\left[0,0,1\right]}\left(\frac{2}{3}\pi n\right)\left[x,y,z\right]^{T}$
for $n=0,1,2$. \label{fig:K-tetra}}
\end{figure}

Similar to the triangle-limited case, construction of the nodes is
equivalent to discretization of the Fourier integral using cascaded
Gauss-Legendre quadratures that can accurately approximate the representation
kernel within a region of interest:

\begin{align}
 & K_{\triangleleft}\left(x,y,z\right)\nonumber \\
 & =\hspace{-0.1cm}\int_{0}^{h}\hspace{-0.1cm}\int_{0}^{k_{z}\Delta p}\hspace{-0.2cm}\int_{-k_{y}s}^{k_{y}s}\hspace{-0.2cm}\mathrm{e}^{\mathrm{i}2\pi\left(k_{x}x+k_{y}y+k_{z}z\right)}dk_{x}dk_{y}dk_{z}\nonumber \\
 & =-\frac{h\Delta ps}{4\pi^{2}}\partial_{y}\partial_{z}\hspace{-0.1cm}\int_{-1}^{1}\hspace{-0.1cm}\int_{0}^{1}\hspace{-0.1cm}\int_{0}^{1}\hspace{-0.2cm}\mathrm{e}^{\mathrm{i}2\pi h\left(\Delta p\left[s\,k_{x}x+y\right]k_{y}+z\right)k_{z}}dk_{z}dk_{y}dk_{x}\nonumber \\
 & =2h^{3}\Delta p^{2}s\hspace{-0.1cm}\sum_{m,n,l}\hspace{-0.1cm}\left\{ \hspace{-0.2cm}\begin{array}{l}
\alpha_{m}\beta_{m,n}\gamma_{m,n,l}k_{z}^{2}\left[m\right]k_{y}\left[m,n\right]\\
\times\mathrm{e}^{\mathrm{i}2\pi h\left(\Delta p\left[s\,\left(2k_{x}\left[m,n,l\right]-1\right)x+y\right]k_{y}\left[m,n\right]+z\right)k_{z}\left[m\right]}
\end{array}\hspace{-0.2cm}\right\} 
\end{align}
where $\left(\alpha_{m},k_{z}\left[m\right]\right)$, $\left(\beta_{m,n},k_{y}\left[m,n\right]\right)$
and $\left(\gamma_{m,n,l},k_{x}\left[m,n,l\right]\right)$ are quadratures
for approximating $\mbox{sinc}\left(Bx\right)$ as a sum of cosines
(see (\ref{eq:sinc_in_cos_approx})) for $B$ equal to $2\pi h\left(\Delta p\left(Y+sX\right)+Z\right)$
and $2\pi h\Delta p\,k_{z}\left[m\right]\left(Y+sX\right)$ and $4\pi h\Delta p\,s\,k_{z}\left[m\right]k_{y}\left[m,n\right]X$,
respectively, where $X=\max_{\left(x,y,z\right)\in S+S}\left|x\right|$,
$Y=\max_{\left(x,y,z\right)\in S+S}\left|y\right|$ and $Z=\max_{\left(x,y,z\right)\in S+S}\left|z\right|$.

Let us consider a unit tetrahedron, i.e. a tetrahedron with all sides
equal to one. In order to construct nodes that satisfy the symmetries
of the unit tetrahedron, first construct nodes for the sub-tetrahedron
with $s=\sqrt{3}$, $\Delta p=\sqrt{2}$, $h=\sqrt{24}^{-1}$ (see
Figure \ref{fig:sub-tetra}) and then use the symmetry group of the
regular tetrahedron. Namely, apply rotations $R_{\hat{v}_{4}}\left(\frac{2}{3}\pi n\right)R_{\hat{v}_{1}}\left(\frac{2}{3}\pi m\right)$,
for $n,m=0,1,2$, and $R_{\hat{v}_{4}}\left(\frac{2}{3}\pi n\right)R_{\hat{v}_{2}}\left(\frac{4}{3}\pi\right)$
to the quadrature of sub-tetrahedron. Here $\hat{v}=v/\left|v\right|$
is the unit vector pointing along vector $v$ with $v_{i}$, for $i=1,2,3,4$,
given by 
\begin{align}
v_{1} & =\left[-\frac{1}{2},-\frac{\sqrt{3}}{6},-\frac{1}{6}\sqrt{\frac{3}{2}}\right]\\
v_{2} & =\left[\frac{1}{2},-\frac{\sqrt{3}}{6},-\frac{1}{6}\sqrt{\frac{3}{2}}\right]\\
v_{3} & =\left[0,\frac{\sqrt{3}}{3},-\frac{1}{6}\sqrt{\frac{3}{2}}\right]\\
v_{4} & =\left[0,0,\sqrt{\frac{2}{3}}-\frac{1}{6}\sqrt{\frac{3}{2}}\right].
\end{align}
and 
\begin{align}
R_{u}\left(\theta\right) & =\left[\begin{array}{ccc}
r_{1}\left(\theta\right) & r_{123}\left(-\theta\right) & r_{123}\left(\theta\right)\\
r_{123}\left(\theta\right) & r_{2}\left(\theta\right) & r_{231}\left(-\theta\right)\\
r_{132-\left(\theta\right)} & r_{231}\left(\theta\right) & r_{3}\left(\theta\right)
\end{array}\right]
\end{align}
is the matrix for a rotation by an angle $\theta$ around the unit
vector $u=\left[u_{1},u_{2},u_{3}\right]$ with 
\begin{align}
r_{i}\left(\theta\right) & =\cos\theta+u{}_{i}^{2}\left(1-\cos\theta\right)\\
r_{ijk}\left(\theta\right) & =u_{i}u_{j}\left(1-\cos\theta\right)+u_{k}\sin\theta.
\end{align}
The quadrature generated for the regular tetrahedron using the discussed
steps is presented in Figure \ref{fig:isometries-of-tetra} .

\begin{figure}
\center

\includegraphics[scale=0.35]{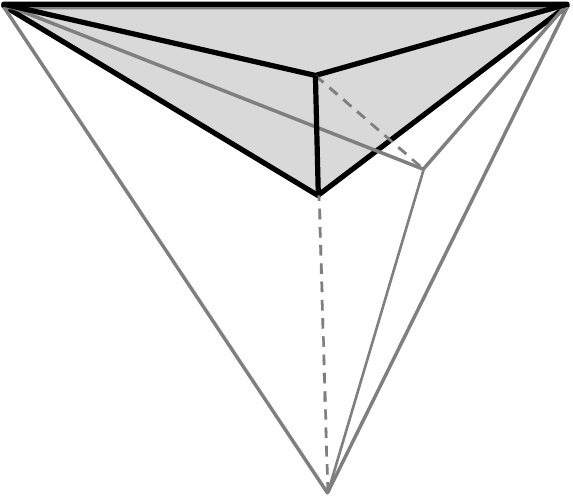}

\includegraphics[scale=0.25]{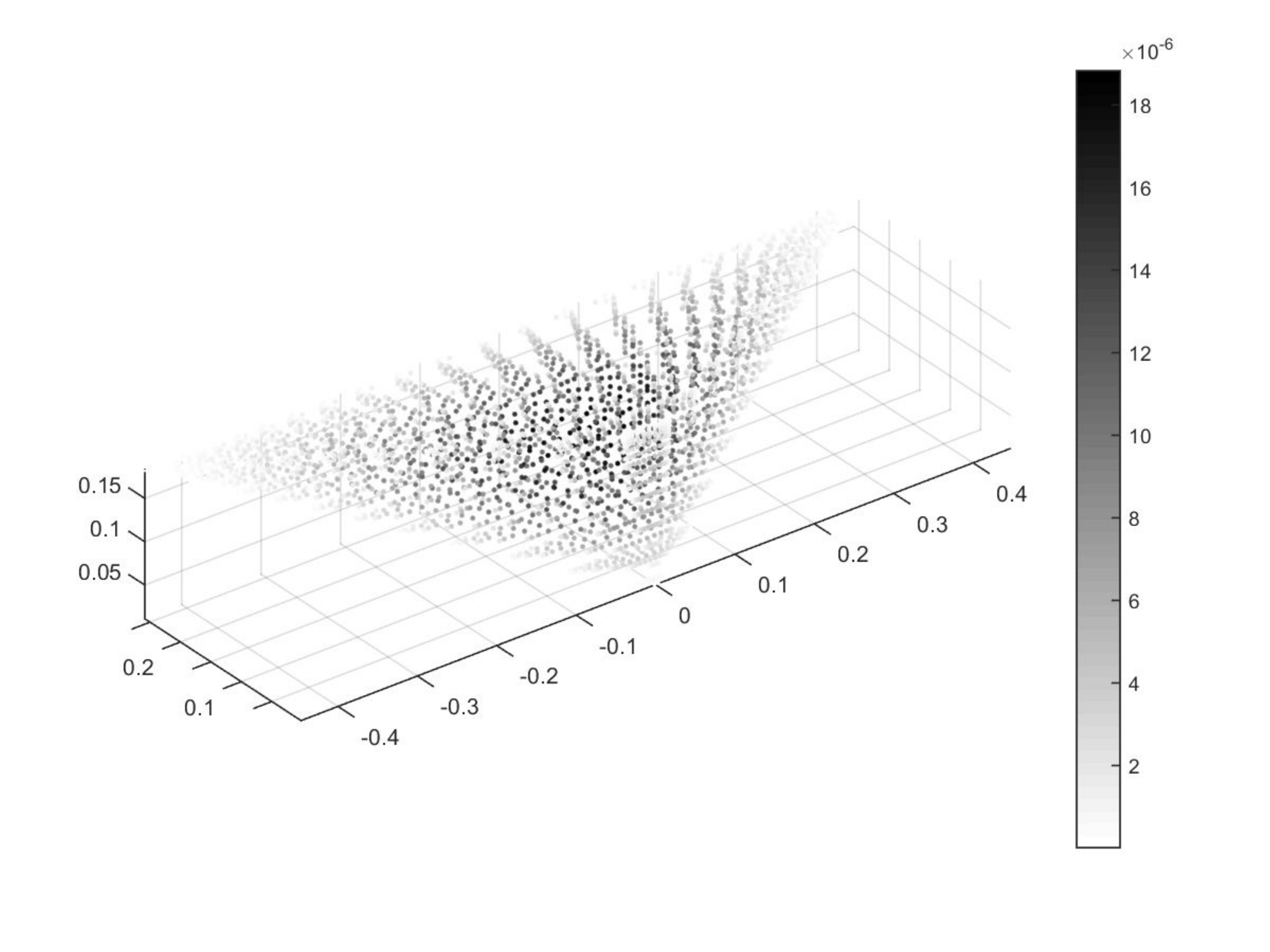}

\caption{Construction of the nodes that satisfies the isometries of the tetrahedron
can be obtained by construct nodes for the shaded tetrahedron (top)
and apply symmetry group of the tetrahedon (see Figure \ref{fig:isometries-of-tetra})
to the constructed nodes (bottom). \label{fig:sub-tetra}}
\end{figure}

\begin{figure}
\center

\includegraphics[scale=0.35]{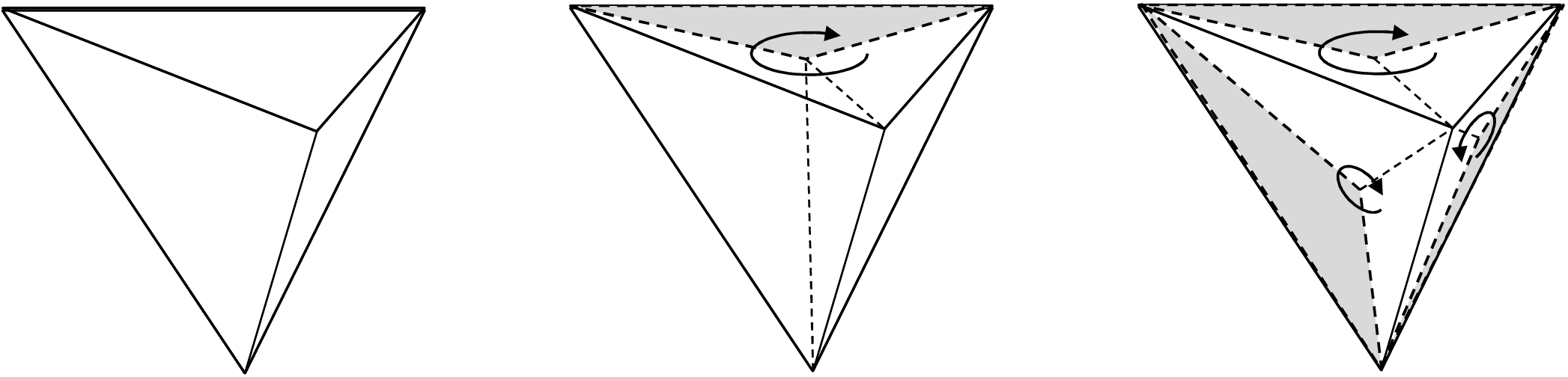}

\includegraphics[scale=0.45]{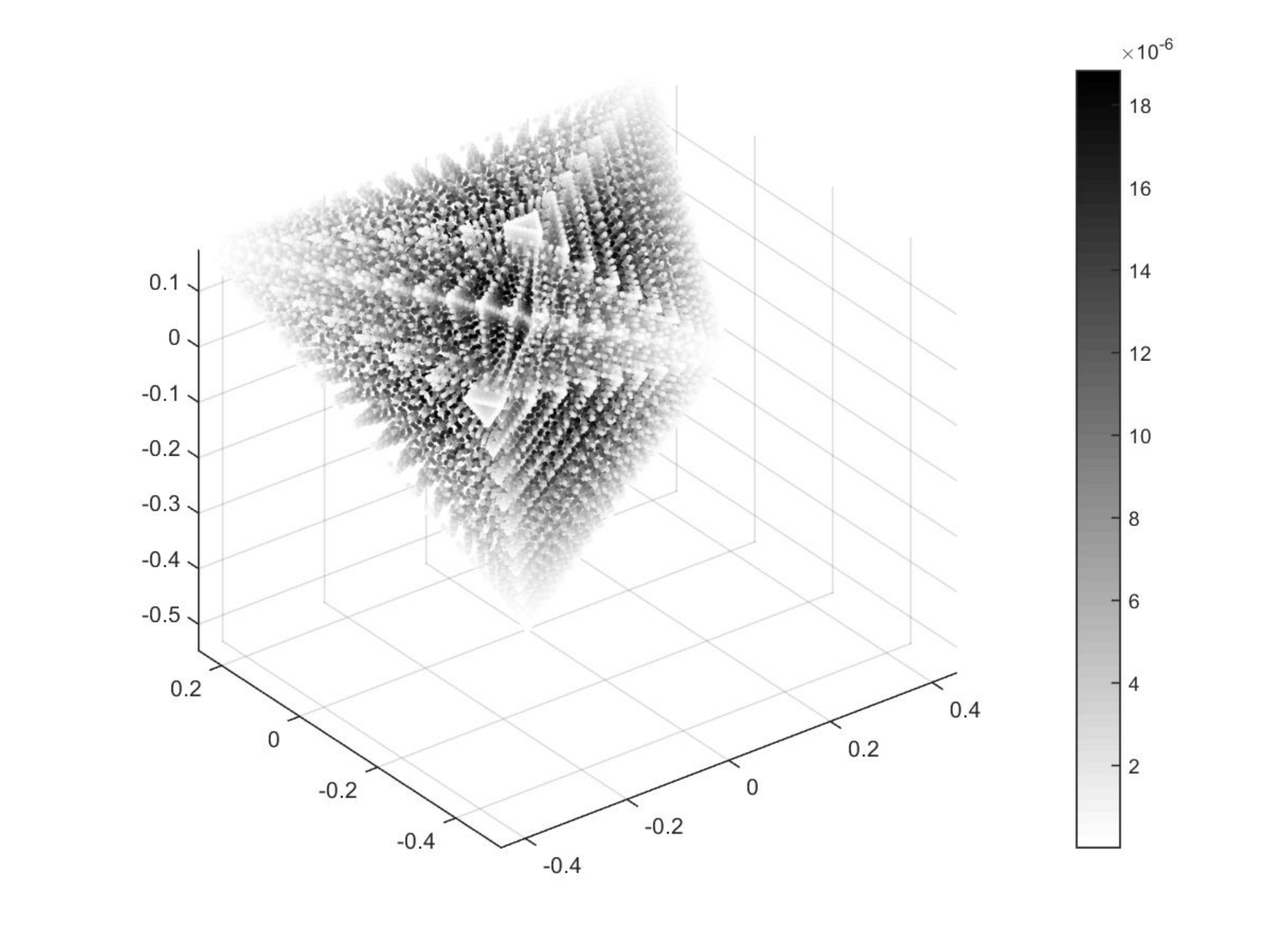}

\caption{A regular tetrahedron (top left) is symmetric under rotation with
respect to one of its faces (top middle). Its rotational symmetries
with respect to 3 out of 4 faces (top right). There are a total of
$4\times3\times2=24$ isometries of the tetrahedron which form the
symmetry group of the tetrahedron. Out of these $12$ of them preserves
orientation which are used to generate the corresponding quadrature
(bottom) from the quadrature of the sub-tetrahedron in Figure \ref{fig:sub-tetra}.
\label{fig:isometries-of-tetra}}
\end{figure}

\section{Cone-limited functions\label{sec:Cone-limited-functions}}

In seismic or electromagnetic signal processing the signal is modeled
through the wave equation. For an acoustic homogeneous medium with
wave speed $c=p^{-1}$, the wave equation provides a dispersion relationship
between the frequency $\omega$ and wave number \textbf{$\mathbf{k}$},
$\left|\mathbf{k}\right|=\omega p$. For a heterogeneous medium, the
dispersion relationship becomes an inequality $\left|\mathbf{k}\right|\le\omega p_{\max}$
where the maximum slowness $p_{\max}=c_{\min}^{-1}$ is the one over
the minimum speed $c_{\min}$ of the heterogeneous medium. Given the
maximum frequency $\omega_{0}$ of the recording system, the Fourier
transform of the measurement is supported inside the cone $C=\left\{ \left(\omega,\mathbf{k}\right)\in\mathbb{R}\times\mathbb{R}^{n}|\omega\in\left[-\omega_{0},\omega_{0}\right],\left|{\bf k}\right|\le\omega p_{\max}\right\} $
for $t\in\mathbb{R}$ and $\mathbf{x}\in\mathbb{R}^{n}$, which is
referred to as the signal cone. Temporal and spatial Fourier transform
of video images also have their Fourier transforms supported effectively
in a similar cone.

We say a function $f\left(t,\mathbf{x}\right)$ is cone-limited, $C$-limited
for short, if its Fourier transform $\hat{f}\left(\omega,\mathbf{k}\right)$
is supported within the cone $C$. $C$-limited functions are invariant
under convolution with the kernel $K\left(t,\mathbf{x}\right)$, whose
Fourier transform, $\hat{K}\left(\omega,\mathbf{k}\right)$, is equal
to one within $C$: 
\begin{align}
K\left(t,\mathbf{x}\right) & =\int_{C}\mathrm{e}^{\mathrm{i}2\pi\left(\omega t-{\bf k}\cdot{\bf x}\right)}d\omega\,d{\bf k}.
\end{align}
For $n$ odd, $K\left(t,\mathbf{x}\right)$ can be represented in
terms of elementary functions. For example, for $n=1$ and $n=3$,
we have 
\begin{align}
K\left(t,x\right) & =\frac{\omega_{0}}{\pi x}\left(\begin{array}{l}
\mathrm{cosinc}\left(2\pi\omega_{0}\left[t+p_{\max}x\right]\right)\\
-\mathrm{cosinc}\left(2\pi\omega_{0}\left[t-p_{\max}x\right]\right)
\end{array}\right)
\end{align}
and 
\begin{align}
K\left(t,\mathbf{x}\right) & =-\frac{\omega_{0}}{\pi r}\partial_{r}\left\{ \frac{1}{r}\left[\begin{array}{l}
\mbox{cosinc}\left(2\pi\omega_{0}\left[t+p_{\max}r\right]\right)\\
-\mbox{cosinc}\left(2\pi\omega_{0}\left[t-p_{\max}r\right]\right)
\end{array}\right]\right\} ,
\end{align}
respectively. On the other hand, for $n$ even, $K\left(t,\mathbf{x}\right)$
is a multivariate special function. For example, for $n=2$ , we have

\begin{align}
K\left(t,\mathbf{x}\right) & =\int_{0}^{2\pi}\int_{0}^{p_{\max}}\int_{-\omega_{0}}^{\omega_{0}}\mathrm{e}^{\mathrm{i}2\pi\omega\left(t-pr\cos\theta\right)}\omega^{2}p\,d\omega\,dp\,d\theta\nonumber \\
 & =\frac{2\omega_{0}^{2}\,p_{\max}}{\pi r}\int_{0}^{1}J_{1}(2\pi\omega r\omega_{0}p_{\max})\cos(2\pi\omega\omega_{0}t)\,\omega\,d\omega\label{eq:integral_J1*cos*w}
\end{align}
where $\left|\mathbf{x}\right|=r$ and $J_{n}\left(t\right)$ is the
$n^{th}$ order Bessel function of the first kind. Becuase all the
cases for $n$ even, requires evaluation of an integral of the form
(\ref{eq:integral_J1*cos*w}), for the ease of our discussion we will
focus on the case $n=2$.

Following \cite{YF2014}, $J_{1}\left(x\right)$ can be approximate
by $J_{1}(x)\approx\sum_{m=1}^{M}\alpha_{m}\mbox{cosinc}(\gamma_{m}x)$,
where $\left(\alpha_{m},\gamma_{m}\right)$ satisfies a moment problem.
(see Table \ref{tab:Example-of-moment-problems} in Appendix \ref{sec:Generalization-of-Pad=00003D0000E9})
. Consequently, we have 
\begin{align}
K\left(t,\mathbf{x}\right) & \approx\tilde{K}\left(t,\mathbf{x}\right)\nonumber \\
 & =\frac{\omega_{0}}{\pi^{2}}\sum_{m=1}^{M}\frac{\alpha_{m}}{\gamma_{m}}\left(\begin{array}{l}
\mbox{sinc}(2\pi\omega_{0}t)\\
-\frac{1}{2}\left[\begin{array}{l}
\mathrm{sinc}\left(2\pi\omega_{0}\left[\gamma_{m}p_{\max}r-t\right]\right)\\
+\mathrm{sinc}\left(2\pi\omega_{0}\left[\gamma_{m}p_{\max}r+t\right]\right)
\end{array}\right]
\end{array}\right)
\end{align}
The least square error is given by 
\begin{align}
\int\left|K(t,\mathbf{x})-\tilde{K}(t,\mathbf{x})\right|^{2}dt\,d\mathbf{x} & =\frac{4\omega_{0}^{3}\,p_{\max}^{2}}{\pi^{2}3}{\scriptscriptstyle \left(\frac{1}{2}-\sum_{m=1}^{M}\alpha_{m}\gamma_{m}+\sum_{m,m'=1}^{M}\alpha_{m}w_{m,m'}\alpha_{m'}\right)},
\end{align}
where 
\begin{align}
w_{m,m'}=\int_{\mathbb{R}^{+}}\frac{\sin^{2}(\gamma_{m}r/2)}{\gamma_{m}r/2}\,\frac{\sin^{2}(\gamma_{m'}r/2)}{\gamma_{m'}r/2}\frac{dr}{r},
\end{align}
which can be explicitly computed using integration by parts and the
identities 3.827-3.828 on pages 462-463 of \cite{gradshteyn2007}.

Because

\begin{align}
\tilde{K}_{p}\left(t,\mathbf{x}\right) & =\frac{1}{r^{2}}\left\{ \mathrm{sinc}\left(2\pi\omega_{0}t\right)-\frac{1}{2}\hspace{-0.1cm}\left[\hspace{-0.1cm}\begin{array}{l}
\mathrm{sinc}\left(2\pi\omega_{0}\left[pr-t\right]\right)\\
+\mathrm{sinc}\left(2\pi\omega_{0}\left[pr+t\right]\right)
\end{array}\hspace{-0.1cm}\right]\right\} 
\end{align}
has its Fourier transform 
\begin{multline}
\int\tilde{K}_{p}\left(t,\mathbf{x}\right)\mathrm{e}^{\mathrm{i}2\pi\left(\omega t-\mathbf{k}\cdot\mathbf{x}\right)}dt\,d\mathbf{x}\\
=2\pi^{2}\left|\omega p\right|\chi_{\left[-1,1\right]}\left(\omega\omega_{0}^{-1}\right)\mbox{arcsinh}\left(\frac{\sqrt{\left(\omega p\right)^{2}-\left|\mathbf{k}\right|^{2}}}{\left|\mathbf{k}\right|}\right),\\
0<\left|\mathbf{k}\right|\le\left|\omega p\right|,
\end{multline}
supported over 
\begin{align}
C_{p} & =\left\{ \left(\omega,\mathbf{k}\right)\in\mathbb{R}\times\mathbb{R}^{2}|\omega\in\left[-\omega_{0},\omega_{0}\right],\left|{\bf k}\right|\le\omega p\right\} ,
\end{align}
$\tilde{K}\left(t,\mathbf{x}\right)$ is C-limited within $\tilde{C}=\left\{ \left(\omega,\mathbf{k}\right)\in\mathbb{R}\times\mathbb{R}^{2}|\omega\in\left[-\omega_{0},\omega_{0}\right],\left|{\bf k}\right|\le\omega p_{\max}\,\max_{m}\left\{ \gamma_{m}\right\} \right\} $.
If $\max_{m}\left\{ \gamma_{m}\right\} \approx1$, then the cone-limit
$C$ of $K\left(t,\mathbf{x}\right)$ is approximated by the cone-limit
$\tilde{C}$ of $\tilde{K}\left(t,\mathbf{x}\right)$ which is the
case in practice. For $\omega_{0}=50$ and $p_{\max}=1$, we present
$\tilde{K}\left(t,\mathbf{x}\right)$ and its Fourier transform in
Figure \ref{fig:K-cone}.

For $n=2$, discretizaton of the integral representation of $K\left(t,\mathbf{x}\right)$
can be obtained by 
\begin{align}
K\left(t,\mathbf{x}\right) & =2p_{\max}^{2}\omega_{0}^{3}\int_{-1}^{1}\int_{0}^{1}\int_{-1}^{1}\mathrm{e}^{\mathrm{i}2\pi\phi\left(\omega,p,\tau;t,x,y\right)}\frac{\omega^{2}p}{\sqrt{1-\tau^{2}}}d\omega\,dp\,d\tau\nonumber \\
 & =2p_{\max}^{2}\omega_{0}^{3}\sum_{m,n,l}a_{m,n,l}\mathrm{e}^{\mathrm{i}2\pi\phi\left[m,n,l\right]\left(t,x,y\right)}
\end{align}
where 
\begin{align}
{\scriptstyle a}_{m,n,l} & {\scriptstyle =\alpha_{m}\beta_{m,n}\gamma_{m,n,l}\frac{\omega\left[m\right]^{2}p\left[m,n\right]}{\sqrt{1-\left(\tau\left[m,n,l\right]\right)^{2}}}}\\
{\scriptstyle \phi\left(\omega,p,\tau;t,x,y\right)} & {\scriptstyle =\omega_{0}\omega\left(t-p_{\max}p\left[x\tau+y\sqrt{1-\tau^{2}}\right]\right)}\\
{\scriptstyle \phi\left[m,n,l\right]\left(t,x,y\right)} & {\scriptstyle =\omega_{0}\omega\left[m\right]\left(t-p_{\max}p\left[m,n\right]\left[x\tau\left[m,n,l\right]+y\sqrt{1-\left(\tau\left[m,n,l\right]\right)^{2}}\right]\right)}
\end{align}
with $\left(\alpha_{m},\omega\left[m\right]\right)$ and $\left(\beta_{m,n},p\left[m,n\right]\right)$
are quadratures for approximating $\mbox{sinc}\left(Bx\right)$ as
a sum of cosines (see (\ref{eq:sinc_in_cos_approx})) for $B$ equal
to $4\pi\omega_{0}\left(T+p_{\max}R\right)$ and $2\pi\omega_{0}\omega\left[m\right]p_{\max}R$,
respectively, and $\left(\gamma_{m,n,l},\tau\left[m,n,l\right]\right)$
is the quadrature for approximating $J_{0}\left(Bx\right)$ as a sum
of cosines for $B$ equal to $\omega_{0}\omega\left[m\right]p_{\max}p\left[m,n\right]R$,
where $R=\max_{\left(t,x,y\right)\in S+S}\sqrt{x^{2}+y^{2}}$, and
$T=\max_{\left(t,x,y\right)\in S+S}\left|t\right|$, for some region
of interest $S\subset\mathbb{R}\times\mathbb{R}^{2}$. While $\left(\alpha_{m},\omega\left[m\right]\right)$
and $\left(\beta_{m,n},p\left[m,n\right]\right)$ are equivalent to
Gauss-Legendre quadrature, and computation of $\left(\gamma_{m,n,l},\tau\left[m,n,l\right]\right)$
requires solving the following moment problem related to the approximation
$J_{0}(x)\approx\sum_{m=1}^{M}\alpha_{m}\cos(\gamma_{m}x)$, which
is equivalent to finding the Clenshaw-Curtis quadrature (see Table
\ref{tab:Example-of-moment-problems} in Appendix \ref{sec:Generalization-of-Pad=00003D0000E9}).

\begin{figure}
\center\includegraphics[scale=0.45]{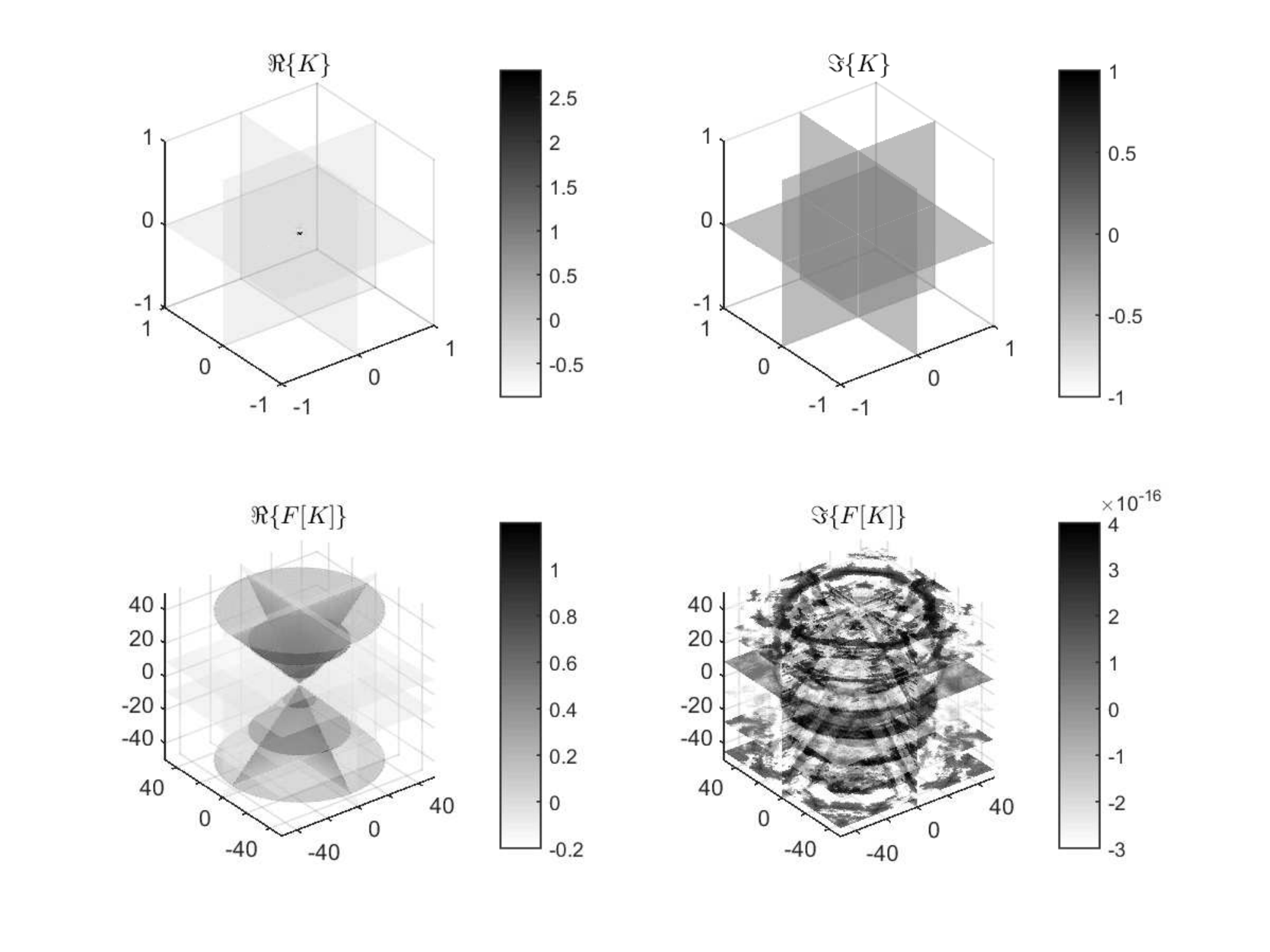}\caption{Real, imaginary parts of $\tilde{K}\left(t,\mathbf{x}\right)$ and
its Fourier transform for $\omega_{0}=50$ and $p_{\max}=1$. \label{fig:K-cone}}
\end{figure}

\begin{figure}
\center

\includegraphics[scale=0.45]{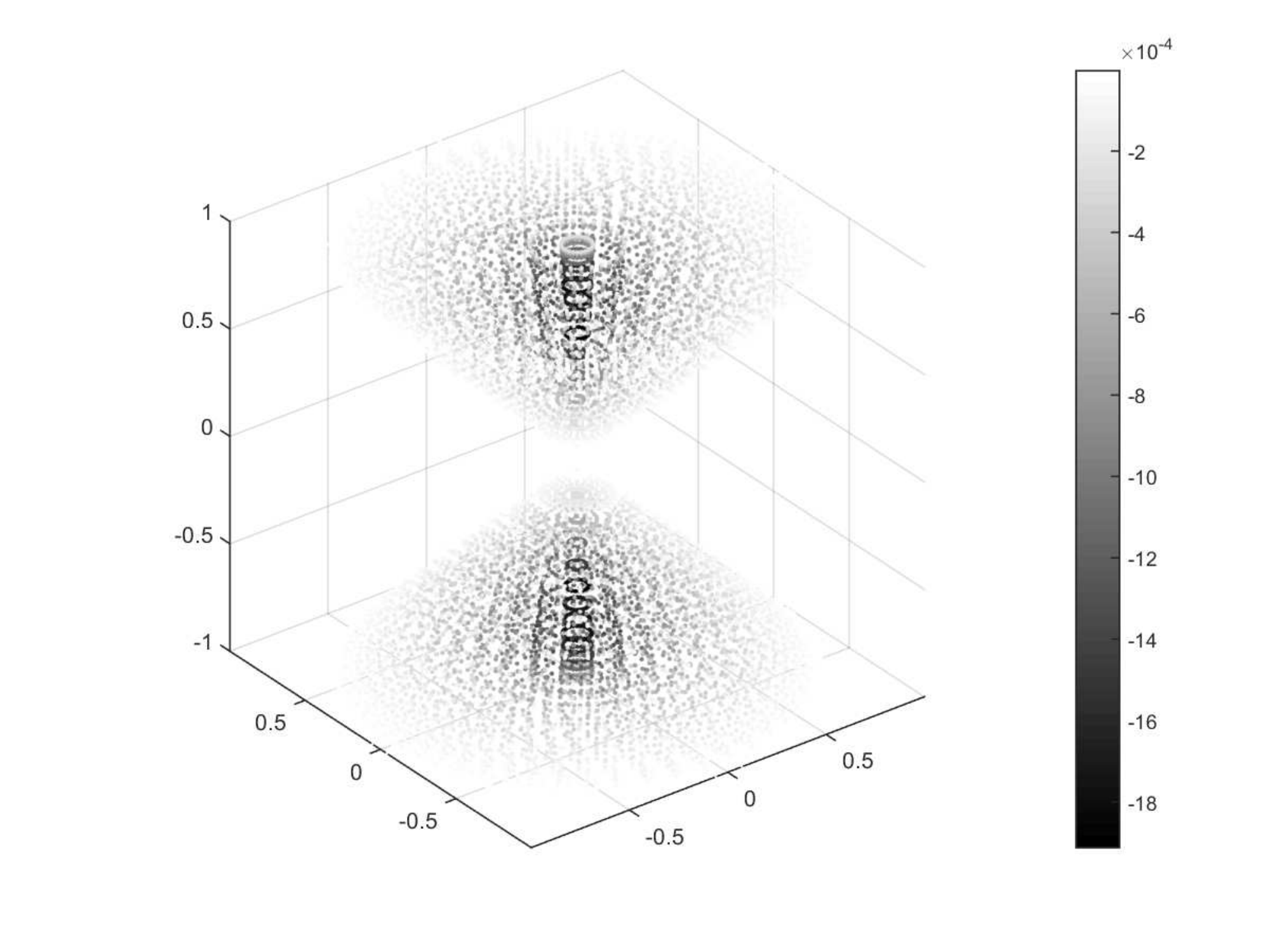}

\caption{Quadrature $\left(a_{mnl},\omega_{0}\left(\omega\left[m\right],p_{\max}\mathbf{p}\left[m,n,l\right]\right)\right)$
for $C$-limited functions for $\omega_{0}=1$ and $p_{\max}=1$.
\label{fig:quad_cone}}
\end{figure}

Note that quadrature for $B$-limited functions, whose Fourier transforms
are supported within a ball $B=\left\{ \mathbf{k}\in\mathbb{R}^{n}|\left|{\bf k}\right|\le k_{\max}\right\} $,
can be generated in a similar fashion (see Figure \ref{fig:quad_ball_3D}).

\begin{figure}
\center

\includegraphics[scale=0.45]{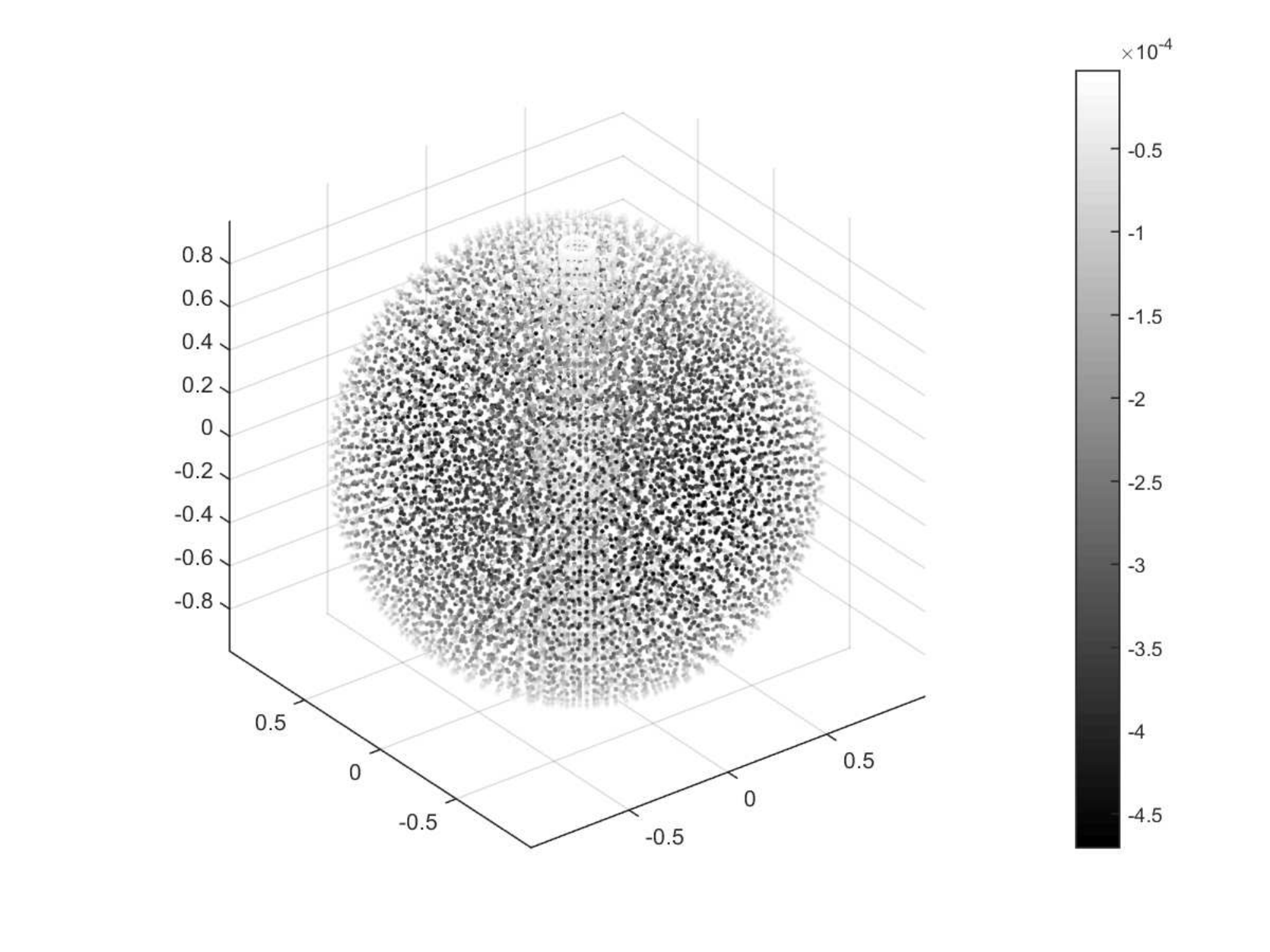}

\caption{Quadrature $\left(a_{mnl},k_{\max}\mathbf{k}_{m,n,l}\right)$ for
$C$-limited functions for $k_{\max}=1$. \label{fig:quad_ball_3D}}
\end{figure}

\end{document}